\documentclass[%
 reprint,
superscriptaddress,
footinbib,
bibnotes,
 amsmath,amssymb,
 longbibliography,
aps,
pra,
]{revtex4-2}

\usepackage{graphicx}
\usepackage{dcolumn}
\usepackage{braket}
\usepackage{textcomp, gensymb}
\usepackage{bm}
\usepackage{amsmath}
\usepackage{float}
\usepackage{mathtools}
\usepackage{color}
\usepackage[dvipsnames]{xcolor}
\usepackage[utf8]{inputenc}
\usepackage[T1]{fontenc}
\usepackage{csquotes}
\usepackage{thmtools}
\usepackage{qcircuit}

\newtheorem{lemma}{Lemma}

\makeatletter
\DeclareRobustCommand{\qed}{%
  \ifmmode 
  \else \leavevmode\unskip\penalty9999 \hbox{}\nobreak\hfill
  \fi
  \quad\hbox{\qedsymbol}}
\newcommand{\openbox}{\leavevmode
  \hbox to.77778em{%
  \hfil\vrule
  \vbox to.675em{\hrule width.6em\vfil\hrule}%
  \vrule\hfil}}
\newcommand{\qedsymbol}{\openbox}
\newenvironment{proof}[1][\proofname]{\par
  \normalfont
  \topsep6\p@\@plus6\p@ \trivlist
  \item[\hskip\labelsep\itshape
    #1.]\ignorespaces
}{%
  \qed\endtrivlist
}
\newcommand{\proofname}{Proof}
\makeatother

\renewcommand{\thesubsection}{\thesection.\Roman{subsection}}

\newcommand{\Jeff}{\tilde{J}}

\newcommand{\approptoinn}[2]{\mathrel{\vcenter{
  \offinterlineskip\halign{\hfil$##$\cr
    #1\propto\cr\noalign{\kern2pt}#1\sim\cr\noalign{\kern-2pt}}}}}

\newcommand{\bigO}[1]{\mathcal O\left(#1\right)}

\usepackage{url}

\usepackage{breakurl}
\usepackage[breaklinks]{hyperref}
\hypersetup{pdfborder=0 0 0}

\usepackage[normalem]{ulem}

\usepackage{physics}
\usepackage[version=4]{mhchem}
\usepackage{dsfont}

\usepackage{natbib}
\usepackage{empheq}
\usepackage{cases}
\usepackage[all]{hypcap}
\usepackage{siunitx}
\usepackage[capitalize]{cleveref}
\usepackage{placeins}
\usepackage{listings}

\usepackage{tabularx}
\usepackage{nicematrix}
\usepackage{stmaryrd}
\usepackage{footmisc}

\crefname{section}{Sec.}{Secs.}

\definecolor{HitachiRed}{HTML}{FF0026}
\definecolor{spincolor}{HTML}{2E6E8E}
\definecolor{potential_color}{HTML}{1E9C89}
\colorlet{wellcolor}{potential_color!25!spincolor}
\definecolor{photon}{HTML}{FE6100}
\colorlet{inscolor}{gray!50!white}
\colorlet{sicolor}{inscolor!50!white}

\lstdefinestyle{inline}{basicstyle=\tt}

\DeclareSIUnit{\Hartree}{Ha}

\definecolor{myorange}{HTML}{FE6100}
\definecolor{myblue}{HTML}{64A5FF}
\definecolor{mypink}{HTML}{FF1972}
\definecolor{mygreen}{HTML}{62D09D}
\definecolor{myviolet}{HTML}{4F428E}
\definecolor{myyellow}{HTML}{FFBF00}

\renewcommand{\thesubsection}{\Roman{section} \Alph{subsection}}

\makeatletter
\def\p@subsection{}
\makeatother
\makeatletter
\def\p@subsubsection{}
\makeatother

\allowdisplaybreaks

\begin{document}

\title{
Minimal evolution times for fast, pulse-based state preparation in silicon spin qubits
}

\author{Christopher K. Long}
\email{ckl45@cam.ac.uk}
\affiliation{ Hitachi  Cambridge  Laboratory,  J.  J.  Thomson  Ave.,  
Cambridge,  CB3  0HE,  United  Kingdom
}
\affiliation{ Cavendish Laboratory,  Department of Physics,  University  of  
Cambridge,  Cambridge,  CB3  0HE,  United  Kingdom
}
\author{Nicholas J. Mayhall}
\affiliation{Department of Chemistry, Virginia Tech, Blacksburg, VA 24061, USA}
\author{Sophia E. Economou}
\affiliation{Department of Physics, Virginia Tech, Blacksburg, VA 24061, USA}
\author{Edwin Barnes}
\affiliation{Department of Physics, Virginia Tech, Blacksburg, VA 24061, USA}
\author{Crispin H. W. Barnes}
\affiliation{ Cavendish Laboratory,  Department of Physics,  University  of  
Cambridge,  Cambridge,  CB3  0HE,  United  Kingdom
}
\author{Frederico Martins}
\affiliation{ Hitachi  Cambridge  Laboratory,  J.  J.  Thomson  Ave.,  
Cambridge,  CB3  0HE,  United  Kingdom
}
\author{David R. M. Arvidsson-Shukur}
\affiliation{ Hitachi  Cambridge  Laboratory,  J.  J.  Thomson  Ave.,  
Cambridge,  CB3  0HE,  United  Kingdom
}
\author{Normann Mertig}
\affiliation{ Hitachi  Cambridge  Laboratory,  J.  J.  Thomson  Ave.,  
Cambridge,  CB3  0HE,  United  Kingdom
}

\date{\today}

\begin{abstract}
Standing as one of the most significant barriers to reaching quantum advantage, state-preparation fidelities on noisy intermediate-scale quantum processors suffer from quantum-gate errors, which accumulate over time. A potential remedy is pulse-based state preparation.
We numerically investigate the minimal evolution times (METs) attainable by optimizing (microwave and exchange) pulses on silicon hardware.
We investigate two state preparation tasks.
First, we consider the preparation of molecular ground states and find the METs for \ce{H2}, \ce{HeH+}, and \ce{LiH} to be \SI{2.4}{\ns}, \SI{4.4}{ns}, and \SI{27.2}{ns}, respectively.
Second, we consider transitions between arbitrary states and find the METs for transitions between arbitrary four-qubit states to be below \SI{50}{\ns}.
For comparison, connecting arbitrary two-qubit states via one- and two-qubit gates on the same silicon processor requires approximately \SI{200}{\ns}.
This comparison indicates that pulse-based state preparation is likely to utilize the coherence times of silicon hardware more efficiently than gate-based state preparation.
Finally, we quantify the effect of silicon device parameters on the MET.
We show that increasing the maximal exchange amplitude from \SI{10}{\MHz} to \SI{1}{\GHz} accelerates the METs, e.g., for \ce{H2} from \SI{84.3}{\ns} to \SI{2.4}{\ns}.
This demonstrates the importance of fast exchange.
We also show that increasing the maximal amplitude of the microwave drive from \SI{884}{\kHz} to \SI{56.6}{\MHz} shortens state transitions, e.g., for two-qubit states from \SI{1000}{\ns} to \SI{25}{\ns}.
Our results bound both the state-preparation times for general quantum algorithms and the execution times of variational quantum algorithms with silicon spin qubits.
\end{abstract}

\maketitle

\section{Introduction}

Quantum simulation is one of the most important applications of quantum computing. However, running quantum-simulation algorithms on fault-tolerant, error-corrected quantum computers requires significant hardware advances and is many years away. Instead, the near-term realization of quantum simulation on quantum processors in the pre-fault-tolerant era is an intriguing alternative. To that end, the dominant approach is based on variational quantum algorithms (VQAs).

In the standard VQA approach \cite{Cerezo2021,TILLY20221}, state preparation is carried out by a parameterized unitary. This unitary is transpiled into fixed two-qubit and parameterized single-qubit gates not necessarily native or natural to the hardware. These gates are then compiled into \textit{pulses}: time-dependent terms of the native Hamiltonian. The parameters of the quantum gates are then optimized iteratively, such that a cost function of the prepared state is lowered according to the variational principle of quantum mechanics. We refer to this standard VQA model as the \textit{gate-based} VQA model. Despite the enormous amount of research in this field, quantum advantage with gate-based VQAs has not been demonstrated. This is largely due to the short coherence times of current quantum hardware; \textit{gate-based} VQAs are too \textit{noise-sensitive} \cite{dalton2022variational, long2023layering, yanakiev2023dynamicadaptqaoa}.

One strategy to improve the noise sensitivity of VQAs is transforming them into \textit{pulse-based} algorithms \cite{magannPulsesCircuitsBack2021a,lloyd2020quantum,lloyd2021hamiltonian,asthanaMinimizingStatePreparation2022,sherbert2024parameterization,meiteiGatefreeStatePreparation2021a,eggerStudyPulsebasedVariational2023,ibrahimEvaluationParameterizedQuantum2022,meiromPANSATZPulsebasedAnsatz2023,yangOptimizingVariationalQuantum2017}. Pulse-based quantum algorithms allow for more fine-grained control of the quantum system and forgo the overhead induced by transpilation and compilation, thus, reducing execution times. This is demonstrated numerically by the recent pulse-based VQA called ``ctrl-VQE'' on superconducting qubits \cite{meiteiGatefreeStatePreparation2021a}. In pulse-based VQAs, the pulses for state preparation are parameterized, and the objective function (same as in gate-based VQAs) is optimized by varying the parameters. Given that all gate-based circuits are ultimately executed as pulses, pulse-based VQAs generalize all their gate-based counterparts. As a result, one can identify fundamental limits and draw conclusions about general VQAs by studying pulse-based state preparation. A prominent example of this is the minimum evolution time (MET), which is the fundamental limit on how long it takes for a system to evolve from one state to another \cite{asthanaMinimizingStatePreparation2022}. In our case, this will be especially useful when the initial state is some classical approximation to the solution (e.g., the Hartree-Fock state) and the final state is the solution to the simulation problem (e.g., a Hamiltonian's ground state). The concept of MET is also known in the optimal-control literature and related to `quantum speed limits' \cite{Deffner_2017}.

\begin{table*}[t]
  \begin{tabularx}{\textwidth}{lXlXrXcXc}
    \toprule
    Name                   && Symbol               && Value or Constraints && SiMOS && Si/SiGe \\\colrule
    Zeeman splitting       && $B$                  && \SI{28}{\GHz}\footnotemark[1]  && \cite{Yoneda2021} && -- \\
    Zeeman-splitting Inhomogeneity && $\Delta B_i$         && Equally spaced in range \qtyrange{-30}{30}{\MHz}\footnotemark[1] && \cite{Yoneda2021,Huang2024} && \cite{Philips2022} \\
    Drive strengths                && $g\left(t\right)$ && $\left|I_i\left(t\right)\right|, \left|Q_i\left(t\right)\right|\le\SI{20}{\MHz}\times\frac{1}{2\sqrt{2}}$\footnotemark[2] && \cite{doi:10.1126/sciadv.1600694} && \cite{Philips2022} \\
    Exchange coupling              && $J_i\left(t\right)$  && $\SI{0}{\GHz}\le J_i\left(t\right)\le\SI{1}{\GHz}$\footnotemark[3] && \cite{t_values} && \cite{Weinstein2023Mar,Simmons2009}\\
    Microwave frequencies          && $\omega_i$           && $\frac{\omega_i}{2\pi}$ in the range \qtyrange{27}{29}{\GHz} && -- && -- \\
    \botrule
  \end{tabularx}
  \caption{\textbf{Model parameters, \cref{eq:DeviceModel}.} Values are in units where $h=1$. For an explanation of the footnotes, see \cref{footnotes}.}
  \label{table: parameter values}
\end{table*}

The performance of pulse-based state preparation can depend dramatically on the type of quantum processor. This is because each hardware platform has distinct characteristics, connectivities, and control capabilities. As a result, pulse-based state preparation constitutes a tool to gauge the ultimate performance of a given processor. Compared to the quantification of coherence times, gate times, and other low-level metrics, the study of METs in pulse-based computation provides a more meaningful way to juxtapose and benchmark different hardwares' computational capabilities: The MET can be used to lower bound the execution time of any quantum algorithm on a given processor.

The early work on pulse-based VQAs \cite{asthanaMinimizingStatePreparation2022,sherbert2024parameterization,meiteiGatefreeStatePreparation2021a,eggerStudyPulsebasedVariational2023,ibrahimEvaluationParameterizedQuantum2022,meiromPANSATZPulsebasedAnsatz2023} focused on superconducting hardware. One of the main competitors to superconducting qubits is a system of quantum-dot spin qubits in silicon \cite{PhysRevA.57.120,PhysRevB.59.2070,RevModPhys.85.961,RevModPhys.95.025003}. The small size of silicon qubits and the potential for industrial scaling through advanced semiconductor-manufacturing technology \cite{Zwerver2022,10185272,9371956,elsayed2022low,Neyens2024} make them an attractive hardware option for quantum processors. Other beneficial features include long coherence times \cite{Stano2022,tanttu2024assessment}, competitive gate fidelities \cite{Yoneda2018, Xue2022, Noiri2022, doi:10.1126/sciadv.abn5130, Veldhorst2015Oct, Huang2024}, and operability above 1K \cite{Huang2024}. This has led to demonstrations of a six-qubit processor \cite{Philips2022}. Silicon hardware has also hosted a quantum simulator of Hubbard models \cite{Hensgens2017,PRXQuantum.2.017003}. While the realization of such simulators constitutes important milestones, the approach of analog quantum simulation is inherently limited to the emulation of specific Hamiltonians that are consistent with the architecture and native interactions of the device, such as lattice spin models. Many other problems of interest, such as simulating molecules, are not feasible with analog-simulator approaches. Pulse-based VQAs, on the other hand, enable universal quantum simulation. However, their performance on silicon hardware has, until now, not been explored.

In this paper, we investigate the performance of pulse-based state preparation on silicon quantum hardware. A schematic of the types of devices we consider is shown in \cref{fig: process diagram}(a). We use parameters extracted from recent experiments \cite{Yoneda2021,Huang2024,Philips2022, t_values,doi:10.1126/sciadv.1600694,Simmons2009,Weinstein2023Mar} to estimate the METs in two important settings. The first setting is the preparation of molecular ground states, a key component in quantum chemistry algorithms. The second setting is the unitarily driven transition between pairs of random states. We find that the METs for preparing the molecular ground state of \ce{H2}, \ce{HeH+}, and \ce{LiH} are \SI{2.4}{\ns}, \SI{4.4}{\ns}, and \SI{27.2}{\ns}, respectively. This is significantly faster than the $\sim$15ns MET computed for \ce{H2} on a transmon architecture in Ref.~\cite{asthanaMinimizingStatePreparation2022}. For transitions between arbitrary pairs of four-qubit states, we show that the MET is always shorter than \SI{50}{\ns}. Interestingly, both values are significantly smaller than METs of gate-based approaches,  
(which require at least \SI{200}{\ns}), indicating that significant noise improvements are available via pulse-based quantum computation. Finally, we explore how pulse parameters affect the MET. For example, we show that increasing the maximal exchange amplitude from \SI{10}{\MHz} to \SI{1}{\GHz} accelerates the METs of \ce{H2} from  \SI{84.3}{\ns} to \SI{2.4}{\ns}. 

\begin{figure*}[t]
    \centering
    \includegraphics{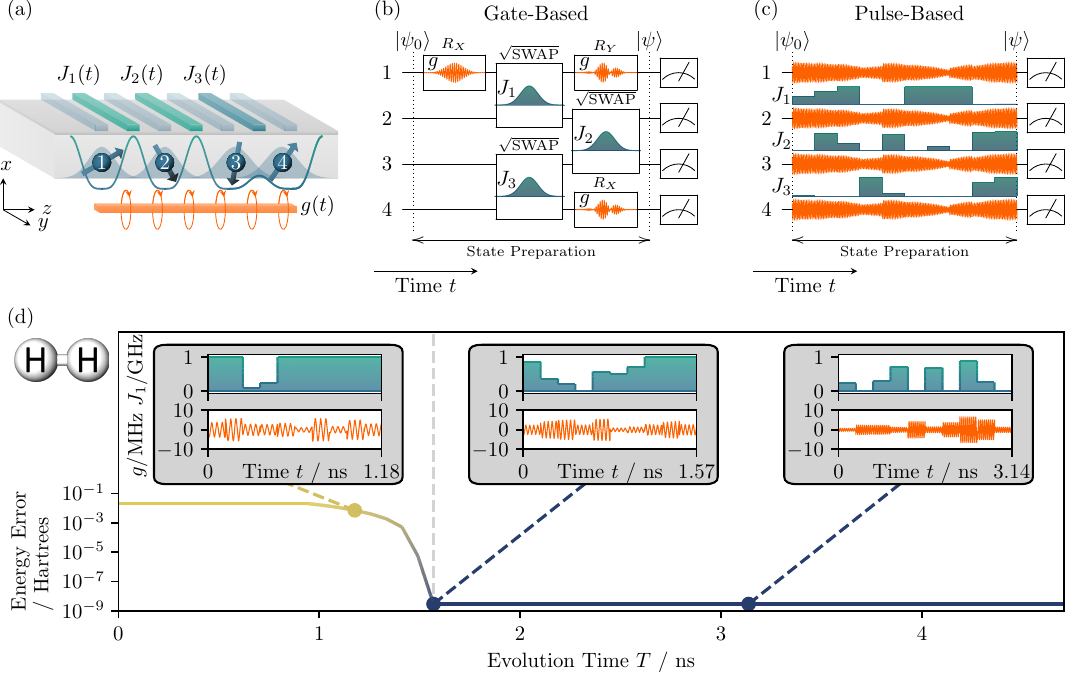}
    \caption{\textbf{Device Diagram, State preparation, and MET. (a) Diagram of a SiMOS processor:} Plunger- (semi-transparent blue) and $J$-gates (blue) on a \ce{SiO_x} layer (gray) confine electron spins in potential wells (blue curve) within silicon (light gray). The $J$-gates can lower the potential barriers and induce exchange interactions (qubits 3 and 4). A static magnetic field along the $z$-direction induces Zeeman splittings. An electron-spin-resonance antenna (orange) induces a time-dependent magnetic field. \textbf{(b) Gate-based state preparation:} $\ket{\psi}$ is prepared using pre-calibrated one- and two-qubit gates (black boxes in the circuit diagram) driven by microwave bursts (orange lines) and voltage pulses (blue shaded regions), respectively. \textbf{(c) Pulse-based state preparation:} $\ket{\psi}$ is prepared by changing microwave bursts (orange lines) and voltage pulses (blue-shaded regions) simultaneously. \textbf{(d) Numerically calculated energy errors $\Delta(T)$'s for \ce{H2}.} A vertical dashed line marks the MET. Insets show the optimal pulse shapes as a function of time $t$ for various total evolution times $T$.}
    \label{fig: process diagram}
\end{figure*}

The remainder of this article is structured as follows: In \cref{sec: device model}, we introduce our silicon-processor model. In \cref{Sec: StatePrep}, we review pulse-based state preparation. In \cref{sec: molecular} and \cref{sec: haar}, we present METs for pulse-based preparation of molecular ground states and Haar random state transitions, respectively. In \cref{sec: limiting factors}, we explore the dependence of METs on silicon control parameters. We conclude in \cref{sec: conclusion}.

\section{Silicon quantum-processor model}
\label{sec: device model}

We use a Heisenberg spin-chain model \cite{PhysRevB.100.035304} to describe a typical 1D chain of spin qubits in SiMOS heterostructures \cite{Yoneda2021,10185272,elsayed2022low,Huang2024} [see \cref{fig: process diagram}(a)]:
\begin{equation}
    \label{eq:DeviceModel}
  \!\hat H\!\left(t\right)\!=\!-\!\!\sum_{i=1}^N\!\frac{B_i}{2}\hat \sigma^{\left(i\right)}_z\!+\!\frac{g\!\left(t\right)}{2}\!\!\sum_{i=1}^N\!\hat \sigma^{\left(i\right)}_x-\!\!\!\sum_{i=1}^{N-1}\!\frac{J_i\!\left(t\right)}{4}\vec{\hat \sigma}^{\left(i\right)}\cdot \vec{\hat \sigma}^{\left(i+1\right)}.\!
\end{equation}
Here, $\vec{\hat{\sigma}}^{\left(i\right)} = \left( \hat{\sigma}_x^{(i)},\hat{\sigma}_y^{(i)},\hat{\sigma}_z^{(i)} \right)$ is a vector of Pauli operators of the spin of a single electron, trapped in the $i$th quantum dot under a plunger gate.

An external magnetic field ($\approx\SI{1}{\tesla}$) induces a Zeeman splitting of $B_i=B+\Delta B_i$ on the electron spin in the $i$th dot with an average value of $B=\frac{1}{N}\sum_{i=1}^N B_i$ ($\approx\SI{28}{GHz}$) \cite{Yoneda2021}. We assume detunings of $B_{i+1}-B_i=\Delta B$, similar to Refs.~\cite{Huang2024,Yoneda2021},  where $\Delta B$ is approximately $\qtyrange{10}{30}{\MHz}$.

One-qubit operations are achieved via electron-spin resonance using a microwave antenna \cite{Dehollain_2013} [see \cref{fig: process diagram}(a)], which induces a time-dependent global magnetic field in the $x$-direction. The field induced by the microwave drive is proportional to
\begin{align}
g\left(t\right)=\sum_{i=1}^N\left[I_i\left(t\right)\cos\left(\omega_it\right)+Q_i\left(t\right)\sin\left(\omega_it\right)\right].
\end{align}
The field can be produced with an arbitrary waveform generator, where $I_i \left( t \right)$ and $Q_i\left(t\right)$ are the in-phase and quadrature components of $N$ microwave pulses with carrier frequency $\omega_i$, respectively. Typically, $|I_i\left(t\right)|, |Q_i\left(t\right)|\le$\SI{10}{\MHz} \cite{doi:10.1126/sciadv.1600694} and $\omega_i$ usually varies by up to \SI{1}{\GHz} from the Zeeman splitting (depending on the bandwidth of the arbitrary waveform generator).

Finally, two-qubit operations are induced via nearest-neighbor exchange couplings between spin $i$ and $i+1$. The entangling operation is controlled by voltage pulses on the barrier gates, which alters the exchange strength $J_i(t)$. Typically, $J_i(t)\le\SI{10}{\MHz}$ \cite{tanttu2024assessment, Huang2024, Veldhorst2015Oct}, but exchange couplings in SiMOS as large as $J_i(t)\approx\SI{1}{\GHz}$ have already been measured \cite{t_values}. Moreover, the control of fast exchange pulses has previously been demonstrated in \ce{GaAs} \cite{PhysRevLett.110.146804, PhysRevLett.116.116801}. As we will show below, fast exchange is crucial to achieve short METs.

The parameters of our spin-chain model are summarized in \cref{table: parameter values}. As discussed above, these parameters describe realistic SiMOS devices \cite{ Yoneda2021, Huang2024, t_values, doi:10.1126/sciadv.1600694}. Similar parameters also describe electron spin-qubits in Silicon/SiliconGermanium (\ce{Si/SiGe}) heterostructures \cite{Weinstein2023Mar,Simmons2009,Xue2022,Noiri2022,doi:10.1126/sciadv.abn5130,Philips2022,doi:10.1126/science.aao5965}. In \ce{Si/SiGe}, the Zeeman-splitting inhomogeneity and $g\left(t\right)$ originates from micromagnets \cite{Yoneda_2015} and electron dipole spin resonance type driving via the confinement gate \cite{Xue2022, Noiri2022,  doi:10.1126/sciadv.abn5130, Philips2022, doi:10.1126/science.aao5965}, respectively. As with SiMOS devices, experiments in Si/SiGe tend to use slow exchange ($J_i(t) \le \SI{10}{\MHz}$) \cite{Xue2022, Noiri2022, doi:10.1126/sciadv.abn5130, Philips2022, doi:10.1126/science.aao5965}. However, fast exchange ($J_i(t) \approx \SI{1}{\GHz}$) has already been demonstrated in Refs.~\cite{Weinstein2023Mar,Simmons2009}.

\section{Pulse-based state preparation}
\label{Sec: StatePrep}

In conventional gate-based approaches, a target quantum state is prepared by evolving an initial state through a sequence of quantum gates. These gates are implemented by executing pre-optimized control pulses to realize the specific parameterized one- and fixed two-qubit gates. In \cref{app: gate-bounds}, we also allow two-qubit gates to be parameterized and calculate gate-based state-preparation times. These times serve as lower bounds for state preparation with fixed two-qubit gates. In contrast, pulse-based state preparation is realized by directly optimizing the control pulses that drive the device state through native interactions. Since gate-based state preparation is less flexible, in a variational sense, pulse-based state preparation is able to execute in shorter times.

A device-agnostic implementation of pulse-based state preparation utilizing variational pulse-shaping \cite{KHANEJA2005296,meiteiGatefreeStatePreparation2021a} consists of three steps: (i) The input is a cost Hamiltonian $\hat{C}$ whose (potentially unknown) ground state $\ket{\phi}$ we wish to prepare. For example, $\hat{C}$ can encode a molecular Hamiltonian \cite{RevModPhys.92.015003}. It can always be expressed as a linear combination of Pauli strings \cite{RevModPhys.92.015003,TILLY20221} whose expectation value, in the case of physically relevant problems, can be sampled efficiently \cite{TILLY20221}. (ii) We fix the total evolution time $T$ and prepare the qubit register in the state $\ket{\psi(\bm{x};T)}$. Here, $\bm{x}$ denotes the control parameters of the silicon processor. (iii) These parameters are iteratively tuned using a classical optimizer to obtain the minimal expectation value 
\begin{align}
    C(T) \coloneqq\min_{\bm{x}}\left\{\mel{\psi(\bm{x};T)}{\hat C}{\psi(\bm{x};T)}\right\}.
\end{align}
Provided that $T$ is sufficiently large and $\bm{x}$ sufficiently expressive, variational pulse-shaping yields $\ket{\psi(\bm{x}_{\star};T)}=\ket{\phi}$ for the optimal parameters $\bm{x}_{\star}$.

Specifically, for our device (see \cref{sec: device model}), a cost Hamiltonian $\hat{C}$ and fixed evolution time $T$, we implement variational pulse-shaping as follows: Working in the frame rotating with the drift Hamiltonian $\hat{H}_D \coloneqq-\frac{1}{2}\sum_{i=1}^N\!B_i\hat \sigma^{\left(i\right)}_z$, we numerically integrate the Schrödinger equation of $\hat{H}$. This yields the state $\ket{\psi(\bm{x}; T)}$ in the rotating frame. We optimize its control parameters $\bm{x} \coloneqq \left\{I_i\left(t\right), Q_i\left(t\right), J_j\left(t\right), \omega_i\right\}$ within the constraints of \cref{table: parameter values}. We use the quadrature parameterization of $g\left(t\right)$ as opposed to the polar parameterization to make optimization easier \cite{sherbert2024parameterization}. Further, for convenience, we parametrize each control using a piecewise linear function. That is, we use $M$ segments of duration $\Delta t=T/M$. For each $t\in\left[\left(m-1\right)\Delta t,m\Delta t\right)$ with $m=1,...,M$ we set $I_i\left(t\right)=I_{i,m}$, $Q_i\left(t\right) =Q_{i,m}$ and $J_j\left(t\right)=J_{j,m}$. We then use the gradient ascent pulse-engineered (GRAPE) algorithm \cite{KHANEJA2005296} to find the parameters that achieve the minimal cost $C(T)$.

\cref{fig: process diagram}(d) illustrates numerically calculated $C(T)$s where $\hat{C}$ encods the Hamiltonian of \ce{H2}. For large $T$, GRAPE finds control parameters, such that $\ket{\psi(\bm{x};T)}$ approches the ground state $\ket{\phi}$ of \ce{H2} and $C(T)$ approaches the \ce{H2} ground-state energy $C_0$ of $\hat{C}$. For sufficiently small $T$, GRAPE no longer finds the ground state of \ce{H2}, and the energy error $\Delta(T)\coloneqq C(T)-C_0$ increases monotonically as a function of $T$. The time below which all $\Delta(T)>10^{-7}$ Hartrees is a numerical estimate of the minimal evolution time (MET). Below the MET, a silicon quantum processor can no longer accurately prepare the ground state of $\hat{C}$. Thus, hardware-specific METs provide lower bounds on the time of desired quantum evolutions. A protocol for implementing variational pulse-shaping and MET searches in experiments is given in \cref{app: cost function}.

\begin{figure*}
  \includegraphics{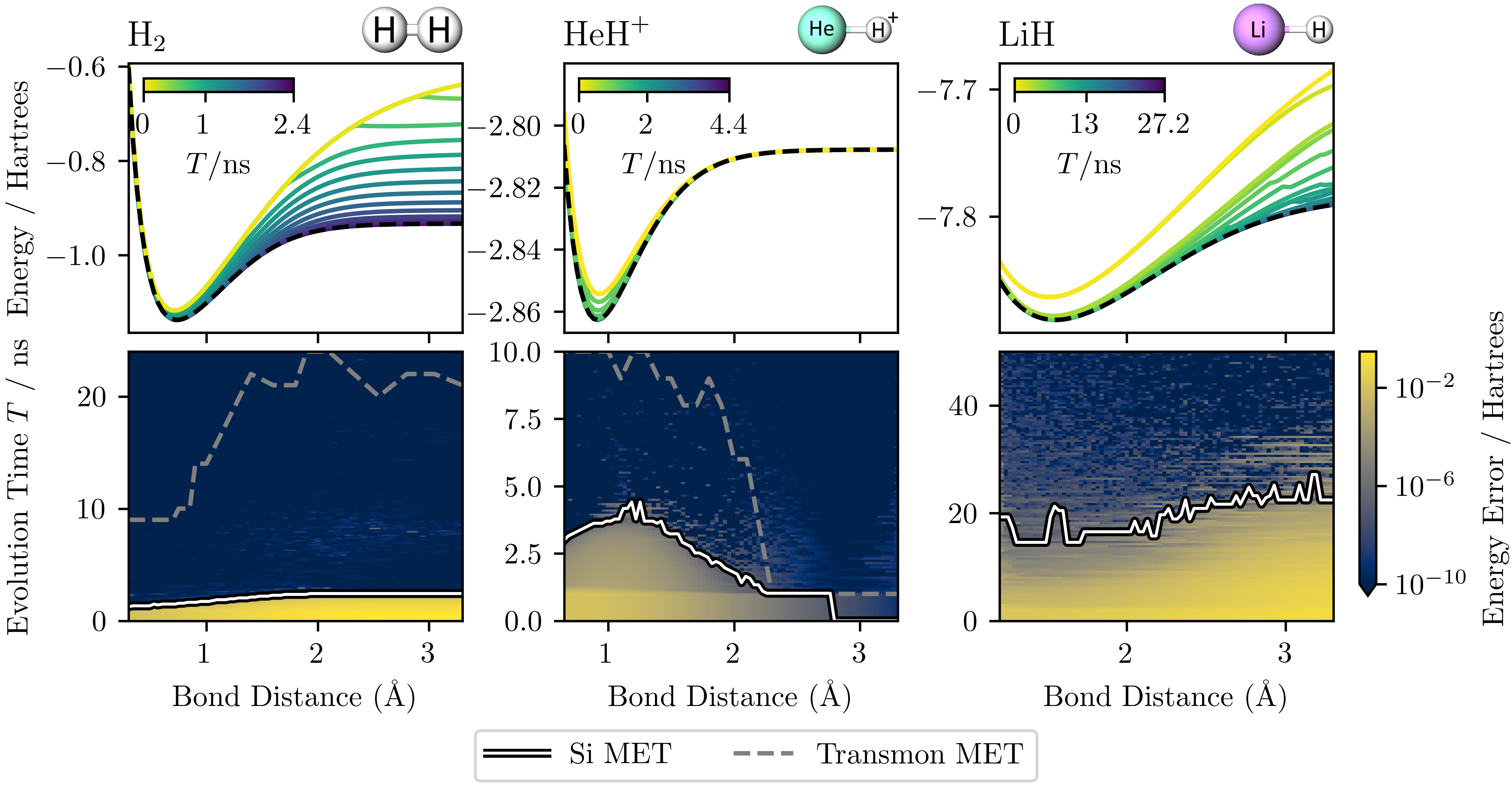}
  \caption{\textbf{Dissociation curves and METs for molecular-ground-state preparation} of \ce{H2}, \ce{HeH+} and \ce{LiH} (left to right).
  (top) Molecular energy $C(T)$ as a function of bond distance for a range of evolution times. At zero evolution time, $C(0)$ corresponds to the (yellow) Hartree-Fock disassociation curve of the initial state. At the MET, the energy $C(T=\textrm{MET})$ (blue dissociation curve)  converges to the full configuration interaction (FCI) energy $C_0$ (dashed black curve). (bottom) Energy error $\Delta(T)\coloneqq C(T)-C_0$ on a logarithmic color scale as a function of bond distance and evolution time. A white line shows the silicon METs. A dashed (grey) curve marks transmon METs~\cite{meiteiGatefreeStatePreparation2021a}.
  }
  \label{fig: disassociation}
\end{figure*}

\section{Minimal evolution times for molecular ground states}
\label{sec: molecular}

In this section, we investigate the silicon processor's METs for variational pulse-based state preparation of molecular ground states. To facilitate comparison with previous works on superconducting hardware, we consider the molecules \ce{H2}, \ce{HeH+}, and \ce{LiH} with the same molecular Hamiltonians as in Ref.~\cite{meiteiGatefreeStatePreparation2021a}. We work in the STO-3G basis and use parity encoding that block diagonalizes the Hamiltonians. We then identify $\hat C$ with the block containing the ground state. We permute our basis so that the Hartree Fock state is given by $\ket{01}$ for \ce{H2} and \ce{HeH+}, and $\ket{0011}$ for \ce{LiH}---we then take these states as our initial states. The time required to prepare these initial states is not counted towards the MET.

Using the pulse-shaping methods of \cref{Sec: StatePrep}, we numerically calculate the molecular energies $C(T)$ of \ce{H2}, \ce{HeH+}, and \ce{LiH} as functions of the bond distances for a range of evolution times $T$. See \cref{fig: disassociation} (top row) for results. At zero evolution time $T=0$, the prepared state is the Hartree-Fock state, and the molecular energy $C(0)$ is the Hartree-Fock energy. A yellow line shows the corresponding dissociation curve. As $T$ increases, the (increasingly blue) dissociation curves $C(T)$ converge to the full configuration interaction (FCI) energies $C_0$, marked by a black dashed curve. For $T=\mathrm{MET}$, the (purple) dissociation curve $C(\mathrm{MET})$ is identical to the (FCI) energies $C_0$. The color bar is truncated at the largest MET along the bond distance.

To identify the METs for molecular ground-state preparation of \ce{H2}, \ce{HeH+}, and \ce{LiH}, we plot the energy error $\Delta(T) \coloneqq C(T)-C_0$ as a function of bond distance and evolution time $T$. The results are shown in \cref{fig: disassociation} (bottom row). For each bond distance, $\Delta(T)$ shows a sharp increase as $T$ decreases. For numerical purposes, we estimate the MET as the point where $\Delta(T)>10^{-7}$ Hartrees. A white line highlights the silicon METs. For \ce{HeH+}, the MET drops to zero at a large bond distance because the Hartree-Fock state converges to the true FCI ground state. For \ce{H2}, \ce{HeH+}, and \ce{LiH} we obtain METs as small as \SI{2.4}{\ns}, \SI{4.4}{ns}, and \SI{27.2}{ns}, respectively. For comparison, a gate-based circuit for preparing arbitrary states on two qubits requires at least \SI{200}{ns} (see \cref{app: two-qubit-gates}). Further,  these METs are one order of magnitude faster than the fastest single qubit $\pi$ rotations in silicon \cite{Stano2022}. These silicon METs are also one order of magnitude faster than the numerically-estimated METs for superconducting transmon qubits from Ref.~\cite{meiteiGatefreeStatePreparation2021a}, which are shown as gray dashed lines in \cref{fig: disassociation} (bottom). Overall, these results fuel the hope that pulse-based variational state preparation methods could significantly accelerate the state preparation in quantum chemistry and, thus, make near-term quantum algorithms on silicon quantum processors more resilient to noise.

\begin{figure*}
  \includegraphics{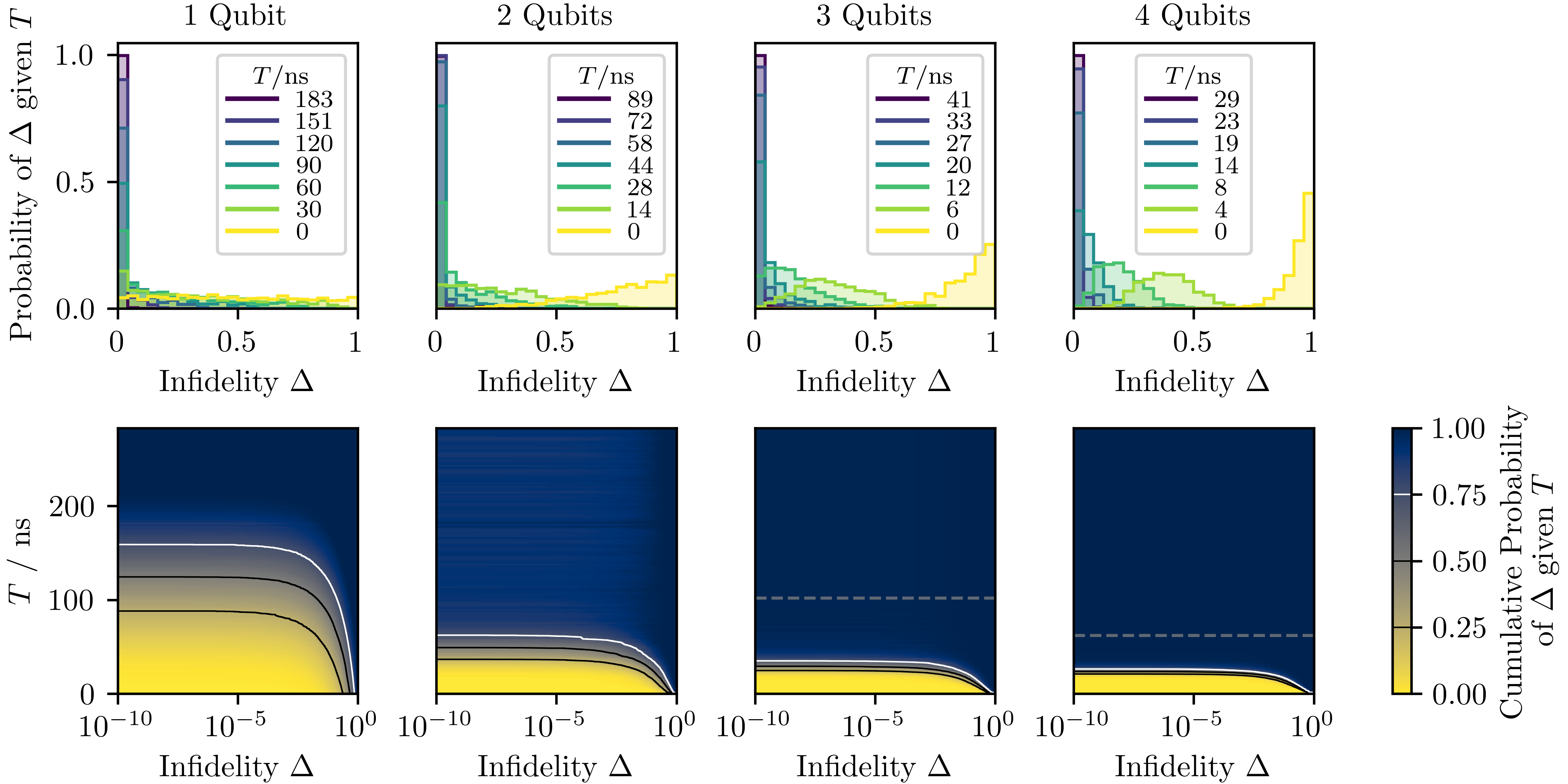}
  \caption{\textbf{Infidelity distributions \textit{vs.}\ evolution time $T$} for $1,2,3$, and $4$ qubits (left to right). (Top) Infidelity of 1024 Haar-random state pairs as histograms. Histogram colors indicate the evolution time $T$ as depicted by the insets. (Bottom) Cumulative distribution of infidelities as a color map, with infidelity on a log scale on the horizontal axis and evolution time $T$ on the vertical axis. (Right two columns bottom) Numerically calculated data is below the dashed lines, while extrapolated data is above. We average the ten distributions with $T$ just below the dashed curves for extrapolation.}
  \label{fig: haar}
\end{figure*}

We conclude this section with three technical remarks concerning our numerical simulations: (i) We discretized the control parameters into $M=10$ segments. (ii) To obtain the numerical data in \cref{fig: disassociation}, we first initialize GRAPE with a pulse of zero amplitude and three random pulse shapes for a bond distance where optimization is easy and select the final pulse shapes with minimal $C(T)$. Next, we move to the neighboring bond distance, where optimization is harder, and initialize GRAPE with three random pulse shapes and the optimal pulse shape of the previous bond distance. Then, we select the pulse shape with the minimal $C(T)$. For \ce{H2} and \ce{LiH}, the Hartree Fock ground state is a good approximation to the true ground state at short bond lengths. Thus, we scan from smaller to larger bond distances. For \ce{HeH+}, the Hartree-Fock state is identical to the FCI ground state in the limit of large bond distances. Hence, we scan from larger to smaller bond distances. (iii) We emphasize that scanning along bond distance, coupled with GRAPE, is not expected to be scalable to large qubit numbers. However, this work aims to investigate the ultimate performance limits of VQAs in silicon on small qubit numbers. Presenting a scalable method to achieve this optimal performance at large qubit numbers is a task left for future investigation.

It is now natural to ask: \textit{How does the MET depend on the device parameters?} We answer this question in \cref{sec: limiting factors}. But first, we consider the task of transitioning between random states in \cref{sec: haar}. Readers who are primarily interested in quantum computational chemistry simulations may skip ahead to \cref{sec: limiting factors}.

\section{Haar-random state transitions}
\label{sec: haar}

In this section, we determine the METs for transitions between arbitrary states. As the initial state for this task is no longer close to the target state, this will require longer METs compared to chemistry or physics problems, where the mean-field solution is often a good starting point. We find that for systems with two, three, or four qubits, the slowest MET is still faster than the fastest possible one-qubit $\pi$ gate. Further, we theoretically characterize the distribution of METs between Haar random states.

To numerically investigate METs between arbitrary state pairs on $N$ qubits, we first sample 1024 pairs of random input states $\ket{\psi_0}$ and target states $\ket{\phi}$ from the Haar measure. If arbitrary quantum control was available one could pick a fixed initial state and just sample random target states. However, real hardware cannot move uniformly through Hilbert space. Consequently, some initial states are ``better'' than others. For each state pair, we use the silicon device Hamiltonian $\hat{H}$ to numerically evolve $\ket{\psi_0}$ into the state $\ket{\psi(\bm{x};T)}$, maximizing its fidelity (overlap) with $\ket{\phi}$ as
\begin{align}
    F(T)\coloneqq \max_{\bm{x}}\left\{\left|\braket{\psi(\bm{x};T)}{\phi}\right|^2\right\}.
\end{align}
To achieve this goal for a given state pair and a given evolution time $T$, we use variational pulse-shaping with the cost Hamiltonian $\hat{C}$ defined as $\hat{C} \coloneqq \mathds{1}-\ket{\phi}\bra{\phi}$. The expectation value of $\hat{C}$ is the infidelity $C\left(T\right)\coloneqq 1-F(T)$ to the state $\ket{\phi}$. The results of our numerical investigation are depicted in Figs.~\ref{fig: haar} and \ref{fig: MET distribution} for one to four qubits (left to right, respectively).

In \cref{fig: haar}, we analyze the infidelity of each state pair as a function of evolution time $T$. The top panels show the infidelities of the 1024 state pairs as histograms, where each histogram corresponds to a distinct evolution time $T$. For $T=0$ (yellow histogram), the distribution of infidelities is broad, indicating the broad range of overlaps $\left|\braket{\psi_0}{\phi}\right|^2$ between Haar-random state pairs prior to time evolution (an analytical expression for this initial distribution is given in \cref{eq: initial dist} in \cref{app: homogeneous and isotropic speed limit}). As $T$ increases (increasingly blue histograms), the evolved input states $\ket{\psi(\bm{x}; T)}$ increase their overlap with the corresponding target states $\ket{\phi}$. This leads to higher fidelities $F(T)$, shifting the histograms towards lower infidelity. Moreover, as $T$ increases, an increasing fraction of states $\ket{\psi(\bm{x};T)}$ reaches the target state $\ket{\phi}$. For this fraction of state pairs $F(T)=1$, leading to an increasing peak of the histograms around infidelity equal to zero. Finally, as the MET for the hardest-to-connect state pair is reached, the (purple) infidelity histogram entirely localizes around infidelity equal to zero. The bottom row shows the estimated cumulative probability distribution of the infidelity, $\mathbb P\left(C\left(T\right)\le \Delta\middle | T\right)$, of state pairs on a log scale along the horizontal axis and as a function of the evolution time $T$ along the vertical axis.

In \cref{fig: MET distribution}, we approximate the cumulative distribution of the MET, $\mathbb{P}\left(\textrm{MET}\le T\right)$, by the fraction of the 1024 state pairs, which reach an infidelity $C(T)<10^{-7}$ in an evolution time less than $T$. A blue-shaded region marks the corresponding confidence interval of 99.99\%. For comparison, we overlay the time required by the same processor to implement $\frac{1}{\sqrt{2}}\left(X\pm Y\right)$ (the fastest single-qubit $\pi$-gate), $X$, $Y$, and $\sqrt{\textrm{SWAP}}$ gates as blue vertical lines in \cref{fig: MET distribution}. In a one-qubit system, the time for operating the fastest $\pi$-gate is slightly shorter than the MET of the hardest-to-connect state pairs. However, for systems with more than two qubits, we observe that even the slowest random state transitions can be accomplished in less than the time set by the fastest one-qubit $\pi$-gate. This may imply that embedding hard-to-connect single-qubit states into a system with two or more qubits could accelerate the transition between one qubit states by using, e.g., fast exchange interactions in the expanded space of neighboring qubits. This effect is reminiscent of accelerated state preparation through excited states in superconducting devices \cite{asthanaMinimizingStatePreparation2022} and will be investigated in more detail in future works.

\begin{figure*}[t]
  \includegraphics{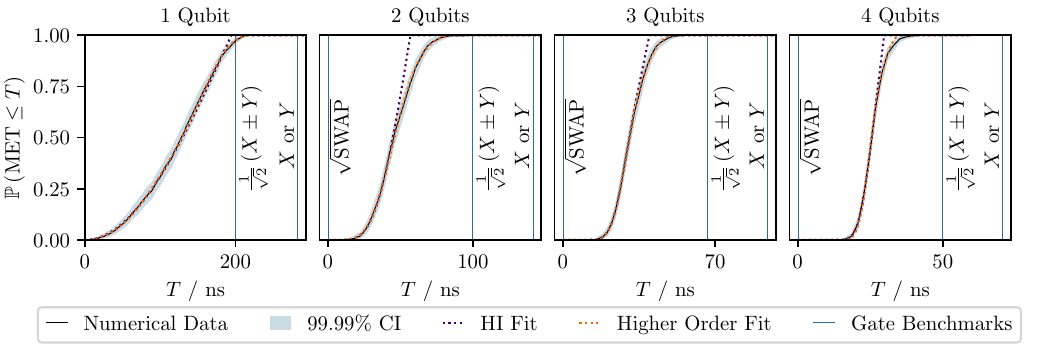}
  \caption{\textbf{Cumulative probability distribution of the MET for haar random state transitions.} The black solid curves are the numerical data for 1024 state pairs. The shaded region corresponds to a 99.99\% confidence interval (CI) estimated via $10^5$ bootstrap resamples (see \cref{app: bootstrapping}). Vertical blue lines mark typical gate times. The dashed purple and orange curves are the homogeneous and isotropic (HI) fits and higher-order fits, respectively. For details of the fit see \cref{app: chis}.}
  \label{fig: MET distribution}
\end{figure*}

Next, to underpin our numerical studies, we derive two expressions to describe the cumulative distribution of METs depicted in \cref{fig: MET distribution}. The first expression assumes that a pair of states is connected in the shortest time by traveling with a constant speed $v$ along the geodesics of the manifold of $N$-qubit states---i.e., the speed limit is homogeneous and isotropic. For a single qubit, this manifold corresponds to the Bloch sphere, and we assume the fastest transition connecting any two quantum states on the Bloch sphere corresponds to traveling along an arc of a great circle at speed $v$. In \cref{app: homogeneous and isotropic speed limit}, we show that this assumption gives rise to the following distribution: 
\begin{multline}\label{eq: analytic}
  \mathbb P\left(\textrm{MET}\le T\right)=\frac{2}{B\left(d-1,\frac{1}{2}\right)}\int\limits_0^{\frac{vT}{2}}\dd{x}\sin^{2d-3}(x),
\end{multline}
for $0\le T\le\frac{\pi}{v}$ where $d$ is the Hilbert space dimension, and $B\left(\bullet,\bullet\right)$ is the beta function. We fit the numerical data by adjusting the speed parameter $v$---this gives the dashed purple curves in \cref{fig: MET distribution}. In general, the assumption (that a pair of states is connected in the shortest time by traveling along a geodesic with a constant speed) will not hold---as in our particular hardware. Our second expression,
\begin{align}\label{eq: expansion}
  &\mathbb P\left(\textrm{MET}\le T\right)=\int\limits_0^{\frac{\tilde vT}{2}}\dd{x}\sum_{n=2d-3}^\infty c_n\sin^{n}(x),\\
  \textrm{where }&\tilde v\coloneqq\frac{\pi}{\displaystyle\max_{\ket{\psi_0}, \ket{\phi}}\textrm{MET}\left(\ket{\psi_0}, \ket{\phi}\right)},
\end{align}
derived in \cref{app: inhomogeneous and anisotropic speed limit}, generalizes the first expression by adding correction terms to account for the fact that the speed at which the state can move on the manifold may depend on the location within the manifold and the direction of motion---i.e., the speed limit is inhomogeneous and anisotropic. For example, in our simulations of silicon hardware, we fix the maximum values for $I_i\left(t\right)$ and $Q_i\left(t\right)$, which induce $X$ and $Y$ rotations, respectively. We can rotate a state about the $\frac{1}{\sqrt{2}}\left(X\pm Y\right)$ axes by simultaneously driving $I_i\left(t\right)$ and $Q_i\left(t\right)$. Thus, rotations about the $\frac{1}{\sqrt{2}}\left(X\pm Y\right)$ axes will be a factor of $\sqrt{2}$ faster than about the $X$ or $Y$ axes. We fit a truncation of the generalized expression [\cref{eq: expansion}]; see the dashed orange curves in \cref{fig: MET distribution} and \cref{app: chis}.

Finally, we conclude this section with two brief technical remarks about our numerical simulations: (i) We discretized the control parameters into $M=40$ segments. (ii) Like in \cref{sec: molecular}, we implement variational pulse-shaping as follows: First, we initialize GRAPE with three random pulse shapes at the maximum evolution time $T$ over which we will scan. Next, we move to the next shortest evolution time $T$ and initialize GRAPE with two random pulse shapes and the optimal pulse shape of the previous value for $T$. Then, we select the pulse shape with the minimal $C(T)$.

\section{Dependence of minimal evolution times on silicon device parameters}
\label{sec: limiting factors}

In this section, we investigate how silicon METs depend on the device parameters and their experimental bounds in a silicon processor. In particular, we investigate the dependence of METs on the maximal drive strength $\textrm{IQ}_{\max}$, the maximal exchange coupling $J_{\max}$, and the Zeeman-splitting inhomogeneity $\Delta B$. The investigation is motivated by gate-based qubit operations in silicon. On the one hand, $\textrm{IQ}_{\max}$ is known to limit the speed of one-qubit operations, which are known to be relatively slow \cite{Stano2022}. On the other hand, $J_{\max}$ is known to enable relatively fast two-qubit power-of-SWAP gates \cite{Stano2022}. We investigate whether similar parameter limitations apply to pulse-based state preparation.

\begin{figure}
  \includegraphics{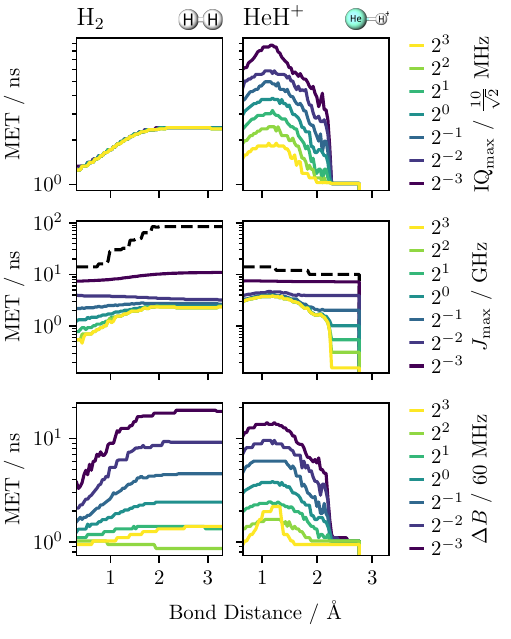}
  \caption{\textbf{MET for molecular ground-state preparation} \textit{vs.}\ bond distance for (left) \ce{H2} and (right) \ce{HeH+}, respectively. Colored curves depict the MET for independently varying (top) the maximal drive strength $\textrm{IQ}_{\max}$, (center) the maximal exchange coupling $J_{\max}$, and (bottom) the Zeeman-splitting inhomogeneity $\Delta B$ relative to their specified values in \cref{table: parameter values} as decoded by the color bars. The black dashed line in the central row is for $J_{\max}=\SI{10}{\MHz}$.}
  \label{fig: mol limit}
\end{figure}

\begin{figure}
  \includegraphics{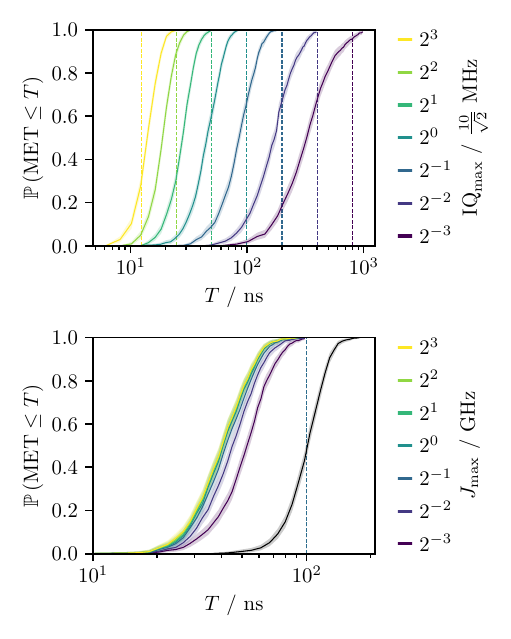}
  \caption{\textbf{Cumulative distribution of METs} for Haar random state transitions (colored curves) for independently varying (top) the maximal drive strength $\textrm{IQ}_{\max}$ and (bottom) the maximal exchange coupling $J_{\max}$ relative to the values in \cref{table: parameter values} as decoded by the color bars. Shaded regions correspond to 99\% confidence intervals. Dashed lines mark the time to perform a single qubit $\left(X\pm Y\right)/\sqrt{2}$ gate. The black curve in the bottom panel is for $J_{\max}=\SI{10}{\MHz}$.}
  \label{fig: haar limit}
\end{figure}

To begin with, we investigate the influence of microwave drives, exchange couplings, and Zeeman-splitting inhomogeneities on METs for molecular ground-state preparation. To this end, we extend the numerical investigation of \cref{sec: molecular} and determine the METs for \ce{H2} and \ce{HeH+}. The data is shown in \cref{fig: mol limit}, where colored curves depict the MET. The top, center, and bottom panels, show METs where $\textrm{IQ}_{\max}$, $J_{\max}$, and $\Delta B$ are varied independently, relative to their values in \cref{table: parameter values}. Roughly speaking, this data confirms that larger $\textrm{IQ}_{\max}$, larger $J_{\max}$, and larger $\Delta B$ lead to faster state-preparation times and thus shorter METs. In addition, we calculate the METs for molecular ground-state preparation (dashed black line in center panels) for $J_{\max}=\SI{10}{\MHz}$---a common value for realizing CZ \cite{Philips2022, Xue2022, doi:10.1126/sciadv.abn5130, tanttu2024assessment} or CNOT \cite{doi:10.1126/science.aao5965} gates. The corresponding METs for \ce{H2} and \ce{HeH+} become as slow as \SI{100}{\ns} and \SI{10}{\ns}, respectively. This observation highlights the importance of fast exchange for accelerated state preparation. A deeper discussion of the METs for molecular ground-state preparation of \ce{H2} and \ce{HeH+} in relation to the one-qubit drives is given below.

Next, we investigate the influence of $\textrm{IQ}_{\max}$ and $J_{\max}$ on transitions between Haar random states. To this end, we extend the numerical investigation of \cref{sec: haar} and determine the cumulative distribution of METs from 1024 pairs of Haar random states. The data is shown in \cref{fig: haar limit}. In the top panel, we apply the device parameters of \cref{table: parameter values} and vary $\textrm{IQ}_{\max}$. As the drive strength varies, the METs are scaled by the same factor. In the bottom panel, we apply the device parameters of \cref{table: parameter values} and vary $J_{\max}$. As the exchange coupling is varied, the cumulative METs remain almost identical. These data confirm our intuition that silicon pulse-based METs---similar to gate-based state-preparation times---are, on average, limited by $\textrm{IQ}_{\max}$. On the other hand, lowering $J_{\max}$ to \SI{125}{MHz} would not lower METs significantly. However, once $J_{\max}$ is reduced to \SI{10}{MHz} (black curve) we find that the MET starts to increase significantly. Again, this illustrates the importance of fast exchange for achieving small METs.

\subsection{Molecular ground states and qubit drives}

We now discuss, in more detail, how specific microwave and exchange pulses achieve the METs in \cref{fig: mol limit}. In \cref{app: ex}, we show that the exchange operation in the rotating frame of the drift Hamiltonian is \textit{power-of-SWAP like}: In particular in the $\ket{01}$, $\ket{10}$ subspace, it drives population transfer with a phase proportional to $\Delta B t_0 /\hbar$. Here, $t_0$ is the time at which exchange is turned on (as compared to the start of the pulse at $t=0$). Thus, waiting until time $t_0$ before switching on the exchange pulse allows us to adjust the phase between the $\ket{01}$ and $\ket{10}$ states.

First, consider the MET data for \ce{H2} in \cref{fig: mol limit}. For \ce{H2}, the initial Hartree-Fock state is $\ket{01}$. The true ground state (over the range scanned) is well approximated by the state proportional to $\ket{01} + \epsilon e^{i\phi} \ket{10}$. Notably, the absence of $\ket{00}$ and $\ket{11}$ components in the molecular ground state of \ce{H2}, implies that single-qubit gates are not required to prepare the molecular ground state of \ce{H2} from $\ket{01}$. For that reason, changing $\textrm{IQ}_{\max}$ (top left figure of \cref{fig: mol limit}) has no influence on the MET. Next, in order to prepare the ground state of \ce{H2}, we require several applications of the exchange pulse to transfer some population from the initial Hartree-Fock state $\ket{01}$ into $\ket{10}$ with the correct phase. In \cref{fig: mol limit} (center left), we can see that this transfer tends to happen faster as $J_{\mathrm{max}}$ is increased. However, after $J_{\mathrm{max}}$ reaches \SI{0.5}{\GHz} the acceleration of the MET slows down. We conjecture that this slowdown is related to a wait time required to build up the correct phase $\Delta B t_0 /\hbar$. Finally, in \cref{fig: mol limit} (bottom left), we see that increasing the detuning $\Delta B$ tends to accelerate the MET. We conjecture that this observation is related to the fact that a larger $\Delta B$ allows for faster accumulation of the phase $\Delta B t_0 /\hbar$.

Finally, for \ce{HeH+}, we observe three different scenarios in \cref{fig: mol limit}, depending on the bond length. At large bond lengths beyond ($\ge\SI{2.8}{\angstrom}$), the initial Hartree-Fock state $\ket{01}$ describes the true ground state with sufficient accuracy---thus, the MET is zero. At intermediate bond distances (\qtyrange{2.2}{2.8}{\angstrom}) the true ground state is approximately proportional to $\ket{01} + \epsilon' e^{i\phi'} \ket{10}$. Neither single qubit gates nor long wait times $t_0$ are required to adjust the transfer phase $\Delta B t_0 /\hbar$ of the exchange pulse. Hence, in that region, the MET is neither dependent on $\textrm{IQ}_{\mathrm{max}}$ nor $\Delta B$. However, it is inversely proportional to the exchange $J_{\mathrm{max}}$. Finally, for smaller bond distances (\qtyrange{0.5}{2.2}{\angstrom}), the true ground state is proportional to $\ket{01} + \epsilon'' e^{i\phi''} \ket{10} + \delta_{00}\ket{00} + \delta_{11}\ket{11}$, with some amplitudes on $\ket{00}$ and $\ket{11}$. Preparing the amplitudes on $\ket{00}$ and $\ket{11}$ requires single qubit drives and increasing $\textrm{IQ}_{\mathrm{max}}$ decreases the MET proportionally, see \cref{fig: mol limit}(top right). Finally, the exchange pulse is needed to prepare probability in $\ket{10}$ with the right amplitude and phase. We conjecture that the trends in $J_{\max}$ and $\Delta B$ arise from the same mechanisms for producing a superposition of $\ket{01}$ and $\ket{10}$ with the correct phase as described for \ce{H2}.

\section{Discussion and Conclusion}
\label{sec: conclusion}

In this article, we have numerically characterized the METs for state preparation in silicon quantum processors. These METs lower bound the ultimate performance for both gate- and pulse-based state preparation and unitary evolution.

We investigated METs for pulse-based preparation of molecular ground states. We found that the ground states of \ce{H2}, \ce{HeH+}, and \ce{LiH} can be prepared in \SI{2.4}{\ns}, \SI{4.4}{ns}, and \ce{27.2}{ns}, respectively. These METs are less than the \SI{200}{ns} required to prepare arbitrary states on two qubits. Additionally, they are an order of magnitude faster than the fastest single qubit $\pi$ rotations in silicon \cite{Stano2022} and the numerically-estimated METs for superconducting transmon qubits from Ref.~\cite{meiteiGatefreeStatePreparation2021a}.

The acceleration observed for molecular ground-state preparation is promising for pulse-based simulations for chemistry purposes. The speed-up stems from the weak entanglement in molecules. Conventional two-qubit gates (e.g., the CNOT or CZ gates) can create maximal entanglement between two qubits. Consequently, gate-based computations utilize significantly longer two-qubit interactions than needed---this is ill-suited for molecular ground-state preparations. Moreover, gate-based computation relies on slow single-qubit gates to limit the amount of entanglement induced during the computation. On the other hand, a pulse-based approach enables the generation of weak entanglement via fast purpose-tailored pulses.

We investigated METs for transitions between Haar-random states and demonstrated that a four-qubit silicon quantum processor can transition between any pair of states within \SI{50}{\ns}. Our data suggests that single-qubit state transitions can be achieved faster by partially entangling the qubits during the operation and disentangling them at the end. This peculiar effect should be studied in more detail in the future.

Finally, we investigated the METs' dependence on silicon device parameters such as the maximally allowed microwave drive, the maximally allowed exchange coupling, or the Zeeman-splitting inhomogeneity. We demonstrated that fast exchange is vital for short state-preparation METs: Decreasing the maximal exchange coupling from \SI{8}{\GHz} to \SI{250}{\MHz} makes little difference to METs. However, decreasing the maximal exchange coupling any further to \SI{10}{\MHz} can, for some problem instances, increase the MET by a factor of $\sim100$. On the other hand, we showed that most METs scale directly proportional to the maximal microwave drive amplitude, which we identify as the limiting factor for most silicon METs.

Our results give hope that pulse-based state preparation on silicon hardware could lead to a new generation of intermediate-scale quantum algorithms that are significantly more resilient to noise.

\begin{acknowledgments}
The authors would like to thank Johnathan Thio, Kyle Sherbert, Ayush Asthana, Guo Xuan Chan, Stefano Bosco, Leon C.~Camenzind, Andrew J.~Ramsay, and Charles G.~Smith for useful discussions.
\end{acknowledgments}

\bibliography{manuscript}

\begin{thebibliography}{55}%
\makeatletter
\providecommand \@ifxundefined [1]{%
 \@ifx{#1\undefined}
}%
\providecommand \@ifnum [1]{%
 \ifnum #1\expandafter \@firstoftwo
 \else \expandafter \@secondoftwo
 \fi
}%
\providecommand \@ifx [1]{%
 \ifx #1\expandafter \@firstoftwo
 \else \expandafter \@secondoftwo
 \fi
}%
\providecommand \natexlab [1]{#1}%
\providecommand \enquote  [1]{``#1''}%
\providecommand \bibnamefont  [1]{#1}%
\providecommand \bibfnamefont [1]{#1}%
\providecommand \citenamefont [1]{#1}%
\providecommand \href@noop [0]{\@secondoftwo}%
\providecommand \href [0]{\begingroup \@sanitize@url \@href}%
\providecommand \@href[1]{\@@startlink{#1}\@@href}%
\providecommand \@@href[1]{\endgroup#1\@@endlink}%
\providecommand \@sanitize@url [0]{\catcode `\\12\catcode `\$12\catcode
  `\&12\catcode `\#12\catcode `\^12\catcode `\_12\catcode `\%12\relax}%
\providecommand \@@startlink[1]{}%
\providecommand \@@endlink[0]{}%
\providecommand \url  [0]{\begingroup\@sanitize@url \@url }%
\providecommand \@url [1]{\endgroup\@href {#1}{\urlprefix }}%
\providecommand \urlprefix  [0]{URL }%
\providecommand \Eprint [0]{\href }%
\providecommand \doibase [0]{https://doi.org/}%
\providecommand \selectlanguage [0]{\@gobble}%
\providecommand \bibinfo  [0]{\@secondoftwo}%
\providecommand \bibfield  [0]{\@secondoftwo}%
\providecommand \translation [1]{[#1]}%
\providecommand \BibitemOpen [0]{}%
\providecommand \bibitemStop [0]{}%
\providecommand \bibitemNoStop [0]{.\EOS\space}%
\providecommand \EOS [0]{\spacefactor3000\relax}%
\providecommand \BibitemShut  [1]{\csname bibitem#1\endcsname}%
\let\auto@bib@innerbib\@empty
\bibitem [{\citenamefont {Cerezo}\ \emph {et~al.}(2021)\citenamefont {Cerezo},
  \citenamefont {Arrasmith}, \citenamefont {Babbush}, \citenamefont {Benjamin},
  \citenamefont {Endo}, \citenamefont {Fujii}, \citenamefont {McClean},
  \citenamefont {Mitarai}, \citenamefont {Yuan}, \citenamefont {Cincio},\ and\
  \citenamefont {Coles}}]{Cerezo2021}%
  \BibitemOpen
  \bibfield  {author} {\bibinfo {author} {\bibfnamefont {M.}~\bibnamefont
  {Cerezo}}, \bibinfo {author} {\bibfnamefont {A.}~\bibnamefont {Arrasmith}},
  \bibinfo {author} {\bibfnamefont {R.}~\bibnamefont {Babbush}}, \bibinfo
  {author} {\bibfnamefont {S.~C.}\ \bibnamefont {Benjamin}}, \bibinfo {author}
  {\bibfnamefont {S.}~\bibnamefont {Endo}}, \bibinfo {author} {\bibfnamefont
  {K.}~\bibnamefont {Fujii}}, \bibinfo {author} {\bibfnamefont {J.~R.}\
  \bibnamefont {McClean}}, \bibinfo {author} {\bibfnamefont {K.}~\bibnamefont
  {Mitarai}}, \bibinfo {author} {\bibfnamefont {X.}~\bibnamefont {Yuan}},
  \bibinfo {author} {\bibfnamefont {L.}~\bibnamefont {Cincio}},\ and\ \bibinfo
  {author} {\bibfnamefont {P.~J.}\ \bibnamefont {Coles}},\ }\bibfield  {title}
  {\bibinfo {title} {Variational quantum algorithms},\ }\href
  {https://doi.org/10.1038/s42254-021-00348-9} {\bibfield  {journal} {\bibinfo
  {journal} {Nature Reviews Physics}\ }\textbf {\bibinfo {volume} {3}},\
  \bibinfo {pages} {625} (\bibinfo {year} {2021})}\BibitemShut {NoStop}%
\bibitem [{\citenamefont {Tilly}\ \emph {et~al.}(2022)\citenamefont {Tilly},
  \citenamefont {Chen}, \citenamefont {Cao}, \citenamefont {Picozzi},
  \citenamefont {Setia}, \citenamefont {Li}, \citenamefont {Grant},
  \citenamefont {Wossnig}, \citenamefont {Rungger}, \citenamefont {Booth},\
  and\ \citenamefont {Tennyson}}]{TILLY20221}%
  \BibitemOpen
  \bibfield  {author} {\bibinfo {author} {\bibfnamefont {J.}~\bibnamefont
  {Tilly}}, \bibinfo {author} {\bibfnamefont {H.}~\bibnamefont {Chen}},
  \bibinfo {author} {\bibfnamefont {S.}~\bibnamefont {Cao}}, \bibinfo {author}
  {\bibfnamefont {D.}~\bibnamefont {Picozzi}}, \bibinfo {author} {\bibfnamefont
  {K.}~\bibnamefont {Setia}}, \bibinfo {author} {\bibfnamefont
  {Y.}~\bibnamefont {Li}}, \bibinfo {author} {\bibfnamefont {E.}~\bibnamefont
  {Grant}}, \bibinfo {author} {\bibfnamefont {L.}~\bibnamefont {Wossnig}},
  \bibinfo {author} {\bibfnamefont {I.}~\bibnamefont {Rungger}}, \bibinfo
  {author} {\bibfnamefont {G.~H.}\ \bibnamefont {Booth}},\ and\ \bibinfo
  {author} {\bibfnamefont {J.}~\bibnamefont {Tennyson}},\ }\bibfield  {title}
  {\bibinfo {title} {The variational quantum eigensolver: A review of methods
  and best practices},\ }\href
  {https://doi.org/https://doi.org/10.1016/j.physrep.2022.08.003} {\bibfield
  {journal} {\bibinfo  {journal} {Physics Reports}\ }\textbf {\bibinfo {volume}
  {986}},\ \bibinfo {pages} {1} (\bibinfo {year} {2022})},\ \bibinfo {note}
  {the Variational Quantum Eigensolver: a review of methods and best
  practices}\BibitemShut {NoStop}%
\bibitem [{\citenamefont {Dalton}\ \emph {et~al.}(2024)\citenamefont {Dalton},
  \citenamefont {Long}, \citenamefont {Yordanov}, \citenamefont {Smith},
  \citenamefont {Barnes}, \citenamefont {Mertig},\ and\ \citenamefont
  {Arvidsson-Shukur}}]{dalton2022variational}%
  \BibitemOpen
  \bibfield  {author} {\bibinfo {author} {\bibfnamefont {K.}~\bibnamefont
  {Dalton}}, \bibinfo {author} {\bibfnamefont {C.~K.}\ \bibnamefont {Long}},
  \bibinfo {author} {\bibfnamefont {Y.~S.}\ \bibnamefont {Yordanov}}, \bibinfo
  {author} {\bibfnamefont {C.~G.}\ \bibnamefont {Smith}}, \bibinfo {author}
  {\bibfnamefont {C.~H.~W.}\ \bibnamefont {Barnes}}, \bibinfo {author}
  {\bibfnamefont {N.}~\bibnamefont {Mertig}},\ and\ \bibinfo {author}
  {\bibfnamefont {D.~R.~M.}\ \bibnamefont {Arvidsson-Shukur}},\ }\bibfield
  {title} {\bibinfo {title} {Quantifying the effect of gate errors on
  variational quantum eigensolvers for quantum chemistry},\ }\href
  {https://doi.org/10.1038/s41534-024-00808-x} {\bibfield  {journal} {\bibinfo
  {journal} {npj Quantum Information}\ }\textbf {\bibinfo {volume} {10}},\
  \bibinfo {pages} {18} (\bibinfo {year} {2024})}\BibitemShut {NoStop}%
\bibitem [{\citenamefont {Long}\ \emph {et~al.}(2024)\citenamefont {Long},
  \citenamefont {Dalton}, \citenamefont {Barnes}, \citenamefont
  {Arvidsson-Shukur},\ and\ \citenamefont {Mertig}}]{long2023layering}%
  \BibitemOpen
  \bibfield  {author} {\bibinfo {author} {\bibfnamefont {C.~K.}\ \bibnamefont
  {Long}}, \bibinfo {author} {\bibfnamefont {K.}~\bibnamefont {Dalton}},
  \bibinfo {author} {\bibfnamefont {C.~H.~W.}\ \bibnamefont {Barnes}}, \bibinfo
  {author} {\bibfnamefont {D.~R.~M.}\ \bibnamefont {Arvidsson-Shukur}},\ and\
  \bibinfo {author} {\bibfnamefont {N.}~\bibnamefont {Mertig}},\ }\bibfield
  {title} {\bibinfo {title} {Layering and subpool exploration for adaptive
  variational quantum eigensolvers: Reducing circuit depth, runtime, and
  susceptibility to noise},\ }\href
  {https://doi.org/10.1103/PhysRevA.109.042413} {\bibfield  {journal} {\bibinfo
   {journal} {Phys. Rev. A}\ }\textbf {\bibinfo {volume} {109}},\ \bibinfo
  {pages} {042413} (\bibinfo {year} {2024})}\BibitemShut {NoStop}%
\bibitem [{\citenamefont {Yanakiev}\ \emph {et~al.}(2024)\citenamefont
  {Yanakiev}, \citenamefont {Mertig}, \citenamefont {Long},\ and\ \citenamefont
  {Arvidsson-Shukur}}]{yanakiev2023dynamicadaptqaoa}%
  \BibitemOpen
  \bibfield  {author} {\bibinfo {author} {\bibfnamefont {N.}~\bibnamefont
  {Yanakiev}}, \bibinfo {author} {\bibfnamefont {N.}~\bibnamefont {Mertig}},
  \bibinfo {author} {\bibfnamefont {C.~K.}\ \bibnamefont {Long}},\ and\
  \bibinfo {author} {\bibfnamefont {D.~R.~M.}\ \bibnamefont
  {Arvidsson-Shukur}},\ }\bibfield  {title} {\bibinfo {title} {Dynamic adaptive
  quantum approximate optimization algorithm for shallow, noise-resilient
  circuits},\ }\href {https://doi.org/10.1103/PhysRevA.109.032420} {\bibfield
  {journal} {\bibinfo  {journal} {Phys. Rev. A}\ }\textbf {\bibinfo {volume}
  {109}},\ \bibinfo {pages} {032420} (\bibinfo {year} {2024})}\BibitemShut
  {NoStop}%
\bibitem [{\citenamefont {Magann}\ \emph {et~al.}(2021)\citenamefont {Magann},
  \citenamefont {Arenz}, \citenamefont {Grace}, \citenamefont {Ho},
  \citenamefont {Kosut}, \citenamefont {McClean}, \citenamefont {Rabitz},\ and\
  \citenamefont {Sarovar}}]{magannPulsesCircuitsBack2021a}%
  \BibitemOpen
  \bibfield  {author} {\bibinfo {author} {\bibfnamefont {A.~B.}\ \bibnamefont
  {Magann}}, \bibinfo {author} {\bibfnamefont {C.}~\bibnamefont {Arenz}},
  \bibinfo {author} {\bibfnamefont {M.~D.}\ \bibnamefont {Grace}}, \bibinfo
  {author} {\bibfnamefont {T.-S.}\ \bibnamefont {Ho}}, \bibinfo {author}
  {\bibfnamefont {R.~L.}\ \bibnamefont {Kosut}}, \bibinfo {author}
  {\bibfnamefont {J.~R.}\ \bibnamefont {McClean}}, \bibinfo {author}
  {\bibfnamefont {H.~A.}\ \bibnamefont {Rabitz}},\ and\ \bibinfo {author}
  {\bibfnamefont {M.}~\bibnamefont {Sarovar}},\ }\bibfield  {title} {\bibinfo
  {title} {From {{Pulses}} to {{Circuits}} and {{Back Again}}: {{A Quantum
  Optimal Control Perspective}} on {{Variational Quantum Algorithms}}},\ }\href
  {https://doi.org/10.1103/PRXQuantum.2.010101} {\bibfield  {journal} {\bibinfo
   {journal} {PRX Quantum}\ }\textbf {\bibinfo {volume} {2}},\ \bibinfo {pages}
  {010101} (\bibinfo {year} {2021})}\BibitemShut {NoStop}%
\bibitem [{\citenamefont {Lloyd}\ \emph {et~al.}(2020)\citenamefont {Lloyd},
  \citenamefont {Bosch}, \citenamefont {Palma}, \citenamefont {Kiani},
  \citenamefont {Liu}, \citenamefont {Marvian}, \citenamefont {Rebentrost},\
  and\ \citenamefont {Arvidsson-Shukur}}]{lloyd2020quantum}%
  \BibitemOpen
  \bibfield  {author} {\bibinfo {author} {\bibfnamefont {S.}~\bibnamefont
  {Lloyd}}, \bibinfo {author} {\bibfnamefont {S.}~\bibnamefont {Bosch}},
  \bibinfo {author} {\bibfnamefont {G.~D.}\ \bibnamefont {Palma}}, \bibinfo
  {author} {\bibfnamefont {B.}~\bibnamefont {Kiani}}, \bibinfo {author}
  {\bibfnamefont {Z.-W.}\ \bibnamefont {Liu}}, \bibinfo {author} {\bibfnamefont
  {M.}~\bibnamefont {Marvian}}, \bibinfo {author} {\bibfnamefont
  {P.}~\bibnamefont {Rebentrost}},\ and\ \bibinfo {author} {\bibfnamefont
  {D.~M.}\ \bibnamefont {Arvidsson-Shukur}},\ }\href@noop {} {\bibinfo {title}
  {Quantum polar decomposition algorithm}} (\bibinfo {year} {2020}),\ \Eprint
  {https://arxiv.org/abs/2006.00841} {arXiv:2006.00841 [quant-ph]} \BibitemShut
  {NoStop}%
\bibitem [{\citenamefont {Lloyd}\ \emph {et~al.}(2021)\citenamefont {Lloyd},
  \citenamefont {Kiani}, \citenamefont {Arvidsson-Shukur}, \citenamefont
  {Bosch}, \citenamefont {Palma}, \citenamefont {Kaminsky}, \citenamefont
  {Liu},\ and\ \citenamefont {Marvian}}]{lloyd2021hamiltonian}%
  \BibitemOpen
  \bibfield  {author} {\bibinfo {author} {\bibfnamefont {S.}~\bibnamefont
  {Lloyd}}, \bibinfo {author} {\bibfnamefont {B.~T.}\ \bibnamefont {Kiani}},
  \bibinfo {author} {\bibfnamefont {D.~R.~M.}\ \bibnamefont
  {Arvidsson-Shukur}}, \bibinfo {author} {\bibfnamefont {S.}~\bibnamefont
  {Bosch}}, \bibinfo {author} {\bibfnamefont {G.~D.}\ \bibnamefont {Palma}},
  \bibinfo {author} {\bibfnamefont {W.~M.}\ \bibnamefont {Kaminsky}}, \bibinfo
  {author} {\bibfnamefont {Z.-W.}\ \bibnamefont {Liu}},\ and\ \bibinfo {author}
  {\bibfnamefont {M.}~\bibnamefont {Marvian}},\ }\href@noop {} {\bibinfo
  {title} {Hamiltonian singular value transformation and inverse block
  encoding}} (\bibinfo {year} {2021}),\ \Eprint
  {https://arxiv.org/abs/2104.01410} {arXiv:2104.01410 [quant-ph]} \BibitemShut
  {NoStop}%
\bibitem [{\citenamefont {Asthana}\ \emph {et~al.}(2023)\citenamefont
  {Asthana}, \citenamefont {Liu}, \citenamefont {Meitei}, \citenamefont
  {Economou}, \citenamefont {Barnes},\ and\ \citenamefont
  {Mayhall}}]{asthanaMinimizingStatePreparation2022}%
  \BibitemOpen
  \bibfield  {author} {\bibinfo {author} {\bibfnamefont {A.}~\bibnamefont
  {Asthana}}, \bibinfo {author} {\bibfnamefont {C.}~\bibnamefont {Liu}},
  \bibinfo {author} {\bibfnamefont {O.~R.}\ \bibnamefont {Meitei}}, \bibinfo
  {author} {\bibfnamefont {S.~E.}\ \bibnamefont {Economou}}, \bibinfo {author}
  {\bibfnamefont {E.}~\bibnamefont {Barnes}},\ and\ \bibinfo {author}
  {\bibfnamefont {N.~J.}\ \bibnamefont {Mayhall}},\ }\bibfield  {title}
  {\bibinfo {title} {Leakage reduces device coherence demands for pulse-level
  molecular simulations},\ }\href
  {https://doi.org/10.1103/PhysRevApplied.19.064071} {\bibfield  {journal}
  {\bibinfo  {journal} {Phys. Rev. Appl.}\ }\textbf {\bibinfo {volume} {19}},\
  \bibinfo {pages} {064071} (\bibinfo {year} {2023})}\BibitemShut {NoStop}%
\bibitem [{\citenamefont {Sherbert}\ \emph {et~al.}(2024)\citenamefont
  {Sherbert}, \citenamefont {Amer}, \citenamefont {Economou}, \citenamefont
  {Barnes},\ and\ \citenamefont {Mayhall}}]{sherbert2024parameterization}%
  \BibitemOpen
  \bibfield  {author} {\bibinfo {author} {\bibfnamefont {K.~M.}\ \bibnamefont
  {Sherbert}}, \bibinfo {author} {\bibfnamefont {H.}~\bibnamefont {Amer}},
  \bibinfo {author} {\bibfnamefont {S.~E.}\ \bibnamefont {Economou}}, \bibinfo
  {author} {\bibfnamefont {E.}~\bibnamefont {Barnes}},\ and\ \bibinfo {author}
  {\bibfnamefont {N.~J.}\ \bibnamefont {Mayhall}},\ }\href@noop {} {\bibinfo
  {title} {Parameterization and optimizability of pulse-level vqes}} (\bibinfo
  {year} {2024}),\ \Eprint {https://arxiv.org/abs/2405.15166} {arXiv:2405.15166
  [quant-ph]} \BibitemShut {NoStop}%
\bibitem [{\citenamefont {Meitei}\ \emph {et~al.}(2021)\citenamefont {Meitei},
  \citenamefont {Gard}, \citenamefont {Barron}, \citenamefont {Pappas},
  \citenamefont {Economou}, \citenamefont {Barnes},\ and\ \citenamefont
  {Mayhall}}]{meiteiGatefreeStatePreparation2021a}%
  \BibitemOpen
  \bibfield  {author} {\bibinfo {author} {\bibfnamefont {O.~R.}\ \bibnamefont
  {Meitei}}, \bibinfo {author} {\bibfnamefont {B.~T.}\ \bibnamefont {Gard}},
  \bibinfo {author} {\bibfnamefont {G.~S.}\ \bibnamefont {Barron}}, \bibinfo
  {author} {\bibfnamefont {D.~P.}\ \bibnamefont {Pappas}}, \bibinfo {author}
  {\bibfnamefont {S.~E.}\ \bibnamefont {Economou}}, \bibinfo {author}
  {\bibfnamefont {E.}~\bibnamefont {Barnes}},\ and\ \bibinfo {author}
  {\bibfnamefont {N.~J.}\ \bibnamefont {Mayhall}},\ }\bibfield  {title}
  {\bibinfo {title} {Gate-free state preparation for fast variational quantum
  eigensolver simulations},\ }\href
  {https://doi.org/10.1038/s41534-021-00493-0} {\bibfield  {journal} {\bibinfo
  {journal} {npj Quantum Information}\ }\textbf {\bibinfo {volume} {7}},\
  \bibinfo {pages} {1} (\bibinfo {year} {2021})}\BibitemShut {NoStop}%
\bibitem [{\citenamefont {Egger}\ \emph {et~al.}(2023)\citenamefont {Egger},
  \citenamefont {Capecci}, \citenamefont {Pokharel}, \citenamefont
  {Barkoutsos}, \citenamefont {Fischer}, \citenamefont {Guidoni},\ and\
  \citenamefont {Tavernelli}}]{eggerStudyPulsebasedVariational2023}%
  \BibitemOpen
  \bibfield  {author} {\bibinfo {author} {\bibfnamefont {D.~J.}\ \bibnamefont
  {Egger}}, \bibinfo {author} {\bibfnamefont {C.}~\bibnamefont {Capecci}},
  \bibinfo {author} {\bibfnamefont {B.}~\bibnamefont {Pokharel}}, \bibinfo
  {author} {\bibfnamefont {P.~K.}\ \bibnamefont {Barkoutsos}}, \bibinfo
  {author} {\bibfnamefont {L.~E.}\ \bibnamefont {Fischer}}, \bibinfo {author}
  {\bibfnamefont {L.}~\bibnamefont {Guidoni}},\ and\ \bibinfo {author}
  {\bibfnamefont {I.}~\bibnamefont {Tavernelli}},\ }\bibfield  {title}
  {\bibinfo {title} {Pulse variational quantum eigensolver on
  cross-resonance-based hardware},\ }\href
  {https://doi.org/10.1103/PhysRevResearch.5.033159} {\bibfield  {journal}
  {\bibinfo  {journal} {Phys. Rev. Res.}\ }\textbf {\bibinfo {volume} {5}},\
  \bibinfo {pages} {033159} (\bibinfo {year} {2023})}\BibitemShut {NoStop}%
\bibitem [{\citenamefont {Ibrahim}\ \emph {et~al.}(2022)\citenamefont
  {Ibrahim}, \citenamefont {Mohammadbagherpoor}, \citenamefont {Rios},
  \citenamefont {Bronn},\ and\ \citenamefont
  {Byrd}}]{ibrahimEvaluationParameterizedQuantum2022}%
  \BibitemOpen
  \bibfield  {author} {\bibinfo {author} {\bibfnamefont {M.~M.}\ \bibnamefont
  {Ibrahim}}, \bibinfo {author} {\bibfnamefont {H.}~\bibnamefont
  {Mohammadbagherpoor}}, \bibinfo {author} {\bibfnamefont {C.}~\bibnamefont
  {Rios}}, \bibinfo {author} {\bibfnamefont {N.~T.}\ \bibnamefont {Bronn}},\
  and\ \bibinfo {author} {\bibfnamefont {G.~T.}\ \bibnamefont {Byrd}},\
  }\bibfield  {title} {\bibinfo {title} {Evaluation of {{Parameterized Quantum
  Circuits With Cross-Resonance Pulse-Driven Entanglers}}},\ }\href
  {https://doi.org/10.1109/TQE.2022.3231124} {\bibfield  {journal} {\bibinfo
  {journal} {IEEE Transactions on Quantum Engineering}\ }\textbf {\bibinfo
  {volume} {3}},\ \bibinfo {pages} {1} (\bibinfo {year} {2022})}\BibitemShut
  {NoStop}%
\bibitem [{\citenamefont {Meirom}\ and\ \citenamefont
  {Frankel}(2023)}]{meiromPANSATZPulsebasedAnsatz2023}%
  \BibitemOpen
  \bibfield  {author} {\bibinfo {author} {\bibfnamefont {D.}~\bibnamefont
  {Meirom}}\ and\ \bibinfo {author} {\bibfnamefont {S.~H.}\ \bibnamefont
  {Frankel}},\ }\bibfield  {title} {\bibinfo {title} {Pansatz: pulse-based
  ansatz for variational quantum algorithms},\ }\bibfield  {journal} {\bibinfo
  {journal} {Frontiers in Quantum Science and Technology}\ }\textbf {\bibinfo
  {volume} {2}},\ \href {https://doi.org/10.3389/frqst.2023.1273581}
  {10.3389/frqst.2023.1273581} (\bibinfo {year} {2023})\BibitemShut {NoStop}%
\bibitem [{\citenamefont {Yang}\ \emph {et~al.}(2017)\citenamefont {Yang},
  \citenamefont {Rahmani}, \citenamefont {Shabani}, \citenamefont {Neven},\
  and\ \citenamefont {Chamon}}]{yangOptimizingVariationalQuantum2017}%
  \BibitemOpen
  \bibfield  {author} {\bibinfo {author} {\bibfnamefont {Z.-C.}\ \bibnamefont
  {Yang}}, \bibinfo {author} {\bibfnamefont {A.}~\bibnamefont {Rahmani}},
  \bibinfo {author} {\bibfnamefont {A.}~\bibnamefont {Shabani}}, \bibinfo
  {author} {\bibfnamefont {H.}~\bibnamefont {Neven}},\ and\ \bibinfo {author}
  {\bibfnamefont {C.}~\bibnamefont {Chamon}},\ }\bibfield  {title} {\bibinfo
  {title} {Optimizing {{Variational Quantum Algorithms Using Pontryagin}}'s
  {{Minimum Principle}}},\ }\href {https://doi.org/10.1103/PhysRevX.7.021027}
  {\bibfield  {journal} {\bibinfo  {journal} {Physical Review X}\ }\textbf
  {\bibinfo {volume} {7}},\ \bibinfo {pages} {021027} (\bibinfo {year}
  {2017})}\BibitemShut {NoStop}%
\bibitem [{\citenamefont {Deffner}\ and\ \citenamefont
  {Campbell}(2017)}]{Deffner_2017}%
  \BibitemOpen
  \bibfield  {author} {\bibinfo {author} {\bibfnamefont {S.}~\bibnamefont
  {Deffner}}\ and\ \bibinfo {author} {\bibfnamefont {S.}~\bibnamefont
  {Campbell}},\ }\bibfield  {title} {\bibinfo {title} {Quantum speed limits:
  from heisenberg’s uncertainty principle to optimal quantum control},\
  }\href {https://doi.org/10.1088/1751-8121/aa86c6} {\bibfield  {journal}
  {\bibinfo  {journal} {Journal of Physics A: Mathematical and Theoretical}\
  }\textbf {\bibinfo {volume} {50}},\ \bibinfo {pages} {453001} (\bibinfo
  {year} {2017})}\BibitemShut {NoStop}%
\bibitem [{\citenamefont {Yoneda}\ \emph {et~al.}(2021)\citenamefont {Yoneda},
  \citenamefont {Huang}, \citenamefont {Feng}, \citenamefont {Yang},
  \citenamefont {Chan}, \citenamefont {Tanttu}, \citenamefont {Gilbert},
  \citenamefont {Leon}, \citenamefont {Hudson}, \citenamefont {Itoh},
  \citenamefont {Morello}, \citenamefont {Bartlett}, \citenamefont {Laucht},
  \citenamefont {Saraiva},\ and\ \citenamefont {Dzurak}}]{Yoneda2021}%
  \BibitemOpen
  \bibfield  {author} {\bibinfo {author} {\bibfnamefont {J.}~\bibnamefont
  {Yoneda}}, \bibinfo {author} {\bibfnamefont {W.}~\bibnamefont {Huang}},
  \bibinfo {author} {\bibfnamefont {M.}~\bibnamefont {Feng}}, \bibinfo {author}
  {\bibfnamefont {C.~H.}\ \bibnamefont {Yang}}, \bibinfo {author}
  {\bibfnamefont {K.~W.}\ \bibnamefont {Chan}}, \bibinfo {author}
  {\bibfnamefont {T.}~\bibnamefont {Tanttu}}, \bibinfo {author} {\bibfnamefont
  {W.}~\bibnamefont {Gilbert}}, \bibinfo {author} {\bibfnamefont {R.~C.~C.}\
  \bibnamefont {Leon}}, \bibinfo {author} {\bibfnamefont {F.~E.}\ \bibnamefont
  {Hudson}}, \bibinfo {author} {\bibfnamefont {K.~M.}\ \bibnamefont {Itoh}},
  \bibinfo {author} {\bibfnamefont {A.}~\bibnamefont {Morello}}, \bibinfo
  {author} {\bibfnamefont {S.~D.}\ \bibnamefont {Bartlett}}, \bibinfo {author}
  {\bibfnamefont {A.}~\bibnamefont {Laucht}}, \bibinfo {author} {\bibfnamefont
  {A.}~\bibnamefont {Saraiva}},\ and\ \bibinfo {author} {\bibfnamefont {A.~S.}\
  \bibnamefont {Dzurak}},\ }\bibfield  {title} {\bibinfo {title} {Coherent spin
  qubit transport in silicon},\ }\href
  {https://doi.org/10.1038/s41467-021-24371-7} {\bibfield  {journal} {\bibinfo
  {journal} {Nature Communications}\ }\textbf {\bibinfo {volume} {12}},\
  \bibinfo {pages} {4114} (\bibinfo {year} {2021})}\BibitemShut {NoStop}%
\bibitem [{\citenamefont {Huang}\ \emph {et~al.}(2024)\citenamefont {Huang},
  \citenamefont {Su}, \citenamefont {Lim}, \citenamefont {Feng}, \citenamefont
  {van Straaten}, \citenamefont {Severin}, \citenamefont {Gilbert},
  \citenamefont {Dumoulin~Stuyck}, \citenamefont {Tanttu}, \citenamefont
  {Serrano}, \citenamefont {Cifuentes}, \citenamefont {Hansen}, \citenamefont
  {Seedhouse}, \citenamefont {Vahapoglu}, \citenamefont {Leon}, \citenamefont
  {Abrosimov}, \citenamefont {Pohl}, \citenamefont {Thewalt}, \citenamefont
  {Hudson}, \citenamefont {Escott}, \citenamefont {Ares}, \citenamefont
  {Bartlett}, \citenamefont {Morello}, \citenamefont {Saraiva}, \citenamefont
  {Laucht}, \citenamefont {Dzurak},\ and\ \citenamefont {Yang}}]{Huang2024}%
  \BibitemOpen
  \bibfield  {author} {\bibinfo {author} {\bibfnamefont {J.~Y.}\ \bibnamefont
  {Huang}}, \bibinfo {author} {\bibfnamefont {R.~Y.}\ \bibnamefont {Su}},
  \bibinfo {author} {\bibfnamefont {W.~H.}\ \bibnamefont {Lim}}, \bibinfo
  {author} {\bibfnamefont {M.}~\bibnamefont {Feng}}, \bibinfo {author}
  {\bibfnamefont {B.}~\bibnamefont {van Straaten}}, \bibinfo {author}
  {\bibfnamefont {B.}~\bibnamefont {Severin}}, \bibinfo {author} {\bibfnamefont
  {W.}~\bibnamefont {Gilbert}}, \bibinfo {author} {\bibfnamefont
  {N.}~\bibnamefont {Dumoulin~Stuyck}}, \bibinfo {author} {\bibfnamefont
  {T.}~\bibnamefont {Tanttu}}, \bibinfo {author} {\bibfnamefont
  {S.}~\bibnamefont {Serrano}}, \bibinfo {author} {\bibfnamefont {J.~D.}\
  \bibnamefont {Cifuentes}}, \bibinfo {author} {\bibfnamefont {I.}~\bibnamefont
  {Hansen}}, \bibinfo {author} {\bibfnamefont {A.~E.}\ \bibnamefont
  {Seedhouse}}, \bibinfo {author} {\bibfnamefont {E.}~\bibnamefont
  {Vahapoglu}}, \bibinfo {author} {\bibfnamefont {R.~C.~C.}\ \bibnamefont
  {Leon}}, \bibinfo {author} {\bibfnamefont {N.~V.}\ \bibnamefont {Abrosimov}},
  \bibinfo {author} {\bibfnamefont {H.-J.}\ \bibnamefont {Pohl}}, \bibinfo
  {author} {\bibfnamefont {M.~L.~W.}\ \bibnamefont {Thewalt}}, \bibinfo
  {author} {\bibfnamefont {F.~E.}\ \bibnamefont {Hudson}}, \bibinfo {author}
  {\bibfnamefont {C.~C.}\ \bibnamefont {Escott}}, \bibinfo {author}
  {\bibfnamefont {N.}~\bibnamefont {Ares}}, \bibinfo {author} {\bibfnamefont
  {S.~D.}\ \bibnamefont {Bartlett}}, \bibinfo {author} {\bibfnamefont
  {A.}~\bibnamefont {Morello}}, \bibinfo {author} {\bibfnamefont
  {A.}~\bibnamefont {Saraiva}}, \bibinfo {author} {\bibfnamefont
  {A.}~\bibnamefont {Laucht}}, \bibinfo {author} {\bibfnamefont {A.~S.}\
  \bibnamefont {Dzurak}},\ and\ \bibinfo {author} {\bibfnamefont {C.~H.}\
  \bibnamefont {Yang}},\ }\bibfield  {title} {\bibinfo {title} {High-fidelity
  spin qubit operation and algorithmic initialization above 1 k},\ }\href
  {https://doi.org/10.1038/s41586-024-07160-2} {\bibfield  {journal} {\bibinfo
  {journal} {Nature}\ }\textbf {\bibinfo {volume} {627}},\ \bibinfo {pages}
  {772} (\bibinfo {year} {2024})}\BibitemShut {NoStop}%
\bibitem [{\citenamefont {Philips}\ \emph {et~al.}(2022)\citenamefont
  {Philips}, \citenamefont {Mądzik}, \citenamefont {Amitonov}, \citenamefont
  {de~Snoo}, \citenamefont {Russ}, \citenamefont {Kalhor}, \citenamefont
  {Volk}, \citenamefont {Lawrie}, \citenamefont {Brousse}, \citenamefont
  {Tryputen}, \citenamefont {Wuetz}, \citenamefont {Sammak}, \citenamefont
  {Veldhorst}, \citenamefont {Scappucci},\ and\ \citenamefont
  {Vandersypen}}]{Philips2022}%
  \BibitemOpen
  \bibfield  {author} {\bibinfo {author} {\bibfnamefont {S.~G.~J.}\
  \bibnamefont {Philips}}, \bibinfo {author} {\bibfnamefont {M.~T.}\
  \bibnamefont {Mądzik}}, \bibinfo {author} {\bibfnamefont {S.~V.}\
  \bibnamefont {Amitonov}}, \bibinfo {author} {\bibfnamefont {S.~L.}\
  \bibnamefont {de~Snoo}}, \bibinfo {author} {\bibfnamefont {M.}~\bibnamefont
  {Russ}}, \bibinfo {author} {\bibfnamefont {N.}~\bibnamefont {Kalhor}},
  \bibinfo {author} {\bibfnamefont {C.}~\bibnamefont {Volk}}, \bibinfo {author}
  {\bibfnamefont {W.~I.~L.}\ \bibnamefont {Lawrie}}, \bibinfo {author}
  {\bibfnamefont {D.}~\bibnamefont {Brousse}}, \bibinfo {author} {\bibfnamefont
  {L.}~\bibnamefont {Tryputen}}, \bibinfo {author} {\bibfnamefont {B.~P.}\
  \bibnamefont {Wuetz}}, \bibinfo {author} {\bibfnamefont {A.}~\bibnamefont
  {Sammak}}, \bibinfo {author} {\bibfnamefont {M.}~\bibnamefont {Veldhorst}},
  \bibinfo {author} {\bibfnamefont {G.}~\bibnamefont {Scappucci}},\ and\
  \bibinfo {author} {\bibfnamefont {L.~M.~K.}\ \bibnamefont {Vandersypen}},\
  }\bibfield  {title} {\bibinfo {title} {Universal control of a six-qubit
  quantum processor in silicon},\ }\href
  {https://doi.org/10.1038/s41586-022-05117-x} {\bibfield  {journal} {\bibinfo
  {journal} {Nature}\ }\textbf {\bibinfo {volume} {609}},\ \bibinfo {pages}
  {919} (\bibinfo {year} {2022})}\BibitemShut {NoStop}%
\bibitem [{\citenamefont {Takeda}\ \emph {et~al.}(2016)\citenamefont {Takeda},
  \citenamefont {Kamioka}, \citenamefont {Otsuka}, \citenamefont {Yoneda},
  \citenamefont {Nakajima}, \citenamefont {Delbecq}, \citenamefont {Amaha},
  \citenamefont {Allison}, \citenamefont {Kodera}, \citenamefont {Oda},\ and\
  \citenamefont {Tarucha}}]{doi:10.1126/sciadv.1600694}%
  \BibitemOpen
  \bibfield  {author} {\bibinfo {author} {\bibfnamefont {K.}~\bibnamefont
  {Takeda}}, \bibinfo {author} {\bibfnamefont {J.}~\bibnamefont {Kamioka}},
  \bibinfo {author} {\bibfnamefont {T.}~\bibnamefont {Otsuka}}, \bibinfo
  {author} {\bibfnamefont {J.}~\bibnamefont {Yoneda}}, \bibinfo {author}
  {\bibfnamefont {T.}~\bibnamefont {Nakajima}}, \bibinfo {author}
  {\bibfnamefont {M.~R.}\ \bibnamefont {Delbecq}}, \bibinfo {author}
  {\bibfnamefont {S.}~\bibnamefont {Amaha}}, \bibinfo {author} {\bibfnamefont
  {G.}~\bibnamefont {Allison}}, \bibinfo {author} {\bibfnamefont
  {T.}~\bibnamefont {Kodera}}, \bibinfo {author} {\bibfnamefont
  {S.}~\bibnamefont {Oda}},\ and\ \bibinfo {author} {\bibfnamefont
  {S.}~\bibnamefont {Tarucha}},\ }\bibfield  {title} {\bibinfo {title} {A
  fault-tolerant addressable spin qubit in a natural silicon quantum dot},\
  }\href {https://doi.org/10.1126/sciadv.1600694} {\bibfield  {journal}
  {\bibinfo  {journal} {Science Advances}\ }\textbf {\bibinfo {volume} {2}},\
  \bibinfo {pages} {e1600694} (\bibinfo {year} {2016})},\ \Eprint
  {https://arxiv.org/abs/https://www.science.org/doi/pdf/10.1126/sciadv.1600694}
  {https://www.science.org/doi/pdf/10.1126/sciadv.1600694} \BibitemShut
  {NoStop}%
\bibitem [{\citenamefont {Stuyck}\ \emph {et~al.}(2021)\citenamefont {Stuyck},
  \citenamefont {Li}, \citenamefont {Godfrin}, \citenamefont {Elsayed},
  \citenamefont {Kubicek}, \citenamefont {Jussot}, \citenamefont {Chan},
  \citenamefont {Mohiyaddin}, \citenamefont {Shehata}, \citenamefont {Simion},
  \citenamefont {Canvel}, \citenamefont {Goux}, \citenamefont {Heyns},
  \citenamefont {Govoreanu},\ and\ \citenamefont {Radu}}]{t_values}%
  \BibitemOpen
  \bibfield  {author} {\bibinfo {author} {\bibfnamefont {N.~I.~D.}\
  \bibnamefont {Stuyck}}, \bibinfo {author} {\bibfnamefont {R.}~\bibnamefont
  {Li}}, \bibinfo {author} {\bibfnamefont {C.}~\bibnamefont {Godfrin}},
  \bibinfo {author} {\bibfnamefont {A.}~\bibnamefont {Elsayed}}, \bibinfo
  {author} {\bibfnamefont {S.}~\bibnamefont {Kubicek}}, \bibinfo {author}
  {\bibfnamefont {J.}~\bibnamefont {Jussot}}, \bibinfo {author} {\bibfnamefont
  {B.~T.}\ \bibnamefont {Chan}}, \bibinfo {author} {\bibfnamefont {F.~A.}\
  \bibnamefont {Mohiyaddin}}, \bibinfo {author} {\bibfnamefont
  {M.}~\bibnamefont {Shehata}}, \bibinfo {author} {\bibfnamefont
  {G.}~\bibnamefont {Simion}}, \bibinfo {author} {\bibfnamefont
  {Y.}~\bibnamefont {Canvel}}, \bibinfo {author} {\bibfnamefont
  {L.}~\bibnamefont {Goux}}, \bibinfo {author} {\bibfnamefont {M.}~\bibnamefont
  {Heyns}}, \bibinfo {author} {\bibfnamefont {B.}~\bibnamefont {Govoreanu}},\
  and\ \bibinfo {author} {\bibfnamefont {I.~P.}\ \bibnamefont {Radu}},\
  }\bibfield  {title} {\bibinfo {title} {Uniform spin qubit devices with
  tunable coupling in an all-silicon 300 mm integrated process},\ }in\ \href
  {https://doi.org/10.23919/VLSICircuits52068.2021.9492427} {\emph {\bibinfo
  {booktitle} {2021 Symposium on VLSI Circuits}}}\ (\bibinfo {year} {2021})\
  pp.\ \bibinfo {pages} {1--2}\BibitemShut {NoStop}%
\bibitem [{\citenamefont {Weinstein}\ \emph {et~al.}(2023)\citenamefont
  {Weinstein}, \citenamefont {Reed}, \citenamefont {Jones}, \citenamefont
  {Andrews}, \citenamefont {Barnes}, \citenamefont {Blumoff}, \citenamefont
  {Euliss}, \citenamefont {Eng}, \citenamefont {Fong}, \citenamefont {Ha},
  \citenamefont {Hulbert}, \citenamefont {Jackson}, \citenamefont {Jura},
  \citenamefont {Keating}, \citenamefont {Kerckhoff}, \citenamefont {Kiselev},
  \citenamefont {Matten}, \citenamefont {Sabbir}, \citenamefont {Smith},
  \citenamefont {Wright}, \citenamefont {Rakher}, \citenamefont {Ladd},\ and\
  \citenamefont {Borselli}}]{Weinstein2023Mar}%
  \BibitemOpen
  \bibfield  {author} {\bibinfo {author} {\bibfnamefont {A.~J.}\ \bibnamefont
  {Weinstein}}, \bibinfo {author} {\bibfnamefont {M.~D.}\ \bibnamefont {Reed}},
  \bibinfo {author} {\bibfnamefont {A.~M.}\ \bibnamefont {Jones}}, \bibinfo
  {author} {\bibfnamefont {R.~W.}\ \bibnamefont {Andrews}}, \bibinfo {author}
  {\bibfnamefont {D.}~\bibnamefont {Barnes}}, \bibinfo {author} {\bibfnamefont
  {J.~Z.}\ \bibnamefont {Blumoff}}, \bibinfo {author} {\bibfnamefont {L.~E.}\
  \bibnamefont {Euliss}}, \bibinfo {author} {\bibfnamefont {K.}~\bibnamefont
  {Eng}}, \bibinfo {author} {\bibfnamefont {B.~H.}\ \bibnamefont {Fong}},
  \bibinfo {author} {\bibfnamefont {S.~D.}\ \bibnamefont {Ha}}, \bibinfo
  {author} {\bibfnamefont {D.~R.}\ \bibnamefont {Hulbert}}, \bibinfo {author}
  {\bibfnamefont {C.~A.~C.}\ \bibnamefont {Jackson}}, \bibinfo {author}
  {\bibfnamefont {M.}~\bibnamefont {Jura}}, \bibinfo {author} {\bibfnamefont
  {T.~E.}\ \bibnamefont {Keating}}, \bibinfo {author} {\bibfnamefont
  {J.}~\bibnamefont {Kerckhoff}}, \bibinfo {author} {\bibfnamefont {A.~A.}\
  \bibnamefont {Kiselev}}, \bibinfo {author} {\bibfnamefont {J.}~\bibnamefont
  {Matten}}, \bibinfo {author} {\bibfnamefont {G.}~\bibnamefont {Sabbir}},
  \bibinfo {author} {\bibfnamefont {A.}~\bibnamefont {Smith}}, \bibinfo
  {author} {\bibfnamefont {J.}~\bibnamefont {Wright}}, \bibinfo {author}
  {\bibfnamefont {M.~T.}\ \bibnamefont {Rakher}}, \bibinfo {author}
  {\bibfnamefont {T.~D.}\ \bibnamefont {Ladd}},\ and\ \bibinfo {author}
  {\bibfnamefont {M.~G.}\ \bibnamefont {Borselli}},\ }\bibfield  {title}
  {\bibinfo {title} {{Universal logic with encoded spin qubits in silicon}},\
  }\href {https://doi.org/10.1038/s41586-023-05777-3} {\bibfield  {journal}
  {\bibinfo  {journal} {Nature}\ }\textbf {\bibinfo {volume} {615}},\ \bibinfo
  {pages} {817} (\bibinfo {year} {2023})}\BibitemShut {NoStop}%
\bibitem [{\citenamefont {Simmons}\ \emph {et~al.}(2009)\citenamefont
  {Simmons}, \citenamefont {Thalakulam}, \citenamefont {Rosemeyer},
  \citenamefont {Van~Bael}, \citenamefont {Sackmann}, \citenamefont {Savage},
  \citenamefont {Lagally}, \citenamefont {Joynt}, \citenamefont {Friesen},
  \citenamefont {Coppersmith},\ and\ \citenamefont {Eriksson}}]{Simmons2009}%
  \BibitemOpen
  \bibfield  {author} {\bibinfo {author} {\bibfnamefont {C.~B.}\ \bibnamefont
  {Simmons}}, \bibinfo {author} {\bibfnamefont {M.}~\bibnamefont {Thalakulam}},
  \bibinfo {author} {\bibfnamefont {B.~M.}\ \bibnamefont {Rosemeyer}}, \bibinfo
  {author} {\bibfnamefont {B.~J.}\ \bibnamefont {Van~Bael}}, \bibinfo {author}
  {\bibfnamefont {E.~K.}\ \bibnamefont {Sackmann}}, \bibinfo {author}
  {\bibfnamefont {D.~E.}\ \bibnamefont {Savage}}, \bibinfo {author}
  {\bibfnamefont {M.~G.}\ \bibnamefont {Lagally}}, \bibinfo {author}
  {\bibfnamefont {R.}~\bibnamefont {Joynt}}, \bibinfo {author} {\bibfnamefont
  {M.}~\bibnamefont {Friesen}}, \bibinfo {author} {\bibfnamefont {S.~N.}\
  \bibnamefont {Coppersmith}},\ and\ \bibinfo {author} {\bibfnamefont {M.~A.}\
  \bibnamefont {Eriksson}},\ }\bibfield  {title} {\bibinfo {title} {Charge
  sensing and controllable tunnel coupling in a si/sige double quantum dot},\
  }\href {https://doi.org/10.1021/nl9014974} {\bibfield  {journal} {\bibinfo
  {journal} {Nano Letters}\ }\textbf {\bibinfo {volume} {9}},\ \bibinfo {pages}
  {3234} (\bibinfo {year} {2009})}\BibitemShut {NoStop}%
\bibitem [{\citenamefont {Loss}\ and\ \citenamefont
  {DiVincenzo}(1998)}]{PhysRevA.57.120}%
  \BibitemOpen
  \bibfield  {author} {\bibinfo {author} {\bibfnamefont {D.}~\bibnamefont
  {Loss}}\ and\ \bibinfo {author} {\bibfnamefont {D.~P.}\ \bibnamefont
  {DiVincenzo}},\ }\bibfield  {title} {\bibinfo {title} {Quantum computation
  with quantum dots},\ }\href {https://doi.org/10.1103/PhysRevA.57.120}
  {\bibfield  {journal} {\bibinfo  {journal} {Phys. Rev. A}\ }\textbf {\bibinfo
  {volume} {57}},\ \bibinfo {pages} {120} (\bibinfo {year} {1998})}\BibitemShut
  {NoStop}%
\bibitem [{\citenamefont {Burkard}\ \emph {et~al.}(1999)\citenamefont
  {Burkard}, \citenamefont {Loss},\ and\ \citenamefont
  {DiVincenzo}}]{PhysRevB.59.2070}%
  \BibitemOpen
  \bibfield  {author} {\bibinfo {author} {\bibfnamefont {G.}~\bibnamefont
  {Burkard}}, \bibinfo {author} {\bibfnamefont {D.}~\bibnamefont {Loss}},\ and\
  \bibinfo {author} {\bibfnamefont {D.~P.}\ \bibnamefont {DiVincenzo}},\
  }\bibfield  {title} {\bibinfo {title} {Coupled quantum dots as quantum
  gates},\ }\href {https://doi.org/10.1103/PhysRevB.59.2070} {\bibfield
  {journal} {\bibinfo  {journal} {Phys. Rev. B}\ }\textbf {\bibinfo {volume}
  {59}},\ \bibinfo {pages} {2070} (\bibinfo {year} {1999})}\BibitemShut
  {NoStop}%
\bibitem [{\citenamefont {Zwanenburg}\ \emph {et~al.}(2013)\citenamefont
  {Zwanenburg}, \citenamefont {Dzurak}, \citenamefont {Morello}, \citenamefont
  {Simmons}, \citenamefont {Hollenberg}, \citenamefont {Klimeck}, \citenamefont
  {Rogge}, \citenamefont {Coppersmith},\ and\ \citenamefont
  {Eriksson}}]{RevModPhys.85.961}%
  \BibitemOpen
  \bibfield  {author} {\bibinfo {author} {\bibfnamefont {F.~A.}\ \bibnamefont
  {Zwanenburg}}, \bibinfo {author} {\bibfnamefont {A.~S.}\ \bibnamefont
  {Dzurak}}, \bibinfo {author} {\bibfnamefont {A.}~\bibnamefont {Morello}},
  \bibinfo {author} {\bibfnamefont {M.~Y.}\ \bibnamefont {Simmons}}, \bibinfo
  {author} {\bibfnamefont {L.~C.~L.}\ \bibnamefont {Hollenberg}}, \bibinfo
  {author} {\bibfnamefont {G.}~\bibnamefont {Klimeck}}, \bibinfo {author}
  {\bibfnamefont {S.}~\bibnamefont {Rogge}}, \bibinfo {author} {\bibfnamefont
  {S.~N.}\ \bibnamefont {Coppersmith}},\ and\ \bibinfo {author} {\bibfnamefont
  {M.~A.}\ \bibnamefont {Eriksson}},\ }\bibfield  {title} {\bibinfo {title}
  {Silicon quantum electronics},\ }\href
  {https://doi.org/10.1103/RevModPhys.85.961} {\bibfield  {journal} {\bibinfo
  {journal} {Rev. Mod. Phys.}\ }\textbf {\bibinfo {volume} {85}},\ \bibinfo
  {pages} {961} (\bibinfo {year} {2013})}\BibitemShut {NoStop}%
\bibitem [{\citenamefont {Burkard}\ \emph {et~al.}(2023)\citenamefont
  {Burkard}, \citenamefont {Ladd}, \citenamefont {Pan}, \citenamefont
  {Nichol},\ and\ \citenamefont {Petta}}]{RevModPhys.95.025003}%
  \BibitemOpen
  \bibfield  {author} {\bibinfo {author} {\bibfnamefont {G.}~\bibnamefont
  {Burkard}}, \bibinfo {author} {\bibfnamefont {T.~D.}\ \bibnamefont {Ladd}},
  \bibinfo {author} {\bibfnamefont {A.}~\bibnamefont {Pan}}, \bibinfo {author}
  {\bibfnamefont {J.~M.}\ \bibnamefont {Nichol}},\ and\ \bibinfo {author}
  {\bibfnamefont {J.~R.}\ \bibnamefont {Petta}},\ }\bibfield  {title} {\bibinfo
  {title} {Semiconductor spin qubits},\ }\href
  {https://doi.org/10.1103/RevModPhys.95.025003} {\bibfield  {journal}
  {\bibinfo  {journal} {Rev. Mod. Phys.}\ }\textbf {\bibinfo {volume} {95}},\
  \bibinfo {pages} {025003} (\bibinfo {year} {2023})}\BibitemShut {NoStop}%
\bibitem [{\citenamefont {Zwerver}\ \emph {et~al.}(2022)\citenamefont
  {Zwerver}, \citenamefont {Kr{\"a}henmann}, \citenamefont {Watson},
  \citenamefont {Lampert}, \citenamefont {George}, \citenamefont
  {Pillarisetty}, \citenamefont {Bojarski}, \citenamefont {Amin}, \citenamefont
  {Amitonov}, \citenamefont {Boter}, \citenamefont {Caudillo}, \citenamefont
  {Correas-Serrano}, \citenamefont {Dehollain}, \citenamefont {Droulers},
  \citenamefont {Henry}, \citenamefont {Kotlyar}, \citenamefont {Lodari},
  \citenamefont {L{\"u}thi}, \citenamefont {Michalak}, \citenamefont {Mueller},
  \citenamefont {Neyens}, \citenamefont {Roberts}, \citenamefont {Samkharadze},
  \citenamefont {Zheng}, \citenamefont {Zietz}, \citenamefont {Scappucci},
  \citenamefont {Veldhorst}, \citenamefont {Vandersypen},\ and\ \citenamefont
  {Clarke}}]{Zwerver2022}%
  \BibitemOpen
  \bibfield  {author} {\bibinfo {author} {\bibfnamefont {A.~M.~J.}\
  \bibnamefont {Zwerver}}, \bibinfo {author} {\bibfnamefont {T.}~\bibnamefont
  {Kr{\"a}henmann}}, \bibinfo {author} {\bibfnamefont {T.~F.}\ \bibnamefont
  {Watson}}, \bibinfo {author} {\bibfnamefont {L.}~\bibnamefont {Lampert}},
  \bibinfo {author} {\bibfnamefont {H.~C.}\ \bibnamefont {George}}, \bibinfo
  {author} {\bibfnamefont {R.}~\bibnamefont {Pillarisetty}}, \bibinfo {author}
  {\bibfnamefont {S.~A.}\ \bibnamefont {Bojarski}}, \bibinfo {author}
  {\bibfnamefont {P.}~\bibnamefont {Amin}}, \bibinfo {author} {\bibfnamefont
  {S.~V.}\ \bibnamefont {Amitonov}}, \bibinfo {author} {\bibfnamefont {J.~M.}\
  \bibnamefont {Boter}}, \bibinfo {author} {\bibfnamefont {R.}~\bibnamefont
  {Caudillo}}, \bibinfo {author} {\bibfnamefont {D.}~\bibnamefont
  {Correas-Serrano}}, \bibinfo {author} {\bibfnamefont {J.~P.}\ \bibnamefont
  {Dehollain}}, \bibinfo {author} {\bibfnamefont {G.}~\bibnamefont {Droulers}},
  \bibinfo {author} {\bibfnamefont {E.~M.}\ \bibnamefont {Henry}}, \bibinfo
  {author} {\bibfnamefont {R.}~\bibnamefont {Kotlyar}}, \bibinfo {author}
  {\bibfnamefont {M.}~\bibnamefont {Lodari}}, \bibinfo {author} {\bibfnamefont
  {F.}~\bibnamefont {L{\"u}thi}}, \bibinfo {author} {\bibfnamefont {D.~J.}\
  \bibnamefont {Michalak}}, \bibinfo {author} {\bibfnamefont {B.~K.}\
  \bibnamefont {Mueller}}, \bibinfo {author} {\bibfnamefont {S.}~\bibnamefont
  {Neyens}}, \bibinfo {author} {\bibfnamefont {J.}~\bibnamefont {Roberts}},
  \bibinfo {author} {\bibfnamefont {N.}~\bibnamefont {Samkharadze}}, \bibinfo
  {author} {\bibfnamefont {G.}~\bibnamefont {Zheng}}, \bibinfo {author}
  {\bibfnamefont {O.~K.}\ \bibnamefont {Zietz}}, \bibinfo {author}
  {\bibfnamefont {G.}~\bibnamefont {Scappucci}}, \bibinfo {author}
  {\bibfnamefont {M.}~\bibnamefont {Veldhorst}}, \bibinfo {author}
  {\bibfnamefont {L.~M.~K.}\ \bibnamefont {Vandersypen}},\ and\ \bibinfo
  {author} {\bibfnamefont {J.~S.}\ \bibnamefont {Clarke}},\ }\bibfield  {title}
  {\bibinfo {title} {Qubits made by advanced semiconductor manufacturing},\
  }\href {https://doi.org/10.1038/s41928-022-00727-9} {\bibfield  {journal}
  {\bibinfo  {journal} {Nature Electronics}\ }\textbf {\bibinfo {volume} {5}},\
  \bibinfo {pages} {184} (\bibinfo {year} {2022})}\BibitemShut {NoStop}%
\bibitem [{\citenamefont {Elsayed}\ \emph {et~al.}(2023)\citenamefont
  {Elsayed}, \citenamefont {Godfrin}, \citenamefont {Dumoulin~Stuyck},
  \citenamefont {Shehata}, \citenamefont {Kubicek}, \citenamefont {Massar},
  \citenamefont {Canvel}, \citenamefont {Jussot}, \citenamefont {Hikavyy},
  \citenamefont {Loo}, \citenamefont {Simion}, \citenamefont {Mongillo},
  \citenamefont {Wan}, \citenamefont {Govoreanu}, \citenamefont {Li},
  \citenamefont {Radu}, \citenamefont {Van~Dorpe},\ and\ \citenamefont
  {De~Greve}}]{10185272}%
  \BibitemOpen
  \bibfield  {author} {\bibinfo {author} {\bibfnamefont {A.}~\bibnamefont
  {Elsayed}}, \bibinfo {author} {\bibfnamefont {C.}~\bibnamefont {Godfrin}},
  \bibinfo {author} {\bibfnamefont {N.}~\bibnamefont {Dumoulin~Stuyck}},
  \bibinfo {author} {\bibfnamefont {M.}~\bibnamefont {Shehata}}, \bibinfo
  {author} {\bibfnamefont {S.}~\bibnamefont {Kubicek}}, \bibinfo {author}
  {\bibfnamefont {S.}~\bibnamefont {Massar}}, \bibinfo {author} {\bibfnamefont
  {Y.}~\bibnamefont {Canvel}}, \bibinfo {author} {\bibfnamefont
  {J.}~\bibnamefont {Jussot}}, \bibinfo {author} {\bibfnamefont
  {A.}~\bibnamefont {Hikavyy}}, \bibinfo {author} {\bibfnamefont
  {R.}~\bibnamefont {Loo}}, \bibinfo {author} {\bibfnamefont {G.}~\bibnamefont
  {Simion}}, \bibinfo {author} {\bibfnamefont {M.}~\bibnamefont {Mongillo}},
  \bibinfo {author} {\bibfnamefont {D.}~\bibnamefont {Wan}}, \bibinfo {author}
  {\bibfnamefont {B.}~\bibnamefont {Govoreanu}}, \bibinfo {author}
  {\bibfnamefont {R.}~\bibnamefont {Li}}, \bibinfo {author} {\bibfnamefont
  {I.}~\bibnamefont {Radu}}, \bibinfo {author} {\bibfnamefont {P.}~\bibnamefont
  {Van~Dorpe}},\ and\ \bibinfo {author} {\bibfnamefont {K.}~\bibnamefont
  {De~Greve}},\ }\bibfield  {title} {\bibinfo {title} {Comprehensive 300 mm
  process for silicon spin qubits with modular integration},\ }in\ \href
  {https://doi.org/10.23919/VLSITechnologyandCir57934.2023.10185272} {\emph
  {\bibinfo {booktitle} {2023 IEEE Symposium on VLSI Technology and Circuits
  (VLSI Technology and Circuits)}}}\ (\bibinfo {year} {2023})\ pp.\ \bibinfo
  {pages} {1--2}\BibitemShut {NoStop}%
\bibitem [{\citenamefont {Li}\ \emph {et~al.}(2020)\citenamefont {Li},
  \citenamefont {Stuyck}, \citenamefont {Kubicek}, \citenamefont {Jussot},
  \citenamefont {Chan}, \citenamefont {Mohiyaddin}, \citenamefont {Elsayed},
  \citenamefont {Shehata}, \citenamefont {Simion}, \citenamefont {Godfrin},
  \citenamefont {Canvel}, \citenamefont {Ivanov}, \citenamefont {Goux},
  \citenamefont {Govoreanu},\ and\ \citenamefont {Radu}}]{9371956}%
  \BibitemOpen
  \bibfield  {author} {\bibinfo {author} {\bibfnamefont {R.}~\bibnamefont
  {Li}}, \bibinfo {author} {\bibfnamefont {N.~I.~D.}\ \bibnamefont {Stuyck}},
  \bibinfo {author} {\bibfnamefont {S.}~\bibnamefont {Kubicek}}, \bibinfo
  {author} {\bibfnamefont {J.}~\bibnamefont {Jussot}}, \bibinfo {author}
  {\bibfnamefont {B.~T.}\ \bibnamefont {Chan}}, \bibinfo {author}
  {\bibfnamefont {F.~A.}\ \bibnamefont {Mohiyaddin}}, \bibinfo {author}
  {\bibfnamefont {A.}~\bibnamefont {Elsayed}}, \bibinfo {author} {\bibfnamefont
  {M.}~\bibnamefont {Shehata}}, \bibinfo {author} {\bibfnamefont
  {G.}~\bibnamefont {Simion}}, \bibinfo {author} {\bibfnamefont
  {C.}~\bibnamefont {Godfrin}}, \bibinfo {author} {\bibfnamefont
  {Y.}~\bibnamefont {Canvel}}, \bibinfo {author} {\bibfnamefont
  {T.}~\bibnamefont {Ivanov}}, \bibinfo {author} {\bibfnamefont
  {L.}~\bibnamefont {Goux}}, \bibinfo {author} {\bibfnamefont {B.}~\bibnamefont
  {Govoreanu}},\ and\ \bibinfo {author} {\bibfnamefont {I.~P.}\ \bibnamefont
  {Radu}},\ }\bibfield  {title} {\bibinfo {title} {A flexible 300 mm integrated
  si mos platform for electron- and hole-spin qubits exploration},\ }in\ \href
  {https://doi.org/10.1109/IEDM13553.2020.9371956} {\emph {\bibinfo {booktitle}
  {2020 IEEE International Electron Devices Meeting (IEDM)}}}\ (\bibinfo {year}
  {2020})\ pp.\ \bibinfo {pages} {38.3.1--38.3.4}\BibitemShut {NoStop}%
\bibitem [{\citenamefont {Elsayed}\ \emph {et~al.}(2022)\citenamefont
  {Elsayed}, \citenamefont {Shehata}, \citenamefont {Godfrin}, \citenamefont
  {Kubicek}, \citenamefont {Massar}, \citenamefont {Canvel}, \citenamefont
  {Jussot}, \citenamefont {Simion}, \citenamefont {Mongillo}, \citenamefont
  {Wan}, \citenamefont {Govoreanu}, \citenamefont {Radu}, \citenamefont {Li},
  \citenamefont {Dorpe},\ and\ \citenamefont {Greve}}]{elsayed2022low}%
  \BibitemOpen
  \bibfield  {author} {\bibinfo {author} {\bibfnamefont {A.}~\bibnamefont
  {Elsayed}}, \bibinfo {author} {\bibfnamefont {M.}~\bibnamefont {Shehata}},
  \bibinfo {author} {\bibfnamefont {C.}~\bibnamefont {Godfrin}}, \bibinfo
  {author} {\bibfnamefont {S.}~\bibnamefont {Kubicek}}, \bibinfo {author}
  {\bibfnamefont {S.}~\bibnamefont {Massar}}, \bibinfo {author} {\bibfnamefont
  {Y.}~\bibnamefont {Canvel}}, \bibinfo {author} {\bibfnamefont
  {J.}~\bibnamefont {Jussot}}, \bibinfo {author} {\bibfnamefont
  {G.}~\bibnamefont {Simion}}, \bibinfo {author} {\bibfnamefont
  {M.}~\bibnamefont {Mongillo}}, \bibinfo {author} {\bibfnamefont
  {D.}~\bibnamefont {Wan}}, \bibinfo {author} {\bibfnamefont {B.}~\bibnamefont
  {Govoreanu}}, \bibinfo {author} {\bibfnamefont {I.~P.}\ \bibnamefont {Radu}},
  \bibinfo {author} {\bibfnamefont {R.}~\bibnamefont {Li}}, \bibinfo {author}
  {\bibfnamefont {P.~V.}\ \bibnamefont {Dorpe}},\ and\ \bibinfo {author}
  {\bibfnamefont {K.~D.}\ \bibnamefont {Greve}},\ }\href@noop {} {\bibinfo
  {title} {Low charge noise quantum dots with industrial cmos manufacturing}}
  (\bibinfo {year} {2022}),\ \Eprint {https://arxiv.org/abs/2212.06464}
  {arXiv:2212.06464 [cond-mat.mes-hall]} \BibitemShut {NoStop}%
\bibitem [{\citenamefont {Neyens}\ \emph {et~al.}(2024)\citenamefont {Neyens},
  \citenamefont {Zietz}, \citenamefont {Watson}, \citenamefont {Luthi},
  \citenamefont {Nethwewala}, \citenamefont {George}, \citenamefont {Henry},
  \citenamefont {Islam}, \citenamefont {Wagner}, \citenamefont {Borjans},
  \citenamefont {Connors}, \citenamefont {Corrigan}, \citenamefont {Curry},
  \citenamefont {Keith}, \citenamefont {Kotlyar}, \citenamefont {Lampert},
  \citenamefont {M{\k{a}}dzik}, \citenamefont {Millard}, \citenamefont
  {Mohiyaddin}, \citenamefont {Pellerano}, \citenamefont {Pillarisetty},
  \citenamefont {Ramsey}, \citenamefont {Savytskyy}, \citenamefont {Schaal},
  \citenamefont {Zheng}, \citenamefont {Ziegler}, \citenamefont {Bishop},
  \citenamefont {Bojarski}, \citenamefont {Roberts},\ and\ \citenamefont
  {Clarke}}]{Neyens2024}%
  \BibitemOpen
  \bibfield  {author} {\bibinfo {author} {\bibfnamefont {S.}~\bibnamefont
  {Neyens}}, \bibinfo {author} {\bibfnamefont {O.~K.}\ \bibnamefont {Zietz}},
  \bibinfo {author} {\bibfnamefont {T.~F.}\ \bibnamefont {Watson}}, \bibinfo
  {author} {\bibfnamefont {F.}~\bibnamefont {Luthi}}, \bibinfo {author}
  {\bibfnamefont {A.}~\bibnamefont {Nethwewala}}, \bibinfo {author}
  {\bibfnamefont {H.~C.}\ \bibnamefont {George}}, \bibinfo {author}
  {\bibfnamefont {E.}~\bibnamefont {Henry}}, \bibinfo {author} {\bibfnamefont
  {M.}~\bibnamefont {Islam}}, \bibinfo {author} {\bibfnamefont {A.~J.}\
  \bibnamefont {Wagner}}, \bibinfo {author} {\bibfnamefont {F.}~\bibnamefont
  {Borjans}}, \bibinfo {author} {\bibfnamefont {E.~J.}\ \bibnamefont
  {Connors}}, \bibinfo {author} {\bibfnamefont {J.}~\bibnamefont {Corrigan}},
  \bibinfo {author} {\bibfnamefont {M.~J.}\ \bibnamefont {Curry}}, \bibinfo
  {author} {\bibfnamefont {D.}~\bibnamefont {Keith}}, \bibinfo {author}
  {\bibfnamefont {R.}~\bibnamefont {Kotlyar}}, \bibinfo {author} {\bibfnamefont
  {L.~F.}\ \bibnamefont {Lampert}}, \bibinfo {author} {\bibfnamefont {M.~T.}\
  \bibnamefont {M{\k{a}}dzik}}, \bibinfo {author} {\bibfnamefont
  {K.}~\bibnamefont {Millard}}, \bibinfo {author} {\bibfnamefont {F.~A.}\
  \bibnamefont {Mohiyaddin}}, \bibinfo {author} {\bibfnamefont
  {S.}~\bibnamefont {Pellerano}}, \bibinfo {author} {\bibfnamefont
  {R.}~\bibnamefont {Pillarisetty}}, \bibinfo {author} {\bibfnamefont
  {M.}~\bibnamefont {Ramsey}}, \bibinfo {author} {\bibfnamefont
  {R.}~\bibnamefont {Savytskyy}}, \bibinfo {author} {\bibfnamefont
  {S.}~\bibnamefont {Schaal}}, \bibinfo {author} {\bibfnamefont
  {G.}~\bibnamefont {Zheng}}, \bibinfo {author} {\bibfnamefont
  {J.}~\bibnamefont {Ziegler}}, \bibinfo {author} {\bibfnamefont {N.~C.}\
  \bibnamefont {Bishop}}, \bibinfo {author} {\bibfnamefont {S.}~\bibnamefont
  {Bojarski}}, \bibinfo {author} {\bibfnamefont {J.}~\bibnamefont {Roberts}},\
  and\ \bibinfo {author} {\bibfnamefont {J.~S.}\ \bibnamefont {Clarke}},\
  }\bibfield  {title} {\bibinfo {title} {Probing single electrons across 300-mm
  spin qubit wafers},\ }\href {https://doi.org/10.1038/s41586-024-07275-6}
  {\bibfield  {journal} {\bibinfo  {journal} {Nature}\ }\textbf {\bibinfo
  {volume} {629}},\ \bibinfo {pages} {80} (\bibinfo {year} {2024})}\BibitemShut
  {NoStop}%
\bibitem [{\citenamefont {Stano}\ and\ \citenamefont {Loss}(2022)}]{Stano2022}%
  \BibitemOpen
  \bibfield  {author} {\bibinfo {author} {\bibfnamefont {P.}~\bibnamefont
  {Stano}}\ and\ \bibinfo {author} {\bibfnamefont {D.}~\bibnamefont {Loss}},\
  }\bibfield  {title} {\bibinfo {title} {Review of performance metrics of spin
  qubits in gated semiconducting nanostructures},\ }\href
  {https://doi.org/10.1038/s42254-022-00484-w} {\bibfield  {journal} {\bibinfo
  {journal} {Nature Reviews Physics}\ }\textbf {\bibinfo {volume} {4}},\
  \bibinfo {pages} {672} (\bibinfo {year} {2022})}\BibitemShut {NoStop}%
\bibitem [{\citenamefont {Tanttu}\ \emph {et~al.}(2024)\citenamefont {Tanttu},
  \citenamefont {Lim}, \citenamefont {Huang}, \citenamefont {Stuyck},
  \citenamefont {Gilbert}, \citenamefont {Su}, \citenamefont {Feng},
  \citenamefont {Cifuentes}, \citenamefont {Seedhouse}, \citenamefont
  {Seritan}, \citenamefont {Ostrove}, \citenamefont {Rudinger}, \citenamefont
  {Leon}, \citenamefont {Huang}, \citenamefont {Escott}, \citenamefont {Itoh},
  \citenamefont {Abrosimov}, \citenamefont {Pohl}, \citenamefont {Thewalt},
  \citenamefont {Hudson}, \citenamefont {Blume-Kohout}, \citenamefont
  {Bartlett}, \citenamefont {Morello}, \citenamefont {Laucht}, \citenamefont
  {Yang}, \citenamefont {Saraiva},\ and\ \citenamefont
  {Dzurak}}]{tanttu2024assessment}%
  \BibitemOpen
  \bibfield  {author} {\bibinfo {author} {\bibfnamefont {T.}~\bibnamefont
  {Tanttu}}, \bibinfo {author} {\bibfnamefont {W.~H.}\ \bibnamefont {Lim}},
  \bibinfo {author} {\bibfnamefont {J.~Y.}\ \bibnamefont {Huang}}, \bibinfo
  {author} {\bibfnamefont {N.~D.}\ \bibnamefont {Stuyck}}, \bibinfo {author}
  {\bibfnamefont {W.}~\bibnamefont {Gilbert}}, \bibinfo {author} {\bibfnamefont
  {R.~Y.}\ \bibnamefont {Su}}, \bibinfo {author} {\bibfnamefont
  {M.}~\bibnamefont {Feng}}, \bibinfo {author} {\bibfnamefont {J.~D.}\
  \bibnamefont {Cifuentes}}, \bibinfo {author} {\bibfnamefont {A.~E.}\
  \bibnamefont {Seedhouse}}, \bibinfo {author} {\bibfnamefont {S.~K.}\
  \bibnamefont {Seritan}}, \bibinfo {author} {\bibfnamefont {C.~I.}\
  \bibnamefont {Ostrove}}, \bibinfo {author} {\bibfnamefont {K.~M.}\
  \bibnamefont {Rudinger}}, \bibinfo {author} {\bibfnamefont {R.~C.~C.}\
  \bibnamefont {Leon}}, \bibinfo {author} {\bibfnamefont {W.}~\bibnamefont
  {Huang}}, \bibinfo {author} {\bibfnamefont {C.~C.}\ \bibnamefont {Escott}},
  \bibinfo {author} {\bibfnamefont {K.~M.}\ \bibnamefont {Itoh}}, \bibinfo
  {author} {\bibfnamefont {N.~V.}\ \bibnamefont {Abrosimov}}, \bibinfo {author}
  {\bibfnamefont {H.-J.}\ \bibnamefont {Pohl}}, \bibinfo {author}
  {\bibfnamefont {M.~L.~W.}\ \bibnamefont {Thewalt}}, \bibinfo {author}
  {\bibfnamefont {F.~E.}\ \bibnamefont {Hudson}}, \bibinfo {author}
  {\bibfnamefont {R.}~\bibnamefont {Blume-Kohout}}, \bibinfo {author}
  {\bibfnamefont {S.~D.}\ \bibnamefont {Bartlett}}, \bibinfo {author}
  {\bibfnamefont {A.}~\bibnamefont {Morello}}, \bibinfo {author} {\bibfnamefont
  {A.}~\bibnamefont {Laucht}}, \bibinfo {author} {\bibfnamefont {C.~H.}\
  \bibnamefont {Yang}}, \bibinfo {author} {\bibfnamefont {A.}~\bibnamefont
  {Saraiva}},\ and\ \bibinfo {author} {\bibfnamefont {A.~S.}\ \bibnamefont
  {Dzurak}},\ }\href@noop {} {\bibinfo {title} {Assessment of error variation
  in high-fidelity two-qubit gates in silicon}} (\bibinfo {year} {2024}),\
  \Eprint {https://arxiv.org/abs/2303.04090} {arXiv:2303.04090 [quant-ph]}
  \BibitemShut {NoStop}%
\bibitem [{\citenamefont {Yoneda}\ \emph {et~al.}(2018)\citenamefont {Yoneda},
  \citenamefont {Takeda}, \citenamefont {Otsuka}, \citenamefont {Nakajima},
  \citenamefont {Delbecq}, \citenamefont {Allison}, \citenamefont {Honda},
  \citenamefont {Kodera}, \citenamefont {Oda}, \citenamefont {Hoshi},
  \citenamefont {Usami}, \citenamefont {Itoh},\ and\ \citenamefont
  {Tarucha}}]{Yoneda2018}%
  \BibitemOpen
  \bibfield  {author} {\bibinfo {author} {\bibfnamefont {J.}~\bibnamefont
  {Yoneda}}, \bibinfo {author} {\bibfnamefont {K.}~\bibnamefont {Takeda}},
  \bibinfo {author} {\bibfnamefont {T.}~\bibnamefont {Otsuka}}, \bibinfo
  {author} {\bibfnamefont {T.}~\bibnamefont {Nakajima}}, \bibinfo {author}
  {\bibfnamefont {M.~R.}\ \bibnamefont {Delbecq}}, \bibinfo {author}
  {\bibfnamefont {G.}~\bibnamefont {Allison}}, \bibinfo {author} {\bibfnamefont
  {T.}~\bibnamefont {Honda}}, \bibinfo {author} {\bibfnamefont
  {T.}~\bibnamefont {Kodera}}, \bibinfo {author} {\bibfnamefont
  {S.}~\bibnamefont {Oda}}, \bibinfo {author} {\bibfnamefont {Y.}~\bibnamefont
  {Hoshi}}, \bibinfo {author} {\bibfnamefont {N.}~\bibnamefont {Usami}},
  \bibinfo {author} {\bibfnamefont {K.~M.}\ \bibnamefont {Itoh}},\ and\
  \bibinfo {author} {\bibfnamefont {S.}~\bibnamefont {Tarucha}},\ }\bibfield
  {title} {\bibinfo {title} {A quantum-dot spin qubit with coherence limited by
  charge noise and fidelity higher than 99.9{\%}},\ }\href
  {https://doi.org/10.1038/s41565-017-0014-x} {\bibfield  {journal} {\bibinfo
  {journal} {Nature Nanotechnology}\ }\textbf {\bibinfo {volume} {13}},\
  \bibinfo {pages} {102} (\bibinfo {year} {2018})}\BibitemShut {NoStop}%
\bibitem [{\citenamefont {Xue}\ \emph {et~al.}(2022)\citenamefont {Xue},
  \citenamefont {Russ}, \citenamefont {Samkharadze}, \citenamefont {Undseth},
  \citenamefont {Sammak}, \citenamefont {Scappucci},\ and\ \citenamefont
  {Vandersypen}}]{Xue2022}%
  \BibitemOpen
  \bibfield  {author} {\bibinfo {author} {\bibfnamefont {X.}~\bibnamefont
  {Xue}}, \bibinfo {author} {\bibfnamefont {M.}~\bibnamefont {Russ}}, \bibinfo
  {author} {\bibfnamefont {N.}~\bibnamefont {Samkharadze}}, \bibinfo {author}
  {\bibfnamefont {B.}~\bibnamefont {Undseth}}, \bibinfo {author} {\bibfnamefont
  {A.}~\bibnamefont {Sammak}}, \bibinfo {author} {\bibfnamefont
  {G.}~\bibnamefont {Scappucci}},\ and\ \bibinfo {author} {\bibfnamefont
  {L.~M.~K.}\ \bibnamefont {Vandersypen}},\ }\bibfield  {title} {\bibinfo
  {title} {Quantum logic with spin qubits crossing the surface code
  threshold},\ }\href {https://doi.org/10.1038/s41586-021-04273-w} {\bibfield
  {journal} {\bibinfo  {journal} {Nature}\ }\textbf {\bibinfo {volume} {601}},\
  \bibinfo {pages} {343} (\bibinfo {year} {2022})}\BibitemShut {NoStop}%
\bibitem [{\citenamefont {Noiri}\ \emph {et~al.}(2022)\citenamefont {Noiri},
  \citenamefont {Takeda}, \citenamefont {Nakajima}, \citenamefont {Kobayashi},
  \citenamefont {Sammak}, \citenamefont {Scappucci},\ and\ \citenamefont
  {Tarucha}}]{Noiri2022}%
  \BibitemOpen
  \bibfield  {author} {\bibinfo {author} {\bibfnamefont {A.}~\bibnamefont
  {Noiri}}, \bibinfo {author} {\bibfnamefont {K.}~\bibnamefont {Takeda}},
  \bibinfo {author} {\bibfnamefont {T.}~\bibnamefont {Nakajima}}, \bibinfo
  {author} {\bibfnamefont {T.}~\bibnamefont {Kobayashi}}, \bibinfo {author}
  {\bibfnamefont {A.}~\bibnamefont {Sammak}}, \bibinfo {author} {\bibfnamefont
  {G.}~\bibnamefont {Scappucci}},\ and\ \bibinfo {author} {\bibfnamefont
  {S.}~\bibnamefont {Tarucha}},\ }\bibfield  {title} {\bibinfo {title} {Fast
  universal quantum gate above the fault-tolerance threshold in silicon},\
  }\href {https://doi.org/10.1038/s41586-021-04182-y} {\bibfield  {journal}
  {\bibinfo  {journal} {Nature}\ }\textbf {\bibinfo {volume} {601}},\ \bibinfo
  {pages} {338} (\bibinfo {year} {2022})}\BibitemShut {NoStop}%
\bibitem [{\citenamefont {Mills}\ \emph {et~al.}(2022)\citenamefont {Mills},
  \citenamefont {Guinn}, \citenamefont {Gullans}, \citenamefont {Sigillito},
  \citenamefont {Feldman}, \citenamefont {Nielsen},\ and\ \citenamefont
  {Petta}}]{doi:10.1126/sciadv.abn5130}%
  \BibitemOpen
  \bibfield  {author} {\bibinfo {author} {\bibfnamefont {A.~R.}\ \bibnamefont
  {Mills}}, \bibinfo {author} {\bibfnamefont {C.~R.}\ \bibnamefont {Guinn}},
  \bibinfo {author} {\bibfnamefont {M.~J.}\ \bibnamefont {Gullans}}, \bibinfo
  {author} {\bibfnamefont {A.~J.}\ \bibnamefont {Sigillito}}, \bibinfo {author}
  {\bibfnamefont {M.~M.}\ \bibnamefont {Feldman}}, \bibinfo {author}
  {\bibfnamefont {E.}~\bibnamefont {Nielsen}},\ and\ \bibinfo {author}
  {\bibfnamefont {J.~R.}\ \bibnamefont {Petta}},\ }\bibfield  {title} {\bibinfo
  {title} {Two-qubit silicon quantum processor with operation fidelity
  exceeding 99\%},\ }\href {https://doi.org/10.1126/sciadv.abn5130} {\bibfield
  {journal} {\bibinfo  {journal} {Science Advances}\ }\textbf {\bibinfo
  {volume} {8}},\ \bibinfo {pages} {eabn5130} (\bibinfo {year} {2022})},\
  \Eprint
  {https://arxiv.org/abs/https://www.science.org/doi/pdf/10.1126/sciadv.abn5130}
  {https://www.science.org/doi/pdf/10.1126/sciadv.abn5130} \BibitemShut
  {NoStop}%
\bibitem [{\citenamefont {Veldhorst}\ \emph {et~al.}(2015)\citenamefont
  {Veldhorst}, \citenamefont {Yang}, \citenamefont {Hwang}, \citenamefont
  {Huang}, \citenamefont {Dehollain}, \citenamefont {Muhonen}, \citenamefont
  {Simmons}, \citenamefont {Laucht}, \citenamefont {Hudson}, \citenamefont
  {Itoh}, \citenamefont {Morello},\ and\ \citenamefont
  {Dzurak}}]{Veldhorst2015Oct}%
  \BibitemOpen
  \bibfield  {author} {\bibinfo {author} {\bibfnamefont {M.}~\bibnamefont
  {Veldhorst}}, \bibinfo {author} {\bibfnamefont {C.~H.}\ \bibnamefont {Yang}},
  \bibinfo {author} {\bibfnamefont {J.~C.~C.}\ \bibnamefont {Hwang}}, \bibinfo
  {author} {\bibfnamefont {W.}~\bibnamefont {Huang}}, \bibinfo {author}
  {\bibfnamefont {J.~P.}\ \bibnamefont {Dehollain}}, \bibinfo {author}
  {\bibfnamefont {J.~T.}\ \bibnamefont {Muhonen}}, \bibinfo {author}
  {\bibfnamefont {S.}~\bibnamefont {Simmons}}, \bibinfo {author} {\bibfnamefont
  {A.}~\bibnamefont {Laucht}}, \bibinfo {author} {\bibfnamefont {F.~E.}\
  \bibnamefont {Hudson}}, \bibinfo {author} {\bibfnamefont {K.~M.}\
  \bibnamefont {Itoh}}, \bibinfo {author} {\bibfnamefont {A.}~\bibnamefont
  {Morello}},\ and\ \bibinfo {author} {\bibfnamefont {A.~S.}\ \bibnamefont
  {Dzurak}},\ }\bibfield  {title} {\bibinfo {title} {{A two-qubit logic gate in
  silicon}},\ }\href {https://doi.org/10.1038/nature15263} {\bibfield
  {journal} {\bibinfo  {journal} {Nature}\ }\textbf {\bibinfo {volume} {526}},\
  \bibinfo {pages} {410} (\bibinfo {year} {2015})}\BibitemShut {NoStop}%
\bibitem [{\citenamefont {Hensgens}\ \emph {et~al.}(2017)\citenamefont
  {Hensgens}, \citenamefont {Fujita}, \citenamefont {Janssen}, \citenamefont
  {Li}, \citenamefont {Van~Diepen}, \citenamefont {Reichl}, \citenamefont
  {Wegscheider}, \citenamefont {Das~Sarma},\ and\ \citenamefont
  {Vandersypen}}]{Hensgens2017}%
  \BibitemOpen
  \bibfield  {author} {\bibinfo {author} {\bibfnamefont {T.}~\bibnamefont
  {Hensgens}}, \bibinfo {author} {\bibfnamefont {T.}~\bibnamefont {Fujita}},
  \bibinfo {author} {\bibfnamefont {L.}~\bibnamefont {Janssen}}, \bibinfo
  {author} {\bibfnamefont {X.}~\bibnamefont {Li}}, \bibinfo {author}
  {\bibfnamefont {C.~J.}\ \bibnamefont {Van~Diepen}}, \bibinfo {author}
  {\bibfnamefont {C.}~\bibnamefont {Reichl}}, \bibinfo {author} {\bibfnamefont
  {W.}~\bibnamefont {Wegscheider}}, \bibinfo {author} {\bibfnamefont
  {S.}~\bibnamefont {Das~Sarma}},\ and\ \bibinfo {author} {\bibfnamefont
  {L.~M.~K.}\ \bibnamefont {Vandersypen}},\ }\bibfield  {title} {\bibinfo
  {title} {Quantum simulation of a fermi--hubbard model using a semiconductor
  quantum dot array},\ }\href {https://doi.org/10.1038/nature23022} {\bibfield
  {journal} {\bibinfo  {journal} {Nature}\ }\textbf {\bibinfo {volume} {548}},\
  \bibinfo {pages} {70} (\bibinfo {year} {2017})}\BibitemShut {NoStop}%
\bibitem [{\citenamefont {Altman}\ \emph {et~al.}(2021)\citenamefont {Altman},
  \citenamefont {Brown}, \citenamefont {Carleo}, \citenamefont {Carr},
  \citenamefont {Demler}, \citenamefont {Chin}, \citenamefont {DeMarco},
  \citenamefont {Economou}, \citenamefont {Eriksson}, \citenamefont {Fu},
  \citenamefont {Greiner}, \citenamefont {Hazzard}, \citenamefont {Hulet},
  \citenamefont {Koll\'ar}, \citenamefont {Lev}, \citenamefont {Lukin},
  \citenamefont {Ma}, \citenamefont {Mi}, \citenamefont {Misra}, \citenamefont
  {Monroe}, \citenamefont {Murch}, \citenamefont {Nazario}, \citenamefont {Ni},
  \citenamefont {Potter}, \citenamefont {Roushan}, \citenamefont {Saffman},
  \citenamefont {Schleier-Smith}, \citenamefont {Siddiqi}, \citenamefont
  {Simmonds}, \citenamefont {Singh}, \citenamefont {Spielman}, \citenamefont
  {Temme}, \citenamefont {Weiss}, \citenamefont {Vu\ifmmode \check{c}\else
  \v{c}\fi{}kovi\ifmmode~\acute{c}\else \'{c}\fi{}}, \citenamefont
  {Vuleti\ifmmode~\acute{c}\else \'{c}\fi{}}, \citenamefont {Ye},\ and\
  \citenamefont {Zwierlein}}]{PRXQuantum.2.017003}%
  \BibitemOpen
  \bibfield  {author} {\bibinfo {author} {\bibfnamefont {E.}~\bibnamefont
  {Altman}}, \bibinfo {author} {\bibfnamefont {K.~R.}\ \bibnamefont {Brown}},
  \bibinfo {author} {\bibfnamefont {G.}~\bibnamefont {Carleo}}, \bibinfo
  {author} {\bibfnamefont {L.~D.}\ \bibnamefont {Carr}}, \bibinfo {author}
  {\bibfnamefont {E.}~\bibnamefont {Demler}}, \bibinfo {author} {\bibfnamefont
  {C.}~\bibnamefont {Chin}}, \bibinfo {author} {\bibfnamefont {B.}~\bibnamefont
  {DeMarco}}, \bibinfo {author} {\bibfnamefont {S.~E.}\ \bibnamefont
  {Economou}}, \bibinfo {author} {\bibfnamefont {M.~A.}\ \bibnamefont
  {Eriksson}}, \bibinfo {author} {\bibfnamefont {K.-M.~C.}\ \bibnamefont {Fu}},
  \bibinfo {author} {\bibfnamefont {M.}~\bibnamefont {Greiner}}, \bibinfo
  {author} {\bibfnamefont {K.~R.}\ \bibnamefont {Hazzard}}, \bibinfo {author}
  {\bibfnamefont {R.~G.}\ \bibnamefont {Hulet}}, \bibinfo {author}
  {\bibfnamefont {A.~J.}\ \bibnamefont {Koll\'ar}}, \bibinfo {author}
  {\bibfnamefont {B.~L.}\ \bibnamefont {Lev}}, \bibinfo {author} {\bibfnamefont
  {M.~D.}\ \bibnamefont {Lukin}}, \bibinfo {author} {\bibfnamefont
  {R.}~\bibnamefont {Ma}}, \bibinfo {author} {\bibfnamefont {X.}~\bibnamefont
  {Mi}}, \bibinfo {author} {\bibfnamefont {S.}~\bibnamefont {Misra}}, \bibinfo
  {author} {\bibfnamefont {C.}~\bibnamefont {Monroe}}, \bibinfo {author}
  {\bibfnamefont {K.}~\bibnamefont {Murch}}, \bibinfo {author} {\bibfnamefont
  {Z.}~\bibnamefont {Nazario}}, \bibinfo {author} {\bibfnamefont {K.-K.}\
  \bibnamefont {Ni}}, \bibinfo {author} {\bibfnamefont {A.~C.}\ \bibnamefont
  {Potter}}, \bibinfo {author} {\bibfnamefont {P.}~\bibnamefont {Roushan}},
  \bibinfo {author} {\bibfnamefont {M.}~\bibnamefont {Saffman}}, \bibinfo
  {author} {\bibfnamefont {M.}~\bibnamefont {Schleier-Smith}}, \bibinfo
  {author} {\bibfnamefont {I.}~\bibnamefont {Siddiqi}}, \bibinfo {author}
  {\bibfnamefont {R.}~\bibnamefont {Simmonds}}, \bibinfo {author}
  {\bibfnamefont {M.}~\bibnamefont {Singh}}, \bibinfo {author} {\bibfnamefont
  {I.}~\bibnamefont {Spielman}}, \bibinfo {author} {\bibfnamefont
  {K.}~\bibnamefont {Temme}}, \bibinfo {author} {\bibfnamefont {D.~S.}\
  \bibnamefont {Weiss}}, \bibinfo {author} {\bibfnamefont {J.}~\bibnamefont
  {Vu\ifmmode \check{c}\else \v{c}\fi{}kovi\ifmmode~\acute{c}\else
  \'{c}\fi{}}}, \bibinfo {author} {\bibfnamefont {V.}~\bibnamefont
  {Vuleti\ifmmode~\acute{c}\else \'{c}\fi{}}}, \bibinfo {author} {\bibfnamefont
  {J.}~\bibnamefont {Ye}},\ and\ \bibinfo {author} {\bibfnamefont
  {M.}~\bibnamefont {Zwierlein}},\ }\bibfield  {title} {\bibinfo {title}
  {Quantum simulators: Architectures and opportunities},\ }\href
  {https://doi.org/10.1103/PRXQuantum.2.017003} {\bibfield  {journal} {\bibinfo
   {journal} {PRX Quantum}\ }\textbf {\bibinfo {volume} {2}},\ \bibinfo {pages}
  {017003} (\bibinfo {year} {2021})}\BibitemShut {NoStop}%
\bibitem [{\citenamefont {Calderon-Vargas}\ \emph {et~al.}(2019)\citenamefont
  {Calderon-Vargas}, \citenamefont {Barron}, \citenamefont {Deng},
  \citenamefont {Sigillito}, \citenamefont {Barnes},\ and\ \citenamefont
  {Economou}}]{PhysRevB.100.035304}%
  \BibitemOpen
  \bibfield  {author} {\bibinfo {author} {\bibfnamefont {F.~A.}\ \bibnamefont
  {Calderon-Vargas}}, \bibinfo {author} {\bibfnamefont {G.~S.}\ \bibnamefont
  {Barron}}, \bibinfo {author} {\bibfnamefont {X.-H.}\ \bibnamefont {Deng}},
  \bibinfo {author} {\bibfnamefont {A.~J.}\ \bibnamefont {Sigillito}}, \bibinfo
  {author} {\bibfnamefont {E.}~\bibnamefont {Barnes}},\ and\ \bibinfo {author}
  {\bibfnamefont {S.~E.}\ \bibnamefont {Economou}},\ }\bibfield  {title}
  {\bibinfo {title} {Fast high-fidelity entangling gates for spin qubits in si
  double quantum dots},\ }\href {https://doi.org/10.1103/PhysRevB.100.035304}
  {\bibfield  {journal} {\bibinfo  {journal} {Phys. Rev. B}\ }\textbf {\bibinfo
  {volume} {100}},\ \bibinfo {pages} {035304} (\bibinfo {year}
  {2019})}\BibitemShut {NoStop}%
\bibitem [{\citenamefont {Dehollain}\ \emph {et~al.}(2012)\citenamefont
  {Dehollain}, \citenamefont {Pla}, \citenamefont {Siew}, \citenamefont {Tan},
  \citenamefont {Dzurak},\ and\ \citenamefont {Morello}}]{Dehollain_2013}%
  \BibitemOpen
  \bibfield  {author} {\bibinfo {author} {\bibfnamefont {J.~P.}\ \bibnamefont
  {Dehollain}}, \bibinfo {author} {\bibfnamefont {J.~J.}\ \bibnamefont {Pla}},
  \bibinfo {author} {\bibfnamefont {E.}~\bibnamefont {Siew}}, \bibinfo {author}
  {\bibfnamefont {K.~Y.}\ \bibnamefont {Tan}}, \bibinfo {author} {\bibfnamefont
  {A.~S.}\ \bibnamefont {Dzurak}},\ and\ \bibinfo {author} {\bibfnamefont
  {A.}~\bibnamefont {Morello}},\ }\bibfield  {title} {\bibinfo {title}
  {Nanoscale broadband transmission lines for spin qubit control},\ }\href
  {https://doi.org/10.1088/0957-4484/24/1/015202} {\bibfield  {journal}
  {\bibinfo  {journal} {Nanotechnology}\ }\textbf {\bibinfo {volume} {24}},\
  \bibinfo {pages} {015202} (\bibinfo {year} {2012})}\BibitemShut {NoStop}%
\bibitem [{\citenamefont {Dial}\ \emph {et~al.}(2013)\citenamefont {Dial},
  \citenamefont {Shulman}, \citenamefont {Harvey}, \citenamefont {Bluhm},
  \citenamefont {Umansky},\ and\ \citenamefont
  {Yacoby}}]{PhysRevLett.110.146804}%
  \BibitemOpen
  \bibfield  {author} {\bibinfo {author} {\bibfnamefont {O.~E.}\ \bibnamefont
  {Dial}}, \bibinfo {author} {\bibfnamefont {M.~D.}\ \bibnamefont {Shulman}},
  \bibinfo {author} {\bibfnamefont {S.~P.}\ \bibnamefont {Harvey}}, \bibinfo
  {author} {\bibfnamefont {H.}~\bibnamefont {Bluhm}}, \bibinfo {author}
  {\bibfnamefont {V.}~\bibnamefont {Umansky}},\ and\ \bibinfo {author}
  {\bibfnamefont {A.}~\bibnamefont {Yacoby}},\ }\bibfield  {title} {\bibinfo
  {title} {Charge noise spectroscopy using coherent exchange oscillations in a
  singlet-triplet qubit},\ }\href
  {https://doi.org/10.1103/PhysRevLett.110.146804} {\bibfield  {journal}
  {\bibinfo  {journal} {Phys. Rev. Lett.}\ }\textbf {\bibinfo {volume} {110}},\
  \bibinfo {pages} {146804} (\bibinfo {year} {2013})}\BibitemShut {NoStop}%
\bibitem [{\citenamefont {Martins}\ \emph {et~al.}(2016)\citenamefont
  {Martins}, \citenamefont {Malinowski}, \citenamefont {Nissen}, \citenamefont
  {Barnes}, \citenamefont {Fallahi}, \citenamefont {Gardner}, \citenamefont
  {Manfra}, \citenamefont {Marcus},\ and\ \citenamefont
  {Kuemmeth}}]{PhysRevLett.116.116801}%
  \BibitemOpen
  \bibfield  {author} {\bibinfo {author} {\bibfnamefont {F.}~\bibnamefont
  {Martins}}, \bibinfo {author} {\bibfnamefont {F.~K.}\ \bibnamefont
  {Malinowski}}, \bibinfo {author} {\bibfnamefont {P.~D.}\ \bibnamefont
  {Nissen}}, \bibinfo {author} {\bibfnamefont {E.}~\bibnamefont {Barnes}},
  \bibinfo {author} {\bibfnamefont {S.}~\bibnamefont {Fallahi}}, \bibinfo
  {author} {\bibfnamefont {G.~C.}\ \bibnamefont {Gardner}}, \bibinfo {author}
  {\bibfnamefont {M.~J.}\ \bibnamefont {Manfra}}, \bibinfo {author}
  {\bibfnamefont {C.~M.}\ \bibnamefont {Marcus}},\ and\ \bibinfo {author}
  {\bibfnamefont {F.}~\bibnamefont {Kuemmeth}},\ }\bibfield  {title} {\bibinfo
  {title} {Noise suppression using symmetric exchange gates in spin qubits},\
  }\href {https://doi.org/10.1103/PhysRevLett.116.116801} {\bibfield  {journal}
  {\bibinfo  {journal} {Phys. Rev. Lett.}\ }\textbf {\bibinfo {volume} {116}},\
  \bibinfo {pages} {116801} (\bibinfo {year} {2016})}\BibitemShut {NoStop}%
\bibitem [{\citenamefont {Zajac}\ \emph {et~al.}(2018)\citenamefont {Zajac},
  \citenamefont {Sigillito}, \citenamefont {Russ}, \citenamefont {Borjans},
  \citenamefont {Taylor}, \citenamefont {Burkard},\ and\ \citenamefont
  {Petta}}]{doi:10.1126/science.aao5965}%
  \BibitemOpen
  \bibfield  {author} {\bibinfo {author} {\bibfnamefont {D.~M.}\ \bibnamefont
  {Zajac}}, \bibinfo {author} {\bibfnamefont {A.~J.}\ \bibnamefont
  {Sigillito}}, \bibinfo {author} {\bibfnamefont {M.}~\bibnamefont {Russ}},
  \bibinfo {author} {\bibfnamefont {F.}~\bibnamefont {Borjans}}, \bibinfo
  {author} {\bibfnamefont {J.~M.}\ \bibnamefont {Taylor}}, \bibinfo {author}
  {\bibfnamefont {G.}~\bibnamefont {Burkard}},\ and\ \bibinfo {author}
  {\bibfnamefont {J.~R.}\ \bibnamefont {Petta}},\ }\bibfield  {title} {\bibinfo
  {title} {Resonantly driven cnot gate for electron spins},\ }\href
  {https://doi.org/10.1126/science.aao5965} {\bibfield  {journal} {\bibinfo
  {journal} {Science}\ }\textbf {\bibinfo {volume} {359}},\ \bibinfo {pages}
  {439} (\bibinfo {year} {2018})},\ \Eprint
  {https://arxiv.org/abs/https://www.science.org/doi/pdf/10.1126/science.aao5965}
  {https://www.science.org/doi/pdf/10.1126/science.aao5965} \BibitemShut
  {NoStop}%
\bibitem [{\citenamefont {Yoneda}\ \emph {et~al.}(2015)\citenamefont {Yoneda},
  \citenamefont {Otsuka}, \citenamefont {Takakura}, \citenamefont
  {Pioro-Ladrière}, \citenamefont {Brunner}, \citenamefont {Lu}, \citenamefont
  {Nakajima}, \citenamefont {Obata}, \citenamefont {Noiri}, \citenamefont
  {Palmstrøm}, \citenamefont {Gossard},\ and\ \citenamefont
  {Tarucha}}]{Yoneda_2015}%
  \BibitemOpen
  \bibfield  {author} {\bibinfo {author} {\bibfnamefont {J.}~\bibnamefont
  {Yoneda}}, \bibinfo {author} {\bibfnamefont {T.}~\bibnamefont {Otsuka}},
  \bibinfo {author} {\bibfnamefont {T.}~\bibnamefont {Takakura}}, \bibinfo
  {author} {\bibfnamefont {M.}~\bibnamefont {Pioro-Ladrière}}, \bibinfo
  {author} {\bibfnamefont {R.}~\bibnamefont {Brunner}}, \bibinfo {author}
  {\bibfnamefont {H.}~\bibnamefont {Lu}}, \bibinfo {author} {\bibfnamefont
  {T.}~\bibnamefont {Nakajima}}, \bibinfo {author} {\bibfnamefont
  {T.}~\bibnamefont {Obata}}, \bibinfo {author} {\bibfnamefont
  {A.}~\bibnamefont {Noiri}}, \bibinfo {author} {\bibfnamefont {C.~J.}\
  \bibnamefont {Palmstrøm}}, \bibinfo {author} {\bibfnamefont {A.~C.}\
  \bibnamefont {Gossard}},\ and\ \bibinfo {author} {\bibfnamefont
  {S.}~\bibnamefont {Tarucha}},\ }\bibfield  {title} {\bibinfo {title} {Robust
  micromagnet design for fast electrical manipulations of single spins in
  quantum dots},\ }\href {https://doi.org/10.7567/APEX.8.084401} {\bibfield
  {journal} {\bibinfo  {journal} {Applied Physics Express}\ }\textbf {\bibinfo
  {volume} {8}},\ \bibinfo {pages} {084401} (\bibinfo {year}
  {2015})}\BibitemShut {NoStop}%
\bibitem [{\citenamefont {Khaneja}\ \emph {et~al.}(2005)\citenamefont
  {Khaneja}, \citenamefont {Reiss}, \citenamefont {Kehlet}, \citenamefont
  {Schulte-Herbrüggen},\ and\ \citenamefont {Glaser}}]{KHANEJA2005296}%
  \BibitemOpen
  \bibfield  {author} {\bibinfo {author} {\bibfnamefont {N.}~\bibnamefont
  {Khaneja}}, \bibinfo {author} {\bibfnamefont {T.}~\bibnamefont {Reiss}},
  \bibinfo {author} {\bibfnamefont {C.}~\bibnamefont {Kehlet}}, \bibinfo
  {author} {\bibfnamefont {T.}~\bibnamefont {Schulte-Herbrüggen}},\ and\
  \bibinfo {author} {\bibfnamefont {S.~J.}\ \bibnamefont {Glaser}},\ }\bibfield
   {title} {\bibinfo {title} {Optimal control of coupled spin dynamics: design
  of nmr pulse sequences by gradient ascent algorithms},\ }\href
  {https://doi.org/https://doi.org/10.1016/j.jmr.2004.11.004} {\bibfield
  {journal} {\bibinfo  {journal} {Journal of Magnetic Resonance}\ }\textbf
  {\bibinfo {volume} {172}},\ \bibinfo {pages} {296} (\bibinfo {year}
  {2005})}\BibitemShut {NoStop}%
\bibitem [{\citenamefont {McArdle}\ \emph {et~al.}(2020)\citenamefont
  {McArdle}, \citenamefont {Endo}, \citenamefont {Aspuru-Guzik}, \citenamefont
  {Benjamin},\ and\ \citenamefont {Yuan}}]{RevModPhys.92.015003}%
  \BibitemOpen
  \bibfield  {author} {\bibinfo {author} {\bibfnamefont {S.}~\bibnamefont
  {McArdle}}, \bibinfo {author} {\bibfnamefont {S.}~\bibnamefont {Endo}},
  \bibinfo {author} {\bibfnamefont {A.}~\bibnamefont {Aspuru-Guzik}}, \bibinfo
  {author} {\bibfnamefont {S.~C.}\ \bibnamefont {Benjamin}},\ and\ \bibinfo
  {author} {\bibfnamefont {X.}~\bibnamefont {Yuan}},\ }\bibfield  {title}
  {\bibinfo {title} {Quantum computational chemistry},\ }\href
  {https://doi.org/10.1103/RevModPhys.92.015003} {\bibfield  {journal}
  {\bibinfo  {journal} {Rev. Mod. Phys.}\ }\textbf {\bibinfo {volume} {92}},\
  \bibinfo {pages} {015003} (\bibinfo {year} {2020})}\BibitemShut {NoStop}%
\bibitem [{\citenamefont {Das~Sarma}\ \emph {et~al.}(2011)\citenamefont
  {Das~Sarma}, \citenamefont {Wang},\ and\ \citenamefont
  {Yang}}]{PhysRevB.83.235314}%
  \BibitemOpen
  \bibfield  {author} {\bibinfo {author} {\bibfnamefont {S.}~\bibnamefont
  {Das~Sarma}}, \bibinfo {author} {\bibfnamefont {X.}~\bibnamefont {Wang}},\
  and\ \bibinfo {author} {\bibfnamefont {S.}~\bibnamefont {Yang}},\ }\bibfield
  {title} {\bibinfo {title} {Hubbard model description of silicon spin qubits:
  Charge stability diagram and tunnel coupling in si double quantum dots},\
  }\href {https://doi.org/10.1103/PhysRevB.83.235314} {\bibfield  {journal}
  {\bibinfo  {journal} {Phys. Rev. B}\ }\textbf {\bibinfo {volume} {83}},\
  \bibinfo {pages} {235314} (\bibinfo {year} {2011})}\BibitemShut {NoStop}%
\bibitem [{\citenamefont {Rubin}\ \emph {et~al.}(2018)\citenamefont {Rubin},
  \citenamefont {Babbush},\ and\ \citenamefont {McClean}}]{Rubin_2018}%
  \BibitemOpen
  \bibfield  {author} {\bibinfo {author} {\bibfnamefont {N.~C.}\ \bibnamefont
  {Rubin}}, \bibinfo {author} {\bibfnamefont {R.}~\bibnamefont {Babbush}},\
  and\ \bibinfo {author} {\bibfnamefont {J.}~\bibnamefont {McClean}},\
  }\bibfield  {title} {\bibinfo {title} {Application of fermionic marginal
  constraints to hybrid quantum algorithms},\ }\href
  {https://doi.org/10.1088/1367-2630/aab919} {\bibfield  {journal} {\bibinfo
  {journal} {New Journal of Physics}\ }\textbf {\bibinfo {volume} {20}},\
  \bibinfo {pages} {053020} (\bibinfo {year} {2018})}\BibitemShut {NoStop}%
\bibitem [{\citenamefont {Hoeffding}(1963)}]{Hoeffding}%
  \BibitemOpen
  \bibfield  {author} {\bibinfo {author} {\bibfnamefont {W.}~\bibnamefont
  {Hoeffding}},\ }\bibfield  {title} {\bibinfo {title} {Probability
  inequalities for sums of bounded random variables},\ }\href
  {http://www.jstor.org/stable/2282952} {\bibfield  {journal} {\bibinfo
  {journal} {Journal of the American Statistical Association}\ }\textbf
  {\bibinfo {volume} {58}},\ \bibinfo {pages} {13} (\bibinfo {year}
  {1963})}\BibitemShut {NoStop}%
\bibitem [{\citenamefont {Corna}\ \emph {et~al.}(2018)\citenamefont {Corna},
  \citenamefont {Bourdet}, \citenamefont {Maurand}, \citenamefont {Crippa},
  \citenamefont {Kotekar-Patil}, \citenamefont {Bohuslavskyi}, \citenamefont
  {Lavi{\'e}ville}, \citenamefont {Hutin}, \citenamefont {Barraud},
  \citenamefont {Jehl}, \citenamefont {Vinet}, \citenamefont {De~Franceschi},
  \citenamefont {Niquet},\ and\ \citenamefont {Sanquer}}]{Corna2018}%
  \BibitemOpen
  \bibfield  {author} {\bibinfo {author} {\bibfnamefont {A.}~\bibnamefont
  {Corna}}, \bibinfo {author} {\bibfnamefont {L.}~\bibnamefont {Bourdet}},
  \bibinfo {author} {\bibfnamefont {R.}~\bibnamefont {Maurand}}, \bibinfo
  {author} {\bibfnamefont {A.}~\bibnamefont {Crippa}}, \bibinfo {author}
  {\bibfnamefont {D.}~\bibnamefont {Kotekar-Patil}}, \bibinfo {author}
  {\bibfnamefont {H.}~\bibnamefont {Bohuslavskyi}}, \bibinfo {author}
  {\bibfnamefont {R.}~\bibnamefont {Lavi{\'e}ville}}, \bibinfo {author}
  {\bibfnamefont {L.}~\bibnamefont {Hutin}}, \bibinfo {author} {\bibfnamefont
  {S.}~\bibnamefont {Barraud}}, \bibinfo {author} {\bibfnamefont
  {X.}~\bibnamefont {Jehl}}, \bibinfo {author} {\bibfnamefont {M.}~\bibnamefont
  {Vinet}}, \bibinfo {author} {\bibfnamefont {S.}~\bibnamefont
  {De~Franceschi}}, \bibinfo {author} {\bibfnamefont {Y.-M.}\ \bibnamefont
  {Niquet}},\ and\ \bibinfo {author} {\bibfnamefont {M.}~\bibnamefont
  {Sanquer}},\ }\bibfield  {title} {\bibinfo {title} {Electrically driven
  electron spin resonance mediated by spin--valley--orbit coupling in a silicon
  quantum dot},\ }\href {https://doi.org/10.1038/s41534-018-0059-1} {\bibfield
  {journal} {\bibinfo  {journal} {npj Quantum Information}\ }\textbf {\bibinfo
  {volume} {4}},\ \bibinfo {pages} {6} (\bibinfo {year} {2018})}\BibitemShut
  {NoStop}%
\bibitem [{\citenamefont {\ifmmode \check{Z}\else
  \v{Z}\fi{}nidari\ifmmode~\check{c}\else \v{c}\fi{}}\ \emph
  {et~al.}(2008)\citenamefont {\ifmmode \check{Z}\else
  \v{Z}\fi{}nidari\ifmmode~\check{c}\else \v{c}\fi{}}, \citenamefont {Giraud},\
  and\ \citenamefont {Georgeot}}]{CNOTCount}%
  \BibitemOpen
  \bibfield  {author} {\bibinfo {author} {\bibfnamefont {M.}~\bibnamefont
  {\ifmmode \check{Z}\else \v{Z}\fi{}nidari\ifmmode~\check{c}\else
  \v{c}\fi{}}}, \bibinfo {author} {\bibfnamefont {O.}~\bibnamefont {Giraud}},\
  and\ \bibinfo {author} {\bibfnamefont {B.}~\bibnamefont {Georgeot}},\
  }\bibfield  {title} {\bibinfo {title} {Optimal number of controlled-not gates
  to generate a three-qubit state},\ }\href
  {https://doi.org/10.1103/PhysRevA.77.032320} {\bibfield  {journal} {\bibinfo
  {journal} {Phys. Rev. A}\ }\textbf {\bibinfo {volume} {77}},\ \bibinfo
  {pages} {032320} (\bibinfo {year} {2008})}\BibitemShut {NoStop}%
\bibitem [{\citenamefont {Song}\ and\ \citenamefont
  {Gupta}(1997)}]{43fd0a6c-7343-3149-83c3-841692a8522d}%
  \BibitemOpen
  \bibfield  {author} {\bibinfo {author} {\bibfnamefont {D.}~\bibnamefont
  {Song}}\ and\ \bibinfo {author} {\bibfnamefont {A.~K.}\ \bibnamefont
  {Gupta}},\ }\bibfield  {title} {\bibinfo {title} {Lp-norm uniform
  distribution},\ }\href {http://www.jstor.org/stable/2161691} {\bibfield
  {journal} {\bibinfo  {journal} {Proceedings of the American Mathematical
  Society}\ }\textbf {\bibinfo {volume} {125}},\ \bibinfo {pages} {595}
  (\bibinfo {year} {1997})}\BibitemShut {NoStop}%
\end{thebibliography}%

\newlength{\doublecolumnwidth}
\setlength{\doublecolumnwidth}{\columnwidth}

\appendix
\renewcommand{\thesubsection}{\Alph{section}.\arabic{subsection}}

\section{Footnotes}\label{footnotes}
{\footnotesize
\noindent\textsuperscript{a}From Ref.~\cite{Yoneda2021}

\noindent\textsuperscript{b}The fastest single qubit gates recorded on silicon electron Loss-DiVincenzo qubits take approximately \SI{50}{\ns} \cite{Stano2022}. This corresponds to a drive strength of $2\times\SI{20}{\MHz}$. However, as there are $S$ signals on the line, each with an in-phase and quadrature component, we divide by $\sqrt{2}S$. As we work with up to four qubits in this paper, we take $S=4$ in this expression for all devices considered to enable fair comparisons.

\noindent\textsuperscript{c}The order of magnitude is found by using $\left|J\right|=4\frac{t^2}{U}$ where $t$ is the interdot hopping strength and $U$ is the double occupation energy. Ref.~\cite{Simmons2009} measures $t\le\SI{12}{\GHz}$ and Ref.~\cite{PhysRevB.83.235314} fits for  $U\approx\SI{240}{\GHz}$ using the data from Ref.~\cite{Simmons2009}. Additionally, Ref.~\cite{t_values} measures $t\le\SI{100}{\GHz}$ and we estimate $U\approx\SI{30}{THz}$ as half the width of the (1,1) diamond along the detuning axis in Fig.~7 using the quoted detuning axis lever arm of \SI{0.3}{eV\per V}.
}

\section{Cost Function}\label{app: cost function}
Throughout this article, we consider the cost function
\begin{align}
  f(\bm{x})\coloneqq\mel{\psi(\bm{x};T)}{\hat C}{\psi(\bm{x};T)}.
\end{align}
If $\hat C$ acts on $n$ qubits we can estimate $f(\bm{x})$ by decomposing $\hat C$ as a sum of weighted Pauli strings:
\begin{align}
  \hat C\equiv\sum_{\hat P\in\mathcal P_n}c_{\hat P}\hat P,
\end{align}
where $\mathcal P_n$ is the set of Pauli strings on $n$ qubits and $c_{\hat P}$ is the coefficient of the Pauli string $\hat P\in\mathcal P_n$. Thus, 
\begin{align}
  f(\bm{x})\equiv\sum_{\hat P\in\mathcal P_n}c_{\hat P}\mel{\psi(\bm{x};T)}{\hat P}{\psi(\bm{x};T)}.
\end{align}

By preparing many copies of $\ket{\psi\left(\bm{x}; T\right)}$ the expectation value $\mathbb E_{\hat P}\coloneqq\mel{\psi\left(\bm{x}; T\right)}{\hat P}{\psi\left(\bm{x}; T\right)}$ of each Pauli string $\hat P$ with non-zero coefficient $c_{\hat P}$ can be estimated as follows: To rotate into the eigenbasis of $\hat P$ apply a Haramard ($R_x\left(-\pi/2\right)$) gate to every qubit upon which $\hat P$ acts as an $X$ ($Y$) gate. If the Hadamard and $R_x\left(-\pi/2\right)$ gates are only calibrated to start at discrete clock times $t=m\Delta\tau$ for $m\in\mathbb N$ then one simply needs to set $g(t)$ and $J(t)$ to zero for $t\in\left[T,\left\lceil T/\Delta\tau\right\rceil \Delta\tau\right]$ and then apply the Hadamard and $R_x\left(-\pi/2\right)$ gates at $t=\left\lceil T/\Delta\tau\right\rceil \Delta\tau$. Next, we can measure the $l$ qubits in the computational basis upon which $\hat P$ acts non-trivially. If the Hamming weight of the resulting $l$-bitstring is even (odd) then $\ket{\psi\left(\bm{x}; T\right)}$ was measured as being in the $+1$ ($-1$) eigenspace of $\hat P$. Let $N_{\hat P}^{\pm}$ be the number of measurements in the $\pm1$ eigenspaces of $\hat P$. We can estimate $\mathbb E_P$ as $\tilde {\mathbb E}_{\hat P}\coloneqq\left(N_{\hat P}^+-N_{\hat P}^-\right)/N_{\hat P}$ where $N_{\hat P}=N_{\hat P}^++N_{\hat P}^-$ is the total number of measurements of $\hat P$. Ref.~\cite{Rubin_2018} proves that given $\bar N\coloneqq \sum_{\hat Q\in\mathcal P_n}N_{\hat Q}$ total measurements, the optimal distribution of these across the Pauli strings is:
\begin{equation}
    N_{\hat P}=\frac{\left|c_{\hat P}\right|}{\sum_{\hat Q\in\mathcal P_n}\left|c_{\hat Q}\right|}\bar N.
\end{equation}
Thus, using Hoeffding's inequality \cite{Hoeffding}, one can show that with
\begin{align}
  N_{\hat P}=\left|c_{\hat P}\right|\sum_{\hat Q\in\mathcal P_n}\left|c_{\hat Q}\right|\frac{2}{\varepsilon^2}\log\left(\frac{2}{\delta}\right)
\end{align}
measurements for each $\hat P\in\mathcal P_n$ we can calculate the estimator $\tilde f\left(\bm{x}\right)$ of $f\left(\bm{x}\right)$ within an accuracy $\varepsilon$ with at least a probability $1-\delta$ as
\begin{align}
    \tilde f\left(\bm{x}\right)
    =\sum_{\hat P\in\mathcal P_n}c_{\hat P}\tilde{\mathbb{E}}_{\hat{P}}=\sum_{\hat P\in\mathcal P_n}c_{\hat P}\left(N_{\hat P}^+-N_{\hat P}^-\right)/N_{\hat P}.
\end{align}

\section{Bounds on the Maximal METs for one and two-qubit systems}
\label{app: gate-bounds}

In this appendix, we bound the maximal MET for state preparation of one- and two-qubit systems. We utilize these bounds in the body of the paper as a comparison and to ensure that we have scanned to a large enough total evolution time $T$ to obtain the entire cumulative probability distribution $\mathbb P\left(\textrm{MET}\le T\right)$.

In deriving these bounds, we assume that within the gate-based model, the quantum computer operates on a clock: parameterized gates have an associated execution time that is independent of the parameter value. If the execution time for a parameterized gate naturally depends on the parameter, we can just fix the execution time as the maximal execution time over the parameter space. For all naturally slower gates, we simply apply the identity operation for the remainder of the execution time.

\subsection{One-Qubit-State Preparation Maximal MET}

Any unitary on one-qubit can be comprised of three Euler rotations, for example:
\begin{align}
  U\left(\alpha,\beta,\gamma\right)\equiv R_X\left(\alpha\right)R_Y\left(\beta\right)R_X\left(\gamma\right).
\end{align}
The rotation is characterized by three proper Euler angles: $\alpha,\gamma\in\left(-\pi,\pi\right]$ and $\beta\in\left[-\pi/2,\pi/2\right]$. The fastest single-qubit gate recorded in silicon electron Loss-DiVincenzo qubits took approximately \SI{50}{\ns} and corresponds to a $\pi$ rotation about a single axis \cite{Corna2018, Stano2022}. In this article, we multiplied this single qubit gate time by four so that a four-qubit system would not implement a faster single qubit gate by aligning the four drive frequencies to the Rabi frequency of a single qubit. That is we take the fastest single-qubit gate to take \SI{200}{\ns}. Thus, at most, a single-qubit gate implemented using the three Euler rotations within our system will take $2.5\times\SI{200}{\ns}=\SI{500}{\ns}$. Further, an arbitrary single-qubit gate must take at least as long as the time for the fastest $\pi$ rotation: \SI{200}{\ns}.

\subsection{Two-Qubit-State Preparation Maximal MET}\label{app: two-qubit-gates}

\citeauthor{CNOTCount} prove in Ref.~\cite{CNOTCount} that a single CNOT and four single-qubit gates are sufficient to transition between arbitrary two-qubit states. Here, we follow a similar proof but replace the CNOT with a power of SWAP. The exchange operator in \cref{eq: exchange operator} is the native two-qubit operation in the rotating frame considered in this article. However, we can use the single-qubit gates in the procedure below to move back into the lab frame so that the native two-qubit operation is once again the power of SWAP.

\begin{lemma}
  Four single-qubit gates and $\textrm{SWAP}^\alpha$  with $\alpha\in\left[0, 1\right]$ are sufficient to transition between arbitrary two-qubit states.
\end{lemma}
\begin{proof}
  We can express our initial state $\ket{\psi_0}$ as a Schmidt decomposition:
  \begin{align}
    \ket{\psi_0}=\cos\left(\theta\right)\ket{a_0}\ket{b_0}+\sin\left(\theta\right)\ket{a_1}\ket{b_1},
  \end{align}
  where $\left\{\ket{a_0}, \ket{a_1}\right\}$ and $\left\{\ket{b_0}, \ket{b_1}\right\}$ are orthonormal bases of each qubit. By applying a single-qubit gate to each qubit, we can induce the following map
  \begin{align}\label{eq: canonical two-qubit form}
    \ket{\psi_0}\mapsto\ket{\psi'}\coloneqq\cos\left(\theta\right)\ket{0}\ket{1}+\sin\left(\theta\right)\ket{1}\ket{0}.
  \end{align}
  Now, applying a power of SWAP gate, we will obtain
  \begin{align}
    \ket{\psi'}\mapsto\ket{\phi'}\coloneqq\cos\left(\theta'\right)\ket{0}\ket{1}+\sin\left(\theta'\right)\ket{1}\ket{0}.
  \end{align}
  Finally, applying a single-qubit gate to each qubit again, we can obtain an arbitrary state
  \begin{align}
    \ket{\phi'}\mapsto\cos\left(\theta'\right)\ket{a'_0}\ket{b'_0}+\sin\left(\theta'\right)\ket{a'_1}\ket{b'_1}.
  \end{align}
\end{proof}
In this article we limit $J$ to \SI{1}{\GHz}, thus, $\textrm{SWAP}^\alpha$ for $\alpha\in\left[0, 1\right]$ takes less than \SI{0.5}{ns} to impliment. Thus, the transition between arbitrary two-qubit states can always be performed in under \SI{1000.5}{ns} by applying the pairs of single-qubit gates in parallel.

When preparing the ground states of \ce{H2} and \ce{HeH+} in \cref{sec: molecular} we start in the $\ket{01}$ state. Thus, for these preparations, we can start from the state $\ket{\psi'}$ in \cref{eq: canonical two-qubit form} and forgo the first layer of arbitrary single-qubit rotations. This reduces the upper bound on the gate-based two-qubit molecular ground-state preparation times to \SI{500.5}{\ns}.

To get the lower bound, we note we must use at least two single qubit gates in parallel and a two-qubit gate to prepare an arbitrary two-qubit state. Thus, the arbitrary two-qubit state preparation time is lower bounded by \SI{200.5}{\ns}.

\section{Homogeneous and Isotropic Speed Limit}
\label{app: homogeneous and isotropic speed limit}

In this appendix, we derive the cumulative probability distribution for the MET between Haar random input and output states in a $d$-dimensional Hilbert space with a homogeneous and isotropic quantum speed limit. More precisely, a Haar distributed pure state $\ket{\varphi}\sim\mu$ is a pure state generated by multiplying a Haar random unitary by a fixed pure state.

The geodesic distance between two pure states $\ket{\psi_0}$ and $\ket{\phi}$ under the Fubini-Study metric is
\begin{align}
  d\left(\ket{\psi_0},\ket{\phi}\right)=2\cos^{-1}\left|\braket{\psi_0}{\phi}\right|.
\end{align}
If we take the homogeneous and isotropic quantum speed limit to be $v$, then the MET is
\begin{align}
  \textrm{MET}\left(\ket{\psi_0},\ket{\phi}\right)=\frac{2}{v}\cos^{-1}\left|\braket{\psi_0}{\phi}\right|.
\end{align}

For notational ease, let
\begin{align}
    f_{X\sim D}\left[g\left(X\right)\right]\left(x\right)\dd{y}=\mathbb P_{X\sim D}\left(g\left(x\right)\le g(X)< g\left(x\right)+\dd{y}\right)
\end{align}
be the probability density function of $g\left(X\right)$ where $X\sim D$. Note that the invariance of the Haar distribution under unitary transformations allows us to fix one of the states
\begin{align}
  f_{\ket{\psi_0},\ket{\phi}\sim\mu}\left[\left|\braket{\psi_0}{\phi}\right|\right]=f_{\ket{\phi}\sim\mu}\left[\left|\braket{\psi_0}{\phi}\right|\right].
\end{align}
Without loss of generality, let $\bra{\psi_0}=\begin{bmatrix}1&0&\cdots&0\end{bmatrix}$ and fix the global phase of $\ket{\phi}$ such that the first entry is real. Thus, $f_{\ket{\phi}\sim\mu}\left[\left|\braket{\psi_0}{\phi}\right|\right]$ is equal to the probability density function of the magnitude of the first component of a $2d-1$ dimensional real vector that is distributed uniformly over the $2d-1$ real unit sphere. Therefore, we find \cite{43fd0a6c-7343-3149-83c3-841692a8522d}:
\begin{align}
  f_{\ket{\psi_0},\ket{\phi}\sim\mu}\left[\left|\braket{\psi_0}{\phi}\right|\right]&=\frac{2\left(1-\left|\braket{\psi_0}{\phi}\right|^2\right)^{d-2}}{B\left(d-1,\frac{1}{2}\right)}\label{eq: initial dist}\\
  &=\frac{2\left(1-\cos^2\left[\frac{v}{2}\textrm{MET}\left(\ket{\psi_0},\ket{\phi}\right)\right]\right)^{d-2}}{B\left(d-1,\frac{1}{2}\right)}\\
  &=\frac{2\sin^{2d-4}\frac{v}{2}\textrm{MET}\left(\ket{\psi_0},\ket{\phi}\right)}{B\left(d-1,\frac{1}{2}\right)},
\end{align}
where $B\left(\bullet, \bullet\right)$ is the beta function.

Now we can use the identity $f_{X\sim D}\left[g\left(X\right)\right]\left(y\right)\equiv f_{X\sim D}\left[X\right]\left(g\left(y\right)\right)\left|\dv{g}{y}\right|^{-1}$ to show
\begin{multline}\label{eq: probability density function}
  f_{\ket{\psi_0},\ket{\phi}\sim\mu}\left[\textrm{MET}\left(\ket{\psi_0},\ket{\phi}\right)\right]\left(T\right)\\=\frac{v}{B\left(d-1,\frac{1}{2}\right)}\sin^{2d-3}\left(\frac{vT}{2}\right).
\end{multline}

Integrating \cref{eq: probability density function} and performing a change of variables gives our final result
\begin{align}\label{eq: app iso eq}
  \mathbb P\left(\textrm{MET}\le T\right)=\frac{2}{B\left(d-1,\frac{1}{2}\right)}\int\limits_0^{\frac{vT}{2}}\sin^{2d-3}\left(x\right)\dd{x}.
\end{align}
While this integral can be performed analytically for integer $d$, it is generally cumbersome. We utilize the analytic form of the integral to calculate the curves in \cref{fig: MET distribution}.

\section{Inhomogeneous and Anisotropic Speed Limit}
\label{app: inhomogeneous and anisotropic speed limit}

In this appendix, we will first use the homogeneous and isotropic result to constrain the functional form of the inhomogeneous and anisotropic distribution in \cref{sec: expansion}. Second, we will provide a possible explanation for the vanishing derivatives of $\mathbb P\left(\textrm{MET}\le T\right)$ at the maximal MET in \cref{sec: derivatives}. Finally, we will combine the two results in \cref{sec: numerical fits}.

\subsection{Series Expansion}\label{sec: expansion}

First, note that the probability density function only has support on the finite interval
\begin{align}
    \left[0,\max_{\ket{\psi_0},\ket{\phi}}\textrm{MET}\left(\ket{\psi_0},\ket{\phi}\right)\right].\\
  \nonumber
\end{align}
Thus, define the quantity
\begin{equation}
    \tilde v\coloneqq\frac{\pi}{\displaystyle\max_{\ket{\psi_0}, \ket{\phi}}\textrm{MET}\left(\ket{\psi_0}, \ket{\phi}\right)},
\end{equation}
which will play an analogous role to the speed limit $v$. Using this definition, we Maclaurin-expand the function
\begin{align}
    g&\coloneqq f_{\ket{\psi_0},\ket{\phi}\sim\mu}\left[\textrm{MET}\left(\ket{\psi_0},\ket{\phi}\right)\right]\circ\frac{2}{\tilde v}\sin^{-1}\\
    g\left(x\right)&\equiv\sum_{n=0}^\infty c_nx^n\quad x\in\left[0,\varepsilon\right),
\end{align}
for some radius of convergence $\varepsilon$. As $\sin^{-1}$ is an analytic function then the power series
\begin{equation}\label{eq: sin expansion}
    f_{\ket{\psi_0},\ket{\phi}\sim\mu}\left[\textrm{MET}\left(\ket{\psi_0},\ket{\phi}\right)\right]\left(T\right)\equiv \sum_{n=0}^\infty c_n\sin^n\left(\frac{\tilde vT}{2}\right),
\end{equation}
converges on the interval $\left[0,\frac{2}{\tilde v}\sin^{-1}\varepsilon\right)$. Further, as the Maclaurin series converges uniformly, so too does the series in \cref{eq: sin expansion}. Therefore, we can integrate $f_{\ket{\psi_0},\ket{\phi}\sim\mu}\left[\textrm{MET}\left(\ket{\psi_0},\ket{\phi}\right)\right]\left(T\right)$ by integrating each term in \cref{eq: sin expansion}:
\begin{align}
  \mathbb P\left(\textrm{MET}\le T\right)&\equiv\int\limits_0^{\frac{\tilde vT}{2}}\dd{x}\sum_{n=0}^\infty c_n\sin^n\left(x\right),\\
  &=\sum_{n=0}^\infty c_n\int\limits_0^{\frac{\tilde vT}{2}}\dd{x}\sin^n\left(x\right).
\end{align}
We will use this expression to find coefficients that must vanish.

First, note that
\begin{equation}
    \int\limits_0^{\delta}\dd{x}\sin^n\left(kx\right)=\bigO{\delta^{n+1}},
\end{equation}
as $\delta\to 0$. Thus, the first non-zero $c_n$ must be positive as $\mathbb P\left(\textrm{MET}\le T\right)$ is positive. Additionally, $\mathbb P\left(\textrm{MET}\le T\right)$ must be upper bounded by the corresponding cumulative probability distribution for a $d$-dimensional Hilbert space with a homogeneous and isotropic speed limit given by the maximal speed limit of the inhomogeneous and anisotropic case. Thus, the first non-zero $c_n$ must be for $n\ge 2d-3$. Therefore we find
\begin{equation}
  \mathbb P\left(\textrm{MET}\le T\right)\equiv\int\limits_0^{\frac{\tilde vT}{2}}\dd{x}\sum_{\mathclap{n=2d-3}}^\infty c_n\sin^n\left(x\right),
\end{equation}
where the first non-zero $c_{n}$ is positive. Only the first term in the expansion remains for the homogeneous and isotropic case. Generally, the terms for $n>2d-3$ are corrections due to the inhomogeneities and anisotropies.

\subsection{Vanishing Derivatives at Maximal MET}\label{sec: derivatives}

In this section, we provide a possible explanation for the vanishing derivatives of $\mathbb P\left(\textrm{MET}\le T\right)$ observed in \cref{fig: MET distribution}. First, we define $T_\star$ as the minimum $T$ that satisfies $\mathbb P\left(\textrm{MET}\le T\right)=1$. Assuming $\mathbb P\left(\textrm{MET}\le T\right)$ is Taylor expandable about $T=T_\star$ with some finite radius of convergence, then for the first $k$ derivatives of $\mathbb P\left(\textrm{MET}\le T\right)$ to vanish at $T=T_\star$ we require $1-\mathbb P\left(\textrm{MET}\le T\right)=\mathbb P\left(\textrm{MET}\ge T\right)=\bigO{[T-T_\star]^{k+1}}$ at $T=T_\star$. Thus, for any $\tau>0$ we find $\mathbb P\left(\textrm{MET}\ge T_\star-\tau\right)>0$.

Next, we note that $\mathbb P\left(\textrm{MET}\ge T\right)$ is a normalized volume in the manifold of state pairs $\mathcal M$:
\begin{align}
    \mathbb P\left(\textrm{MET}\ge T\right)&=\int\limits_{\mathclap{\textrm{MET}\left(\ket{\psi_0},\ket{\phi}\right)\ge T}}\dd{\ket{\psi_0}}\dd{\ket{\phi}}.
\end{align}
Note that for the homogeneous and isotropic case, the volume $\mathbb P\left(\textrm{MET}\ge T\right)$ is linear in $(T-T_\star)$ at $T=T_\star$. We now give one possible mechanism in the inhomogeneous and anisotropic case for the volume to scale as $\mathbb P\left(\textrm{MET}\ge T\right)=\bigO{[T-T_\star]^{k+1}}$.

Consider the subset of states pairs $\mathcal S\subset\mathcal M$ such that $\textrm{MET}\left(\mathcal S\right)=T_\star$ and let us assume that $\mathcal S$ is closed and the normal space well defined at all points in $\mathcal S$. If the velocity one can move at through the manifold of quantum states is limited to be finite at all points and in all directions, then $\mathcal S$ must expand from all points a distance $\bigO{[T-T_\star]^\alpha}$ with $\alpha\ge 1$ in each direction in the normal space to that point. Thus, the volume $\mathbb P\left(\textrm{MET}\ge T\right)=\bigO{[T-T_\star]^{\alpha D}}$ where $D$ is the maximal dimension of the normal spaces to $\mathcal S$. Therefore, if at least the first $D-1$ derivatives of $\mathbb P\left(\textrm{MET}\le T\right)$ will vanish.

The maximal dimension of the normal space $D$ can be increased by breaking symmetries. To gain some intuition, consider the following broken symmetry: Suppose not all input states $\ket{\psi_0}$ have a corresponding output state $\ket{\phi}$ that maximizes the MET. If the set of $\ket{\psi_0}$ that do have at least one $\ket{\phi}$ that maximizes the MET are measured zero in the manifold of pure states, then $D>1$. Therefore, at least the first derivative of $\mathbb P\left(\textrm{MET}\le T\right)$ vanishes.

\subsection{Numerical Fitting}\label{sec: numerical fits}

In numerical fitting, we have to truncate our series expansion
\begin{equation}\label{eq: truncation}
  \mathbb P\left(\textrm{MET}\le T\right)\equiv\int\limits_0^{\frac{\tilde vT}{2}}\dd{x}\sum_{\mathclap{n=2d-3}}^L c_n\sin^n\left(x\right).
\end{equation}
If we want to fit \cref{eq: truncation} to the whole curve of $\mathbb P\left(\textrm{MET}\le T\right)$ we are implicitly assuming the convergence radius $\varepsilon\ge \frac{\pi}{2}$. Thus, if we know that the first $k$ derivatives at $T_\star$ must vanish by symmetry arguments, we can reduce the number of free parameters in the fit. In our work, we will assume the first derivative vanishes, which gives the following constraint:
\begin{equation}\label{eq: constraint}
    c_L=-\sum_{\mathclap{n=2d-3}}^{L-1}c_n.
\end{equation}

\section{Bootstrapping}\label{app: bootstrapping}

We apply bootstrapping to quantify the error in the cumulative distributions, which we estimate by sampling pairs of states from the Haar distribution. Our estimated cumulative distribution consists of 1024 input-output state pairs. We resample, with replacement, from this sample, 1024 input-output state pairs and use this resample to compute a new cumulative distribution. We repeat this $10^5$ times to get $10^5$ resampled cumulative distributions. For each value of $T$, we calculate the confidence interval and variance of the resampled cumulative probabilities. We use the variances during the chi-squared fits in \cref{fig: MET distribution} and \cref{app: chis}. When the variance of the resampled cumulative probabilities is less than $1/1024^2$, we take $1/1024^2$ as the variance, as this is the square of the precision of our estimate.

\section{Chi-Squared Values}
\label{app: chis}

In this appendix, we perform chi-squared fits using the variances from \cref{app: bootstrapping} to the data in \cref{fig: MET distribution}. First, we fit using the homogeneous and isotropic approximation \cref{eq: app iso eq} and a normalization constraint. Second, we perform a higher order fit using the truncation in \cref{eq: truncation}, under the constraints of normalization and \cref{eq: constraint}. In the higher order fit, we include the minimal number of terms such that chi-squared $\chi^2$ divided by the number of degrees of freedom $N_{\textrm{DoF}}$ is less than unity ($\chi^2/N_{\textrm{DoF}}<1$) or $\chi^2/N_{\textrm{DoF}}$ stops decreasing monotonically with the addition of further terms. We present the values for $\chi^2/N_{\textrm{DoF}}$ for the higher order fit in \cref{table: chis}.

\begin{table}[H]
  \begin{tabularx}{\columnwidth}{lXrXrXr}
    \toprule
    &&&&Number of&&Terms in\\
    Qubits && $\chi^2/N_{\textrm{DoF}}$ && Parameters &&  Expansion \\\colrule
    1&&0.62&&6&&7\\
    2&&0.77&&4&&5\\
    3&&1.16&&5&&6\\
    4&&7.53&&7&&8\\
    \botrule
  \end{tabularx}
  \caption{\textbf{$\bm{\chi^2/N_{\textrm{DoF}}}$ values for higher order fits in \cref{fig: MET distribution}.} We also present the number of parameters and terms in \cref{eq: truncation} used for the fit under the constraints of normalization and \cref{eq: constraint}.}
  \label{table: chis}
\end{table}

\section{Exchange in the Rotating Frame}
\label{app: ex}
In this appendix, we analyze the unitary evolution induced by the exchange Hamiltonian in the rotating frame of the drift Hamiltonian:
\begin{equation}
  \hat{H}_J = -\sum_{i=1}^{N-1} \frac{J_i(t)}{4} \hat {\mathcal E}_{i, i+1}\left(t\right).
\end{equation}
On the subspace of two qubits, encoded by the spins of the $i$th and $j$th electron, its operator in matrix form is
\begin{equation}\label{eq: exchange operator}
  \!\!\!\hat {\mathcal E}_{i,j}\!\left(t\right)\!\coloneqq\!\!\begin{bNiceMatrix}[last-col]
  1&&&&\ket{\uparrow_i,\uparrow_j}\\
  &-1&\!2e^{i\left(B_j-\!B_i\right)t\!/\!\hbar}&&\ket{\uparrow_i,\downarrow_j}\\
  &2e^{i\left(B_i-\!B_j\right)t\!/\!\hbar}&-1&&\ket{\downarrow_i,\uparrow_j}\!\\
  &&&1&\ket{\downarrow_i,\downarrow_j}
\end{bNiceMatrix}.\!\!\!
\end{equation}
For a constant $J_i\left(t\right)=J$ the unitary evolution from a time $t_0$ to $t_0+\Delta t$ in the rotating frame induced by the exchange Hamiltonian is 
\begin{align}
  \hat U_J(t_0\to t_0+\Delta t)=
  \begin{bmatrix}
    U_{0,0} &         &         &         \\ 
            & U_{1,1} & U_{1,2} &         \\
            & U_{2,1} & U_{2,2} &         \\
            &         &         & U_{3,3} \\
  \end{bmatrix}.
\end{align}
Using $\Jeff \coloneqq\sqrt{J^2+\Delta B^2}$ its non-zero entries are
\begin{align}
    U_{0,0} &\!=\! U_{3,3} \!=\! \exp(i\frac{J\Delta t}{4\hbar}), \\
    U_{1,1} &\!=\!\!
    \left[
    \!\cos(\!\frac{\Jeff\Delta t}{2\hbar}\!)
    \!\!-\!i \frac{\Delta B}{\Jeff}\!\sin(\!\frac{\Jeff\Delta t}{2\hbar}\!)
    \!\right] \!\!
    e^{-i(\frac{J}{4}-\frac{\Delta B}{2})\frac{\Delta t}{\hbar}},
    \\
    U_{2,2} &\!=\!\!
    \left[\!
    \cos(\!\frac{\Jeff \Delta t}{2\hbar}\!)
    \!\!+\!i \frac{\Delta B }{\Jeff}\!\sin(\!\frac{\Jeff\Delta t}{2\hbar}\!)
    \!\right] \!\!
    e^{-i(\frac{J}{4}+\frac{\Delta B}{2})\frac{\Delta t}{\hbar}},
    \\
    U_{1,2} &\!=\!
    i\frac{J}{\Jeff} \sin(\frac{\Jeff\Delta t}{2\hbar})
    e^{i\Delta B t_0/\hbar} e^{-i(\frac{J}{4}-\frac{\Delta B}{2})\frac{\Delta t}{\hbar}}, \\
    U_{2,1} &\!=\!
    i\frac{J}{\Jeff} \sin(\frac{\Jeff \Delta t}{2\hbar})
    e^{-i\Delta B t_0/\hbar} e^{ -i(\frac{J}{4}+\frac{\Delta B}{2})\frac{\Delta t}{\hbar}}.
\end{align}

In the limit $\Delta B \Delta t/\hbar\to 0$ (the relevant limit for the molecular ground-state preparation METs) we find
\begin{equation}
    \left|\Delta B\sin(\frac{\tilde J\Delta t}{2\hbar})\right|\le\left|\Delta B\frac{\tilde J\Delta t}{2\hbar}\right|\to 0,
\end{equation}
and
\begin{align}
    &\left|\left(1-\frac{J}{\tilde J}\right)\sin(\frac{\tilde J\Delta t}{2\hbar})\right|\\&\le\left|\left(\tilde J-J\right)\frac{\Delta t}{2\hbar}\right|\\&=\left|\left(\sqrt{J^2\Delta t^2+\Delta B^2\Delta t^2}-J\Delta t\right)\frac{1}{2\hbar}\right|\\&\to 0.
\end{align}
Thus,
\begin{align}
    U_{1,1} &\approx
    \cos(\frac{\Jeff\Delta t}{2\hbar})
    e^{-i\frac{J}{4}\frac{\Delta t}{\hbar}},
    \\
    U_{2,2} &\approx
    \cos(\frac{\Jeff \Delta t}{2\hbar})
    e^{-i\frac{J}{4}\frac{\Delta t}{\hbar}},
    \\
    U_{1,2} &\approx
    i\sin(\frac{\Jeff\Delta t}{2\hbar})
    \exp(i\frac{\Delta B t_0}{\hbar}) e^{-i\frac{J}{4}\frac{\Delta t}{\hbar}}, \\
    U_{2,1} &\approx
    i\sin(\frac{\Jeff \Delta t}{2\hbar})
    \exp(-i\frac{\Delta B t_0}{\hbar}) e^{ -i\frac{J}{4}\frac{\Delta t}{\hbar}}.
\end{align}
Finally, using $J\Delta t/\hbar\approx\tilde J\Delta t/\hbar$ in the limit $\Delta B\Delta t/\hbar\to 0$, we find $\hat U_J(t_0\to t_0+\Delta t)$ is a \textit{power-of-SWAP like} gate:
\begin{multline}\label{eq: power swap}
  \hat U_J(t_0\to t_0+\Delta t)\approx\\\exp(i\frac{\tilde J\Delta t}{4\hbar})\begin{bmatrix}
        1\\
        &0&e^{i\Delta B t_0/\hbar}\\
        &e^{-i\Delta B t_0/\hbar}&0\\
        &&&1
    \end{bmatrix}^{\dfrac{\tilde J\Delta t}{\pi\hbar}},
\end{multline}
which is (up to a global phase) a rotation about the $\cos(\Delta B t_0/\hbar)X+\sin(\Delta Bt_0/\hbar)Y$ axis of the Bloch sphere of the $\ket{\tilde 0}\equiv\ket{01}$, $\ket{\tilde 1}\equiv\ket{10}$ subspace where $\left\{\ket{\tilde 0},\ket{\tilde 1}\right\}$ is the eigenbasis of the $Z$ operator on the subspace.

\end{document}